\documentclass{statsoc}

\usepackage[a4paper]{geometry}
\usepackage{amsmath}
\usepackage{graphicx}
\usepackage{enumerate}

\usepackage{multirow}
\usepackage{epstopdf}
\usepackage{epsfig}
\usepackage{url}
\usepackage[toc,page]{appendix}
\usepackage{float}
\usepackage{natbib}
\usepackage{color}

\usepackage{hyperref}
\hypersetup{
    colorlinks=true,
    linkcolor=blue,
    filecolor=magenta,      
    urlcolor=cyan,
    citecolor=blue
}
 
\usepackage{cleveref}
\usepackage{bm}
\usepackage{amsfonts}
\usepackage{algorithm}
\usepackage{algorithmic}
\usepackage{makecell}
\usepackage{multirow}
\usepackage{subfigure}

\usepackage{enumitem}
\usepackage{pifont}
\newcommand{\cmark}{\textcolor{black}{\ding{51}}}
\newcommand{\xmark}{\textcolor{black}{\ding{55}}}

\newcommand{\e}{\mathbb{E}}
\newcommand{\p}{\mathbb{P}}

\newcommand{\ccal}{\mathcal{C}}
\newcommand{\pcal}{\mathcal{P}}

\newcommand{\rcal}{\mathcal{R}}

\newcommand{\ncal}{\mathcal{N}}

\newcommand{\col}{\mathrm{col}}

\newcommand{\lb}{\left(}
\newcommand{\rb}{\right)}
\newcommand{\pare}[1]{\lb{#1}\rb}

\newcommand{\be}{\begin{equation}}
\newcommand{\ee}{\end{equation}}
\newcommand{\bes}{\begin{equation*}}
\newcommand{\ees}{\end{equation*}}
\newcommand{\ba}{\begin{align}}
\newcommand{\ea}{\end{align}}
\newcommand{\bas}{\begin{align*}}
\newcommand{\eas}{\end{align*}}

\newcommand{\bR}{\mathbb{R}}
\newcommand{\cov}{\mathrm{Cov}}

\newcommand{\var}{\mathrm{Var}}

\newcommand{\mbone}{{\mathbf{1}}}

\newcommand{\diag}{\text{diag}}
\newcommand{\Span}{\text{span}}

\usepackage{chngcntr}
 
\counterwithin{figure}{section}
\counterwithin{table}{section}

\newtheorem{defi}{Definition}
\newtheorem{thm}{Theorem}
\newtheorem{prop}{Proposition}
\newtheorem{lem}{Lemma}
\newtheorem{rem}{Remark}
\newtheorem{ass}{Assumption}
\newtheorem{coro}{Corollary}
\newtheorem{algo}{Algorithm}

\def\Var{\textsf{Var}} 

\title[Linear regression and its inference on noisy network-linked data]{ Linear regression and its inference on noisy network-linked data}
\author[Le and Li]{ Can M. Le \thanks{The two authors contributed equally to this paper.}}
\address{University of California, Davis,
USA.}
\email{canle@ucdavis.edu}
\author[Le and Li]{ Tianxi Li }
\address{University of Virginia,
Charlottesville,
USA.}
\email{ tianxili@virginia.edu }

\begin{document}

\newcommand{\norm}[1]{\left\lVert#1\right\rVert}
\newcommand{\abs}[1]{\left|#1\right|}

\begin{abstract}
Linear regression on network-linked observations has been an essential tool in modeling the relationship between response and covariates with additional network structures. Previous methods either lack inference tools or rely on restrictive assumptions on social effects and usually assume that networks are observed without errors. This paper proposes a regression model with nonparametric network effects. The model does not assume that the relational data or network structure is exactly observed and can be provably robust to network perturbations. Asymptotic inference framework is established under a general requirement of the network observational errors, and the robustness of this method is studied in the specific setting when the errors come from random network models. We discover a phase-transition phenomenon of the inference validity concerning the network density when no prior knowledge of the network model is available while also showing a significant improvement achieved by knowing the network model. Simulation studies are conducted to verify these theoretical results and demonstrate the advantage of the proposed method over existing work in terms of accuracy and computational efficiency under different data-generating models. The method is then applied to middle school students' network data to study the effectiveness of educational workshops in reducing school conflicts.
\end{abstract}

 \keywords{Network Modeling;  Network-linked Data; Linear Regression; Random Networks; Network Perturbation}

\section{Introduction}\label{sec:intro}

Nowadays, networks appear frequently in many areas, including social sciences, transportation, and biology. In most cases, networks are used to represent relationships or interactions between units of a complex (social, physical, or biological) system, so analyzing network data may render crucial insights into the dynamics and/or interaction mechanism of the system. One particular, yet commonly encountered, situation is when a group of units is observed connected by a network and a set of attributes for each of these units is available.  Such data sets are sometimes called network-linked data \citep{li2016prediction, li2019high} or multiview network data \citep{gao2019testing}. Network-linked data are widely available in almost all fields involving network analysis about social effects \citep{michell1996peer,pearson2003drifting}, collaborations \citep{ji2016coauthorship,su2019testing}, and causal experiments \citep{basse2018limitations,basse2018model}. In network-linked data, rich information is available from the perspective of both the individual attributes and the network, and the challenge is to find proper statistical methods to incorporate both. Suppose one single network is available. Consider the situation where, for each node $i$ of the network, we observe $({x}_i, y_i)$, in which ${x}_i \in \bR^p$ is a vector of covariates while $y_i \in \bR$ is a scalar response. In particular, we aim for a regression model of $y_i$ against ${x}_i$ that also takes the network information into account. Such a model arises naturally in any problem when a prediction model or inference of a specific attribute is of interest.

Although the systematic study of regression on network-linked data has only recently begun to attract interest in statistics \citep{li2016prediction,zhu2017network,su2019testing}, it has been studied in econometrics by many authors, who mostly focused on multiple networks \citep{manski1993identification,lee2007identification,bramoulle2009identification} or longitudinal data \citep{jackson2007relating, manresa2013estimating}.  Social effects are typically observed in the form that connected units share similar behaviors or properties. The similarity or correlation may be due to either \emph{homophily}, where social connections are established because of similarity, or \emph{contagion}, where individuals become similar through the influence of their social ties. In general, one cannot distinguish homophily from contagion in a single snapshot of observational data \citep{shalizi2011homophily}, as in our setting. Therefore, the two directions of causality will not be distinguished,  and this type of generic similarity between connected nodes will be called ``network cohesion", as in \cite{li2016prediction, li2019high}. 
A significant class of models for this type of network regression problems is the class of autoregressive models, which is based on ideas in spatial statistics.   The spatial autoregressive models have been widely used in econometrics, such as in \cite{manski1993identification,lee2007identification,bramoulle2009identification,hsieh2016social,zhu2017network}, to name a few.  A common form for such a model is 
$$y_i = \gamma \Big(\frac{1}{n_i}\sum_{i'\sim i}y_{i'}\Big) + X\beta + \frac{1}{n_i}\Big(\sum_{i'\sim i}x_{i'}\Big)^T\eta + \epsilon,$$
where $y_i$ is the response, $x_i$ is the covariate vector and $n_i$ is the number of neighbors of $i$ (the notation $i\sim i'$ indicates that $i$ and $i'$ are connected). In this model, the neighborhood averages of the response  and covariates are  used to model the social effect. The parameters $\gamma$ and $\eta$ are typically called ``endogenous" and ``exogenous" effects, respectively, while $\beta$ is the standard regression slope. This class of models is also called the ``linear in means" model or social interaction model (SIM).  When multiple networks for different populations are available, the intercept can be treated as a group effect, called  the``external effect" \citep{manski1993identification,lee2007identification,lee2010specification}.  The SIM framework can also combined with network formulation models to study the correlation, homophily and contagion effects, if rather than a single network, one could observe multiple network snapshots over time \citep{goldsmith2013social, mcfowland2021estimating}.    The SIM framework, though it has been a popular setup for regression on network-linked data,  suffers from two crucial drawbacks. The first comes from its restrictive parametric form of the social effect. Assuming social effect in the form of autoregressive neighborhood average (or summation)  is far from realistic, and this stringent assumption significantly limits the usefulness of the framework. As shown in \cite{li2016prediction} and also in our empirical study, the restrictive assumption of social effect leads to poor prediction performance. The second limitation  comes from the  assumption that the network structure is precisely observed from the data. In practice, it is well known that most network data are subject to errors, due to missingness \citep{lakhina2003sampling,butts2003network,
clauset2005accuracy}, observational errors \citep{handcock2010modeling,
rolland2014proteome,le2017estimating,
khabbazian2017novel,newman2018estimating,
lunagomez2018evaluating}, or data collection method \citep{wu2018link, rohe2019critical}. If the imprecise network is used in SIM, the model is ill-defined and the inference would also be problematic, as shown by \cite{chandrasekhar2011econometrics}.

In this paper, a new regression model to fit the network-linked data is proposed that addresses both drawbacks mentioned above. The new model is based on a flexible network effect assumption and is robust to network observational errors. It also allows us to perform tests and construct confidence intervals for the model parameters. 
Our theoretical analysis provides the support for model estimation and inference and quantifies the magnitude of network observational errors under which the inference remains robust. Moreover, in the random network perturbation scenario, the tradeoff between the available information of the network structure and the level of robustness of the statistical inference is characterized. In its most difficult setting, when no prior knowledge about the network model is available, the result reveals a phase-transition phenomenon at the network average degree of $\sqrt{n}$, above which the inference is asymptotically correct, and below which the inference becomes invalid. The inference  could remain valid for much sparser networks if more information about the network structure is available. For example, when the network model is known and an effective parametric estimation can be applied, the sparsity requirement can be relaxed to $\log{n}$.  To the best of  our knowledge, this paper is the first work that addresses the inference of the network-linked regression models and accounts for the network observational errors.

Related to the network-linked model proposed in this paper is the semi-parametric model called ``regression with network cohesion" (RNC), introduced by \cite{li2016prediction}. In their model, the network effect is represented by individual parameters that are assumed to be ``smooth" over the network, and a similar idea is used in a few other statistical estimation settings \citep{wang2014trend,zhao2016significance,fan2018approximate}.  Despite its flexibility of social effect assumption and excellent  predictive performance, the RNC model lacks a valid inference framework and cannot be applied in many modern applications where statistical inference is needed \citep{ogburn2018challenges, su2019testing}. The result of \cite{li2019high} indicates that the RNC estimator fails to guarantee valid inference under reasonable assumptions unless additional assumptions such as sparsity are made \citep{zhao2016significance}. Moreover, little is known about the robustness of the RNC method to network observational errors, although preliminary results have been obtained in a particular scenario of network sparsification \citep{sadhanala2016graph,li2016prediction}. In contrast, our model does not assume the smoothness of network effects as in the RNC method. Instead, we use a general relational subspace to define the social effects. As can be seen later, the proposed model overcomes both of the two aforementioned limitations of RNC and is computationally more efficient.

Table~\ref{tab:model-compare} gives a high-level comparison between the proposed model and two popular benchmark regression models (discussed previously) in three aspects: the flexibility of modeling social effects, the availability of an inference method, and the provable robustness to the network perturbation. The model we introduced in this paper is the only one of the three which renders all of the  desired properties.
\begin{table}
\caption{Comparison between the proposed method and two other popular methods for network-linked data.}\label{tab:model-compare}
\centering
\fbox{
\begin{tabular}{c|c|c|c}  
    \cline{1-4}
    \multicolumn{1}{c|}{Models} & social effect flexibility & inference & network robustness  \\
    \hline
    SIM \citep{manski1993identification} & \xmark & \cmark & \xmark \\
    \hline
    RNC \citep{li2016prediction} & \cmark & \xmark & \xmark \\
    \hline
    Current method & \cmark & \cmark & \cmark \\
    \hline
\end{tabular}}
\end{table}

The rest of the paper is organized as follows. Section~\ref{sec:model} introduces our model and the corresponding statistical inference algorithm. Section~\ref{secsec:generic-inference} and Section~\ref{sec:theory} are devoted to our main theoretical results; the generic theory for model estimation consistency and the asymptotic inference is given in Section~\ref{secsec:generic-inference}; then, under the random network modeling framework, detailed discussions of the technical requirement are introduced in  Section~\ref{sec:theory}. Extensive simulation experiments are given in Section~\ref{sec:sim} to verify our theory and compare our method with a few benchmark methods mentioned above. Section~\ref{sec:conflict} demonstrates the usefulness of the proposed method in analyzing a study about the effects of educational workshops on reducing school conflicts. Concluding remarks and future directions are discussed in Section~\ref{sec:conclusion}. Extended theoretical results, proofs, additional experiments and data analysis are given as appendix in the supplementary material.

\section{Network regression model}\label{sec:model}

Throughout the paper, $0_{k\times \ell}$ and $I_k$ are used to denote the zero matrix of size $k\times \ell$ and the identity matrix of size $k\times k$, respectively. The subscripts may be dropped when the dimensions are clear from the context. We use $e_i$ to denote the vector whose $i$th coordinate is 1 and the remaining coordinates are zero; the dimension of $e_i$ may vary according to the context. Given a matrix $M\in\mathbb{R}^{k\times \ell}$ and $1\le i\le j\le \ell$, $M_{i:j}\in \mathbb{R}^{k\times(j-i+1)}$ is used to denote the matrix whose columns are those of $M$ with indices $i,i+1,...,j$; when $i=j$, $M_i$ is used instead of $M_{i:i}$ for the simplicity of the notation. Throughout the paper, $\|.\|$  denotes the spectral norm,  which is the largest singular value, for matrices and the Euclidean norm for vectors. The number of nodes in the network is denoted by $n$. We say that an event $E$ occurs with high probability if $\mathbb{P}(E)\ge 1-n^{-c}$ for some constant $c>2$. We use $a_n = o(1)$ and $b_n=O(1)$ to indicate that $\lim_{n\rightarrow\infty}a_n= 0$ and $b_n$ is a bounded sequence, respectively.

\subsection{A semi-parametric regression model with network effects}

The unobserved true network is captured by the relational matrix $P \in \bR^{n\times n}$, where $P_{ij}$ describes the strength of the relation between nodes $i$ and $j$.
For each node $i$ of the network $({x}_i, y_i)$ is observed, where ${x}_i \in \bR^p$ is a vector of covariates while $y_i \in \bR$ is a scalar response. Denote by $Y=(y_1,...,y_n)^T \in \bR^{n}$ the vector of responses and by $X=(x_1,...,x_n)^T \in \bR^{n\times p}$ the design matrix, which sometimes has the first column as the all-one vector. Our goal is to construct a model that describes the dependence of $Y$ on $X$ and $P$. 

The observed network is represented by an $n\times n$ adjacency matrix $A$, where $0\le A_{ij} \le 1$ is the weight of the edge between $i$ and $j$. In the special case of an unweighted network, $A$ is a binary matrix and $A_{ij}=1$ if and only if node $i$ and node $j$ are connected. We view $A$ as a perturbed version of $P$ and focus on undirected networks, for which $A$ and $P$ are symmetric matricies. 
Moreover, we only consider the \emph{fixed design} setting, to avoid unnecessary complication of jointly modeling covariates and relational data. That means $X$ and $P$ are always treated as fixed. The adjacency matrix $A$ is also treated as fixed for the moment as the regression model and its generic inference framework are described. Later on, in only Section~\ref{sec:theory} where we study the robustness of our framework under deviation of $A$ from $P$, we will assume $A$  as a perturbation of $P$ following random network models.

Intuitively, the structural assumption is that $y_i$'s tend to be similar  for individuals having strong connections. This is called the ``assortative mixing" or ``network cohesion" property \citep{kolaczyk2009statistical, li2016prediction}. To incorporate this intuition, consider the model
\begin{equation}\label{eq:model-structure}
\e Y = X\zeta + \mu,
\end{equation}
where $\zeta \in \bR^p$ is the vector of coefficients for the covariates $X$ and $\mu \in \bR^n$ is the vector of individual effects reflecting the network cohesion property. This model was studied in \cite{li2016prediction}. Instead of  modeling the network effect by using a specific auto-regressive dependence as in \cite{manski1993identification}, \eqref{eq:model-structure} treats the network effect as a nonparametric component. As shown by \cite{li2016prediction} and the experiments later on in this paper, this approach is more flexible and gives significantly better predictive performance than the SIM framework of \cite{manski1993identification}. In \cite{li2016prediction}, $\mu$ is assumed to have small sum of squared differences $\sum_{i,j: A_{ij} = 1}(\mu_i-\mu_j)^2$. In contrast, we rely on a different form of the network cohesion that proves to be more general and stable. Specifically, the cohesion requirement of $\mu$ is formulated by assuming that 
\begin{equation}\label{eq:cohesion}
\mu \in S_K(P),
\end{equation}
where $S_K(P)$ is the subspace spanned by the $K$ leading eigenvectors of $P$. Depending on specific problems it may be possible to replace $P$ in \eqref{eq:cohesion} with other appropriate matrix-valued functions of $P$ for which the inference framework remains valid. One such example is provided in Appendix~\ref{sec:Laplacian}. The idea of using the eigenspace of $P$ to encode the cohesive pattern over the network is motivated by many previous studies. It is related to the standard spectral embedding  \citep{shi2000normalized, ng2001link, belkin2003laplacian,tang2013universally}, which maps the nodes of the network to a set of points in a Euclidean space so that the geometric relations between these points are similar to the topological relations of the nodes in the original network.  Moreover, under random network models, \cite{li2018hierarchical} and \cite{lei2020consistency} show that the eigenspaces of the expected adjacency matrix and Laplacian matrix encode varying resolutions of node similarity. Intuitively, since the spectral space combines both the network's macro-scale and micro-scale patterns, it is expected to be reasonably robust to perturbations on local connections. This gives a significant advantage compared with the SIM model. Figure~\ref{fig:school40} (Section~\ref{sec:conflict}) includes such an example. In summary, we assume the following mean structure:
\begin{defi}\label{defi:mean-structure}
The mean structure of the regression model with network effects is defined to be 
\begin{equation}\label{eq:mean-structure}
\e Y \in \Span\{\col(X),S_K(P)\},
\end{equation}
where $\col(X)$ is the column space of $X$ and $\Span$ denotes the linear subspace jointly  spanned by $\col(X)$ and $S_K(P)$. This is equivalent to \eqref{eq:model-structure} with the cohesion property \eqref{eq:cohesion}.
\end{defi}

The advantage of our mean structure \eqref{eq:model-structure} lies in its generality. In particular, one distinction between our model and the more commonly assumed linear mean structure is that our model includes the situation when the two subspaces $\col(X)$ and $S_K{P}$ have nontrivial intersection. For example, one of the covariates may be perfectly cohesive over the network, and this situation is allowed in our model. In this case, one could not uniquely determine $X\zeta$ and $\mu$. However, such a tricky situation is unavoidable for the level of generality and robustness we want to achieve. As a matter of fact, this type of ambiguity is an instance of the general conclusion from \cite{shalizi2011homophily} that contagion, and different types of  homophily effects from a single snapshot of observational data are not distinguishable. Nevertheless, the mean-structure , as well as an interpretable decomposition of the effects, can be uniquely identified with the following parameterization.

\begin{defi}\label{defi:model}
Let $\rcal=\col(X) \cap S_{K}(P)$ be the intersection of $\col(X)$ and $S_{K}(P)$. Model \eqref{eq:mean-structure} can reparametrized as
\begin{equation}\label{eq:RNC}
\e Y = X\beta + \xi + \alpha,
\end{equation}
where
\begin{eqnarray}\label{eq: identifiability assumptions}
\xi \in \rcal, \quad
X{\beta} \perp \rcal, \quad
{\alpha} \in S_{K}(P), \quad 
{\alpha}\perp \rcal,
\end{eqnarray} 
and $\beta \in \bR^p$, $\xi, \alpha \in \bR^n$. We call it the \emph{regression model with network effects}.
\end{defi}

Although there are several ways to reparameterize the mean structure \eqref{eq:mean-structure}, the one in \eqref{eq: identifiability assumptions} is preferable because it combines the two sources of information (the covariates $X$ and the relational information $P$) in a natural way so that the interpretation of the corresponding parameters is straightforward. Specifically, we observe that
\begin{itemize}
\item When $\beta=0$, we have $\e Y \in S_K(P)$ so the model can be specified as a cohesive network effects without using the covariates $X$. Therefore, $\beta$ represents the conditional covariate effects of $X$ given $P$.
\item When $\alpha=0$, we have $\e Y \in \col(X)$ so the model can be specified by the covariates without using the relational information $P$. Therefore, $\alpha$ represents the conditional network effect of $P$ given $X$.
\end{itemize}
Thus, $\beta$ and $\alpha$ reflect the conditional effects similar to the covariate effects in the standard linear regression setting. In addition, $\xi$ is the overlap of the two sources. The following result confirms the identifiability of our parameters.

\begin{prop}[Parameter identifiability]\label{prop:identifiability}
The parameterization in Definition~\ref{defi:model} is identifiable. That is, if there exist $(\beta, \alpha, \xi)$ and $(\beta', \alpha', \xi' )$ satisfying \eqref{eq:RNC} and \eqref{eq: identifiability assumptions} simultaneously then
$\beta = \beta'$, $\alpha = \alpha'$, and $\xi = \xi'$.
\end{prop}

\begin{rem} A seemingly easier way to combine the two sources of information is to treat the eigenvectors of $P$ as another set of covariates and fit $Y$ using both $X$ and the eigenvectors. However, this approach has to assume $\rcal=\{0\}$ for identifiability. This model is a special case of our model, and as can be seen later, our model estimation method would adapt to this reduced case. Our model does not enforce this assumption because it is restrictive and also causes difficulties in interpretations (see Section~\ref{sec:degenerate} in Appendix for details).  Note also that \eqref{eq: identifiability assumptions} does not require $\alpha\perp X\beta$, although this strong restriction would significantly simplify the fitting procedure and its analysis. 
\end{rem}

In the present setting, even when $p$ is much smaller than $n$, the model inference is a high-dimensional problem because of the nonparametric individual effects $\alpha$.  Throughout this paper, we assume that the noise in the regression model of Definition \ref{defi:model} is a multivariate Gaussian vector, as in many other inference methods of high-dimensional regression model \citep{van2014asymptotically,zhang2014confidence,javanmard2014confidence}.
Specifically, 
\begin{equation}\label{eq:Gaussian}
Y = \e Y + \epsilon \text{~~and~~} \epsilon \sim N(0,\sigma^2I),
\end{equation}
where $I\in\mathbb{R}^{n\times n}$ is the identity matrix and $\sigma^2$ is the variance of the noise. We leave the study of other noise distributions for future work.

\begin{rem}
As to be clear in our theory later, the reasonable interpretation of $K$ in \eqref{eq:cohesion} is the index where the spectral space of $P$ has a large eigen gap. The problem of estimating such a $K$ has been extensively studied in network literature and many efficient methods  are now available  \citep{chatterjee2015matrix, chen2014network, le2015estimating,li2016network,jin2020estimating}. 
These methods all provide theoretical guarantees for the recovery of $K$ with overwhelming probability and can be applied to our setting. Since the task of estimating $K$ neither makes our problem conceptually more challenging nor brings more insight about it, for simplicity of presentation, we assume that $K$ is known throughout the paper.
\end{rem}

\subsection{Statistical inference method}\label{secsec:model-estimation}

Our generic inference framework requires access to a certain approximation of $P$, denoted by $\hat{P}$. The specific $\hat{P}$ is determined by the user according to the understanding of the data problem. For example, one may assume that the adjacency matrix $A$ is a perturbed version of $P$. In this situation, without additional information, a natural option is to set $\hat{P} = A$ and we refer to the corresponding estimation and inference procedure as the ``model-free" version of our method, highlighting that no specific model for the network structure is assumed. Section~\ref{sec:theory} discusses some other cases for which additional network model assumptions are considered and more accurate parametric estimates of $P$ are available.

The identifiability condition \eqref{eq: identifiability assumptions} suggests a natural procedure for estimating parameters in model \eqref{eq:RNC} via subspace projections.
However, the fact that $\alpha$ and $X\beta$ need not be orthogonal complicates the estimation procedure and its analysis. 
To highlight the main idea, we first describe the population-level estimation, assuming that both $P$ and $\e Y$ are known.

Let $Z\in\mathbb{R}^{n\times p}$ be a matrix whose columns form an orthonormal basis of the covariate subspace $\col(X)$. Similarly, let $W\in\mathbb{R}^{n\times K}$ be the matrix whose columns are eigenvectors of $P$ that span the subspace $S_K(P)$. Define the singular value decomposition (SVD) of matrix $Z^TW$ to be
$Z^TW = U\Sigma V^T.$
Here, $U \in \bR^{p\times p}$ and  $V\in \bR^{K\times K}$ are orthonormal matrices of singular vectors while $\Sigma \in \bR^{p\times K}$ is the matrix with the following singular values on the main diagonal:
\begin{eqnarray}\label{eq:subspace angles}
\sigma_1 = \sigma_2 = \cdots =\sigma_r =1 >\sigma_{r+1} \ge \cdots \ge \sigma_{r+s}>0 = \sigma_{r+s+1} = \cdots = 0.
\end{eqnarray}
Thus, $r$ is the dimension of $\rcal=\text{col}(X)\cap S_K(P)$ and $Z^TW$ has $r+s$ non-zero singular values. 
Column vectors of the matrices 
\begin{equation}\label{eq:tilde Z W}
\tilde{Z} = ZU \text{~~and~~}\tilde{W} = WV
\end{equation}
also form a basis of $\col(X)$ and $S_K(P)$, respectively. In particular, the first $r$ column vectors of $\tilde{Z}$ and $\tilde{W}$ coincide and form a basis of $\rcal$. Moreover, the last $p-r$ column vectors of $\tilde{Z}$ form a basis of the subspace of $\text{col}(X)$ that is perpendicular to $\rcal$; similarly, the last $K-r$ column vectors of $\tilde{W}$ form a basis of the subspace of $S_K(P)$ that is perpendicular to $\rcal$. The model assumption \eqref{eq: identifiability assumptions} then indicates
\begin{equation}\label{eq:identifiability in detail}
\xi\in\text{col}(\tilde{Z}_{1:r}) = \text{col}(\tilde{W}_{1:r})=\rcal, \quad
X\beta\in\text{col}(\tilde{Z}_{(r+1):p})\perp \rcal, \quad \alpha\in \text{col}(\tilde{W}_{(r+1):K})\perp \rcal.  
\end{equation}
Therefore, $\xi$ can be recovered by $\xi = \mathcal{P}_{\rcal} \e Y$, where $\mathcal{P}_{\rcal}$ is the orthogonal projection onto $\rcal$ written as 
\begin{equation}\label{eq:PR}
\mathcal{P}_{\rcal} = \tilde{Z}_{1:r}\tilde{Z}_{1:r}^T=\tilde{W}_{1:r}\tilde{W}_{1:r}^T.
\end{equation}
For the convenience of theoretical analysis later, we also introduce the projection coefficient $\theta \in \bR^p$
$$\theta = (X^TX)^{-1}X^T\xi.$$ 
Recall that $\xi \in \rcal \subset \col(X)$, so we have $\xi = X\theta$, and we have the one-to-one correspondence between $\theta$ and $\xi$.  However, notice that $\theta$ cannot be attributed to conditional effects of covariates, as explained in the previous section.

For estimating $\alpha$ and $\beta$, we project $\e Y$ on $\text{col}(\tilde{Z}_{(r+1):p})$ and $\text{col}(\tilde{W}_{(r+1):K})$, respectively. Since these subspaces need not be orthogonal to each other (the principle angles between the two subspaces are the singular values $\sigma_{r+1},...,\sigma_{r+s}$), the corresponding projections are not necessarily orthogonal projections. Instead, they admit the following forms: 
\begin{eqnarray}
\label{eq:PC}{\mathcal{P}}_{\ccal} &=& \big(\tilde{Z}_{(r+1):p},0_{n\times(K-r)}\big)({M}^T{M})^{-1}{M}^T,\\
\label{eq:PN}{\mathcal{P}}_{\ncal} &=& \big(0_{n\times(p-r)},\tilde{W}_{(r+1):K}\big)({M}^T{M})^{-1}{M}^T,
\end{eqnarray}
where  $M=(\tilde{Z}_{(r+1):p},\tilde{W}_{(r+1):K})$. Both $\alpha$ and $\beta$ can then be recovered  by
\begin{eqnarray}\label{eq:alpha beta proj form}
\alpha = \pcal_{\ncal}\e Y \text{~~and~~} \quad \beta = (X^TX)^{-1}X^T\pcal_{\ccal}\e Y.  
\end{eqnarray} 

%

Finally, to test the hypothesis $H_0: \alpha  = 0$, we can use $\norm{\alpha}^2$ as the statistic. We have
$\alpha = \sum_{i=1}^{K-r} \tilde{W}_{r+i} \gamma_i  = \tilde{W}_{(r+1):K}\gamma,$
such that $\norm{\alpha} = \norm{\gamma}$. Although $\gamma$ can not be uniquely determined because $\tilde{W}_{(r+1):K}$ is only unique up to an orthogonal transformation, we can still uniquely identify its magnitude $\norm{\gamma}^2$ to perform a chi-squared test against the null hypothesis $H_0:\alpha=0$.

In practice, when $P$ and $\e Y$ are not observed, it is natural to replace them everywhere by the available approximations $\hat{P}$ and $Y$ in the procedure above. There is, however, one crucial issue that requires special attention: the plugin estimate $\col(X) \cap S_K(\hat{P})$ for $\mathcal{R}=\col(X) \cap S_K(P)$ is often a bad approximation of $\rcal$. This is because the intersection of subspaces is not robust with respect to small perturbations to the subspaces. For example, it is easy to see that in low dimensional settings, when $\col(X) \cap S_K(P)$ is nontrivial, a small perturbation of $P$ can easily make $\col(X) \cap S_K(\hat{P})$ a null space.
Therefore, rather than using $\col(X) \cap S_K(\hat{P})$ to approximate $\rcal$, the projections  in \eqref{eq:PR}, \eqref{eq:PC} and \eqref{eq:PN} are directly approximated using the eigenvectors of $\hat{P}$ (this partially explains the detailed discussion of the population level estimation above). The whole estimation procedure is summarized in Algorithm~\ref{algo:estimation-most-general}. From here through Section~\ref{secsec:generic-inference}, we will assume that the dimension $r$ of $\rcal$ is known for simplicity. This is because, so far, $A$ has been treated as fixed while a discussion of selecting $r$ naturally involves a detailed analysis of the perturbation mechanism from $P$ to $A$. In Section~\ref{sec:theory} when the perturbation model for $A$ is introduced, we provide a simple method to select $r$ with a theoretical guarantee (see Corollary~\ref{coro:r-consistency}). 

\begin{algo}[Spectral projection estimation]\label{algo:estimation-most-general}
Given $X$, $Y$, $\hat{P}$, $K$ and $r$. 
\begin{enumerate}
\item Calculate an orthonormal basis of $\text{col}(X)$ and forms $Z\in\bR^{n\times p}$. Similarly, calculate $K$ eigenvectors of $\hat{P}$ and form $\hat{W}\in \bR^{n\times K}$.    
\item Calculate the singular value decomposition 
\begin{equation}\label{eq:svd sample}
Z^T\hat{W} = \hat{U}\hat{\Sigma}\hat{V}^T
\end{equation}
and denote
\begin{equation}\label{eq:Z hat W breve}
\hat{Z} = Z\hat{U}, \quad 
\breve{W} = \hat{W}\hat{V}.
\end{equation} 
\item Let $\hat{M}=(\hat{Z}_{(r+1):p},\breve{W}_{(r+1):K})$. Estimate $\mathcal{P}_R$, $\mathcal{P}_{\ccal}$ and $\mathcal{P}_{\ncal}$ by
\begin{eqnarray}
\label{PRhat}\hat{\mathcal{P}}_{\rcal} &=& \hat{Z}_{1:r}\hat{Z}_{1:r}^T,\\
\label{eq:PChat}\hat{\mathcal{P}}_{\ccal} &=& \big(\hat{Z}_{(r+1):p},0_{n\times(K-r)}\big)(\hat{M}^T\hat{M})^{-1}\hat{M}^T,\\
\label{eq:PNhat}\hat{\mathcal{P}}_{\ncal} &=& \big(0_{n\times(p-r)},\breve{W}_{(r+1):K}\big)(\hat{M}^T\hat{M})^{-1}\hat{M}^T.
\end{eqnarray}
\item Estimate $\xi$, $\beta$ and  $\alpha$  by 
\begin{equation}\label{eq:alpha beta theta hat}
\hat{\xi} = \hat{\mathcal{P}}_{\rcal} Y, \quad \hat{\beta} = (X^TX)^{-1}X^T\hat{\mathcal{P}}_{\ccal}Y, \quad \hat{\alpha} = \hat{\mathcal{P}}_{\ncal} Y. 
\end{equation}
And recover the projection coefficient of $\xi$ to $\col(X)$ as $\hat{\theta} = (X^TX)^{-1}X^T\hat{\xi}$.
\item Let $\hat{H} = \hat{\mathcal{P}}_{\rcal}+\hat{\mathcal{P}}_{\ccal}+\hat{\mathcal{P}}_{\ncal}$. Estimate the variance $\sigma^2$ of $\epsilon$ in \eqref{eq:RNC} by
\begin{equation}\label{eq: variance estimate}
\hat{\sigma}^2 = \norm{Y - \hat{H}Y}^2/(n-p-K+r).
\end{equation}
\item Estimate the covariance of $\hat{\beta}$ by 
\begin{equation}\label{eq:beta-cov}
\hat{\sigma}^2(X^TX)^{-1}X^T\hat{\mathcal{P}}_{\ccal}\hat{\mathcal{P}}_{\ccal}^TX(X^TX)^{-1}. 
\end{equation}
The inference of $\beta$ can be done according to Theorem~\ref{thm:general-beta-inference}.
\item To test the significance of $\alpha$, estimate $\gamma$ by  
\begin{equation}\label{eq:gamma hat}
\hat{\gamma} = \breve{W}_{(r+1):K}\hat{\alpha}
\end{equation}
and estimate the covariance matrix of $\hat{\gamma}$ by
$\hat{\Sigma}_{\hat{\gamma}} = \hat{\sigma}^2\breve{W}_{(r+1):K} \hat{\pcal}_{\ncal} \hat{\mathcal{P}}_{\ncal}^T \breve{W}_{(r+1):K}^T.$
Normalize $\hat{\gamma}$ to obtain
$\hat{\gamma}_0  = \hat{\Sigma}_{\hat{\gamma}}^{-1/2}\hat{\gamma}.$
Use $\norm{\hat{\gamma}_0 }^2$ for a chi-squared test (with $K$ degrees of freedom) against the null hypothesis $H_0: \alpha = 0$ according to Theorem~\ref{thm:chisq-test}.
\end{enumerate}
\end{algo}

\section{Generic inference theory of model parameters}\label{secsec:generic-inference}

For valid inference of parameters,  we make the following standard assumptions about the data.

\begin{ass}[Standardized scale]\label{ass:signal-scale}
All columns of $X$ satisfy
$\norm{X_{i}}= \sqrt{n}$. Moreover, 
$$\norm{\e Y}\le C\sqrt{np}$$
 for some constant $C>0$.
\end{ass}

\begin{ass}[Weak dependence of $X$]\label{ass:weak-dependence} Denote $\Theta=(X^TX/n)^{-1}$. Assume that $\Theta$ is well-conditioned, that is, there exists a constant $\rho > 0$ such that 
$$\rho \le \lambda_{\min}(\Theta) \le \lambda_{\max}(\Theta) \le 1/\rho.$$
\end{ass}

Due to the deviation of $\hat{P}$ from the true signal $P$, the estimators would be biased. A common approach to ensure the  valid inference is to control the bias levels of the parameter estimates  \citep{van2014asymptotically,zhang2014confidence,javanmard2014confidence}. In the current context, we need to guarantee that the biases are negligible compared to the standard deviations of the estimators. This requires $\hat{P}-P$ to be controlled to some extent. In particular, the biases will depend on the magnitude of the  perturbation of $S_{K}(P)$ associated with $\text{col}(X)$, defined by 
\begin{eqnarray}\label{eq: RUU}
\tau_n = \|Z^T\hat{W}\hat{W}^T - Z^TWW^T\| = \|(\hat{W}\hat{W}^T - WW^T)Z\|.
\end{eqnarray} 
 
As a warm-up, it is not difficult to get a bound on $\|\hat{\Sigma}-\Sigma\|$ in terms of $\tau_n$. 
Indeed,  from the previous discussion,  the singular decompositions of $Z^TWW^T$ and $Z^T\hat{W}\hat{W}^T$ are 
\begin{equation}\label{eq: RUU decomposition}
Z^TWW^T = U\Sigma\tilde{W}^T, \quad Z^T\hat{W}\hat{W}^T = \hat{U}\hat{\Sigma}\breve{W}^T.
\end{equation}
Therefore, by Weyl's inequality, 
\begin{equation}\label{eq:singular value bound}
\|\hat{\Sigma}-\Sigma\| \le \tau_n.
\end{equation}
Since the estimates depend crucially on the singular value decomposition of $Z^TWW^T$ and $Z^T\hat{W}\hat{W}^T$, \eqref{eq: RUU}, \eqref{eq: RUU decomposition} and \eqref{eq:singular value bound} play a central role in establishing the theoretical results. We now introduce our assumption on the error level that ensures the validity of the inference. It can be seen as a way to characterize the level of perturbation our framework could tolerate. This validity and theoretical insights about this assumption will be studied in more detail later (see Section~\ref{sec:theory}). As we show there, the condition is mild in many commonly studied network settings.

\begin{ass}[Small projection perturbation]\label{ass:small-perturbation} 
Let $Z$ and $W$ be the matrices formed by the bases of $\text{col}(X)$ and $S_K(P)$, respectively. Let $Z^TW=U\Sigma V^T$ be the singular value decomposition  with singular values
specified in \eqref{eq:subspace angles}. Assume $\hat{P}$ is a small projection perturbation of $P$ with respect to $X$, by which we mean
\begin{eqnarray*}
\frac{\tau_n}{\min\left\{(1-\sigma_{r+1})^3,\sigma_{r+s}^3\right\}} = o\big(1/\sqrt{np}\big) .
\end{eqnarray*}
\end{ass}

The first property to be introduced is the estimation consistency of $\hat{\sigma}^2$, which serves as the critical building block for later inference. 
\begin{prop}[Consistency of variance estimation]\label{prop:variance-consistency}
If Assumption~\ref{ass:signal-scale} holds and 
\begin{eqnarray}\label{eq: singular value gap condition}
40\tau_n\le\min\{(1-\sigma_{r+1})^2,\sigma_{r+s}^2\}.
\end{eqnarray}
Then with high probability, the estimate $\hat{\sigma}^2$ defined by \eqref{eq: variance estimate} satisfies
$$
|\hat{\sigma}^2-\sigma^2| \le \frac{C\tau_n}{\min\{(1-\sigma_{r+1})^3,\sigma_{r+s}^3\}}\cdot\frac{n(\sigma^2+p)}{n-p-K+r} 
$$
for some constant $C>0$. 
In particular, if Assumption~\ref{ass:small-perturbation} holds and $p+K = o(n)$ then $\hat{\sigma}^2$ is consistent.
\end{prop}

The covariate effects $\beta$ will be the central target for inference. The first result is the following bound on the bias of $\hat{\beta}$, defined to be $\norm{\e(\hat{\beta}) - \beta}_{\infty}$. 

\begin{prop}[The bias of $\hat{\beta}$]\label{prop:beta-bias}
If condition \eqref{eq: singular value gap condition} holds then there exists a constant $C>0$ such that 
$$
\norm{\e(\hat{\beta}) - \beta} \le \frac{C\tau_n}{\min\{(1-\sigma_{r+1})^3,\sigma_{r+s}^3\}} \cdot\|(X^TX)^{-1}\|\cdot\|X\|\cdot\|X\beta+X\theta+\alpha\|.
$$
In particular, if Assumptions~\ref{ass:signal-scale},  \ref{ass:weak-dependence} and \ref{ass:small-perturbation} hold then the bias of $\hat{\beta}$ is of order $o(1/\sqrt{n})$.
\end{prop}

In general, one may be interested in inferring the contrast $\omega^T\beta$ for a given unit vector $\omega$, such as in comparing covariate effects or making predictions. Let  $\tilde{X} = X/\sqrt{n}$, then from \eqref{eq:alpha beta theta hat}, 
\begin{equation}\label{eq:true-var}
\var(\omega^T\hat{\beta}) = \frac{\sigma^2}{n}\omega^T\Theta \tilde{X}^T\hat{\mathcal{P}}_{\ccal}\hat{\mathcal{P}}_{\ccal}^T\tilde{X}\Theta\omega,
\end{equation}
with $\hat{\mathcal{P}}_{\ccal}$ defined by \eqref{eq:PChat}. According to Proposition~\ref{prop:beta-bias}, a sufficient condition for valid inference of $\omega^T\beta$ is $ \omega^T\Theta \tilde{X}^T\hat{\mathcal{P}}_{\ccal}\hat{\mathcal{P}}_{\ccal}^T\tilde{X}\Theta\omega\ge c$ for some constant $c>0$. Note that this quantity is directly computable in practice. However, we will focus on the population condition with respect to the true matrix ${\mathcal{P}}_{\ccal}$ in the following result.

\begin{thm}[Asymptotic distribution of $\omega^T\hat{\beta}$]\label{thm:general-beta-inference}
If Assumptions~\ref{ass:signal-scale}, \ref{ass:weak-dependence}, \ref{ass:small-perturbation} hold and 
\begin{equation}\label{eq:general-signal-requirement}
\sqrt{np}\cdot\min\{1-\sigma_{r+1},\sigma_{r+s}\}\ge 1.
\end{equation}
For a given unit vector $\omega \in \bR^p$, assume that 
\begin{equation}\label{eq:general-var-requirement}
\|\tilde{Z}_{(r+1):p}^T\tilde{X}\Theta\omega\| \ge c
\end{equation}
for some constant $c>0$ and sufficiently large $n$. Then as $n\rightarrow\infty$, 
$$\p\left(\frac{\omega^T{\hat{\beta}}-\omega^T{\beta}}{\sqrt{\frac{\hat{\sigma}^2}{n}\omega^T\Theta \tilde{X}^T\hat{\mathcal{P}}_\mathcal{C}\hat{\mathcal{P}}_\mathcal{C}^T\tilde{X}\Theta\omega }} \le t\right)\to \Phi(t),$$
where $\Phi$ is the cumulative distribution function of the standard normal distribution.
\end{thm}

Next, we discuss the special case when  $\omega = e_j$, which corresponds to making inference for the individual parameter $\beta_j = e_j^T\beta$ for some $1\le j\le p$. 
In this case, condition \eqref{eq:general-var-requirement} has a simple interpretation.  

\begin{coro}[Valid inference of individual regression coefficient]\label{coro:inference-individual-beta}
Let ${\eta}_j$ be the partial residual from regressing $\tilde{X}_j$ against all other columns of $\tilde{X}$ and $R_j$ be the corresponding partial $R^2$ for this regression. 
If Assumption~\ref{ass:signal-scale}, \ref{ass:weak-dependence}, \ref{ass:small-perturbation} hold and there exists a constant $c>0$ such that
\begin{eqnarray}\label{eq:nontrivial-partial-correlation}
\|\tilde{Z}_{(r+1):p}^T {\eta}_j\| / (1-R_j^2) \ge c.
\end{eqnarray}
Then, as $n\rightarrow\infty$, 
$$\p\left(\frac{\hat{\beta}_j-\beta_j}{\sqrt{\frac{\hat{\sigma}^2}{n}e_j^T\Theta \tilde{X}^T\hat{\mathcal{P}}_\mathcal{C}\hat{\mathcal{P}}_\mathcal{C}^T\tilde{X}\Theta e_j}} \le t\right)\to \Phi(t),$$
where $\Phi$ is the cumulative distribution function of the standard normal distribution.
\end{coro}

To understand condition \eqref{eq:nontrivial-partial-correlation}, notice that ${\eta}_j$ is the signal from $X_j$, after adjusting for the other covariates, while $1-R_j^2$ is the proportion of this additional  information from $\tilde{X}_j$. Since $\|\tilde{Z}_{(r+1):p}^T {\eta}_j\| $ is the  magnitude of ${\eta}_j$'s projection on the subspace $\col(\tilde{Z}_{(r+1):p})$, \eqref{eq:nontrivial-partial-correlation} essentially requires that the additional information contributed by $X_j$ after conditioning on other covariates should nontrivially align with $\col(\tilde{Z}_{(r+1):p})$, the subspace on which the effect of $\beta_j$ lies. 
This type of condition is needed because if $\tilde{Z}_{(r+1):p}^T{\eta}_j=0 $, the parameter space of $\beta_j$ degenerates to $\{0\}$ and the inference is not meaningful.  Notice that requirement \eqref{eq:nontrivial-partial-correlation} is completely different from assuming that $|\beta_j|$ is large, and it does allow $|\beta_j|$ to be zero or small.

The next result shows the consistency of estimating $\alpha$. It also provides bounds for the bias and variance of $\hat{\gamma}$, which are later used for testing the hypothesis $H_0:\alpha=0$.    

\begin{prop}[Consistency of $\hat{\alpha}$ and $\hat{\gamma}$]\label{prop: bias variance of gamma hat}
If Assumption~\ref{ass:signal-scale} and \eqref{eq: singular value gap condition} hold. Then there exists a constant $C>0$ such that with high probability, the following hold. 
\begin{enumerate}
\item Let $\|W\|_{2\rightarrow\infty}$ be the maximum of the Euclidean norms of the row vectors of $W$, then
$$
\|\hat{\alpha}-\alpha\|_\infty \le \frac{C}{\min\{(1-\sigma_{r+1})^3,\sigma_{r+s}^3\}}\cdot\left(\|W\|_{2\rightarrow\infty}\sqrt{K\sigma^2}\log n+\tau_n\sqrt{n(p+\sigma^2)}\right).
$$
\item There exists an orthogonal matrix $O\in\mathbb{R}^{(K-r)\times(K-r)}$ such that  
\begin{eqnarray*}
\left\|\mathbb{E}\hat{\gamma}-O\gamma
\right\|\le \frac{C\tau_n\sqrt{np}}{\min\{(1-\sigma_{r+1})^3,\sigma_{r+s}^3\}}.
\end{eqnarray*}
\item  Let $\Sigma_{\hat{\gamma}}$ be the covariance matrix of $\hat{\gamma}$, then
\begin{eqnarray*}
\left\|\Sigma_{\hat{\gamma}} - \sigma^2 
\begin{pmatrix}
(I_{s}-{\Gamma}^2)^{-1} & 0\\
0 & I_{K-r-s} 
\end{pmatrix} \right\| 
&\le& \frac{C\tau_n}{\min\{(1-\sigma_{r+1})^2,\sigma_{r+s}^2\}},
\end{eqnarray*}
where $\Gamma = \diag(\sigma_{r+1}, \cdots, \sigma_{r+s})$.
\end{enumerate}
\end{prop} 

Proposition~\ref{prop: bias variance of gamma hat} shows that $\hat{\alpha}$ is an element-wise consistent estimate of $\alpha$ if Assumption~\ref{ass:small-perturbation} holds, $\min\{(1-\sigma_{r+1})^3,\sigma_{r+s}^3\}$ is bounded away from zero, and $\|W\|_{2\to\infty}\sqrt{K\sigma^2}\log n = o(1)$. The last requirement holds when $P$ is an incoherent matrix --- the rows of $W$ are in similar magnitudes ---
 a commonly observed property introduced by \cite{candes2009exact, candes2010power}. The other two inequalities of Proposition~\ref{prop: bias variance of gamma hat} show that the bounds for both bias and variance of $\hat{\gamma}$ are of order $o(1)$ if Assumption~\ref{ass:small-perturbation} holds.  These properties directly lead to the validity of the chi-squared test for network effects. 

\begin{thm}[Chi-squared test for $\alpha$]\label{thm:chisq-test}
If Assumptions~\ref{ass:signal-scale} and \ref{ass:small-perturbation} hold, $K$ is a fixed integer, and $p=o(n)$. Then as $n\rightarrow\infty$, the statistic $\norm{\hat{\gamma}_0}^2$ constructed in Algorithm~\ref{algo:estimation-most-general} satisfies  
$$\norm{\hat{\gamma}_0}^2 \xrightarrow[]{d} \chi^2_{K},$$
where $\chi^2_{K}$ is a random variable that follows the chi-squared distribution with $K$ degrees of freedom and $\xrightarrow[]{d}$ denotes the convergence in distribution.
\end{thm}

The next theorem provides the theoretical result for $\hat{\theta}$. 

\begin{thm}[Inference of $\theta$ and $\xi$]\label{thm:theta-inference}
The estimator $\hat{\theta}$ in \eqref{eq:alpha beta theta hat} follows a Gaussian distribution with covariance matrix
$\cov(\hat{\theta}) = \frac{\sigma^2}{n}\Theta\tilde{X}^T\hat{\pcal}_{\rcal}\tilde{X}\Theta.$ 
Furthermore, under Assumptions~\ref{ass:signal-scale},  \ref{ass:weak-dependence} and \ref{ass:small-perturbation},
$\norm{\e(\hat{\theta}) - \theta}  = o(1/\sqrt{n}).$ If in addition, there exists a constant $c>0$ such that
\begin{eqnarray}\label{eq:nontrivial-partial-correlation-theta}
\|\tilde{Z}_{1:r}^T {\eta}_j\| / (1-R_j^2) \ge c,
\end{eqnarray}
where $\eta_j$ and $R_j$ are defined in Corollary~\ref{coro:inference-individual-beta}, then as $n\to \infty$, 
$$\p\left(\frac{\hat{\theta}_j-\theta_j}{\sqrt{\frac{\hat{\sigma}^2}{n}e_j^T\Theta\tilde{X}^T\hat{\pcal}_{\rcal}\tilde{X}\Theta e_j }} \le t\right)\to \Phi(t).$$
The inference of each $\xi_i, i \in [n]$ can be done by noticing $\xi_i = X_{i\cdot}\theta$ where $X_{i\cdot}$ is the $i$th row of $X$.
\end{thm}


Condition~\eqref{eq:nontrivial-partial-correlation-theta} is in parallel of \eqref{eq:nontrivial-partial-correlation}; see the discussion following Corollary~\ref{coro:inference-individual-beta} for its interpretation.

\section{Small projection perturbation under random network models}\label{sec:theory}

 We have introduced the small projection perturbation  condition as a general way to characterize the structural perturbation level our framework could tolerate. To provide more insights about this characterization, in this section, we study the small projection perturbation assumption in a special setting when the perturbation of $P$ comes from random network models. Specifically, we assume that the edge between each pair of nodes $(i,j)$ is generated independently from a Bernoulli distribution with $\p(A_{ij} = 1) = P_{ij}, 1\le i < j \le n$. In particular, $\e A = P$. This model is known as the inhomogeneous Erd\H{o}s-R\'{e}nyi model \citep{Bollobas2007}. To clarify, the randomness discussed in this section comes from the random network model above, which is assumed to be independent of the randomness from the noise $\epsilon$ in the regression model.

Section~\ref{secsec:nonparametric-setting} investigates the most difficult setting when no prior information is available about the network model, while Section~\ref{secsec:parametric-setting} presents the results in the arguably easiest setting when the true underlying network can be accurately estimated.  Section~\ref{sec:Laplacian} further extends the results in another commonly seen setting when the Laplacian matrix represents the relational information. 

\subsection{The non-informative situation: small projection perturbation for adjacency matrices}
\label{secsec:nonparametric-setting} 
In this section, we investigate the situation when no prior information about the network model is available. This can be seen as the most difficult, yet the most general, situation. Arguably, the only reasonable approximation of $P$ is $\hat{P} = A$. We will study  Assumption~\ref{ass:small-perturbation} for this version of $\hat{P}$.

Recall that the main requirement of Assumption~\ref{ass:small-perturbation} is $\tau_n = o(1/\sqrt{n})$, where $\tau_n$ measures the level of perturbation of the network subspace $S_K(P)$ associated with the covariate subspace $\text{col}(X)$, defined in \eqref{eq: RUU}. Existing results relevant for controlling $\tau_n$, such as \cite{Abbe2017entrywise,cape2019two, mao2020estimating, lei2019unified}, do not give sufficiently tight bounds to support the inference, even for dense networks. The recent work of \cite{xia2019data} contains useful tools for obtaining such error bound that allows sparsity when $P$ is a low-rank matrix. However, since in general problems, $P$ may be of full rank, a new tool is needed for theoretical analysis. The following assumption is made about $P$.

\begin{ass}[Eigenvalue gap of the expected adjacency matrix]\label{ass:eigen-gap} 
Let $A$ be the adjacency matrix of a random network generated from the inhomogeneous Erd\H{o}s-R\'{e}nyi model with the edge probability matrix $P=\e A$. 
Assume that the $K$ largest eigenvalues of $P$
are well separated from the remaining eigenvalues and their range is not too large:
$$
\min_{i\le K, \ i'>K} |\lambda_i - \lambda_{i'}| \ge \rho' d, \qquad \max_{i, i'\le K} |\lambda_i - \lambda_{i'}| \le d/\rho',
$$
where $\rho'>0$ is a constant and $d =n\cdot \max_{ij}P_{ij}$. 
\end{ass}

The quantity $d$ can be seen as an upper bound of the network node degree, which has been widely used to measure network density (e.g., \cite{lei2014consistency,le2017concentration}). Assumption~\ref{ass:eigen-gap} is very general in the sense that it only assumes that $P$ has a sufficiently large eigenvalue gap, but 
$P$  need not be low-rank or even approximately low-rank. It appears that this is already sufficient to guarantee the small projection perturbation property of $A$, as stated in the next theorem. We believe the theorem itself can be used as a general tool for statistical analysis of independent interest.

\begin{thm}[Concentration of perturbed projection for adjacency matrix]\label{thm:projection-concentration}
Let ${w}_1,...,w_n$ and $\lambda_1\ge\lambda_2\ge \cdots \ge\lambda_n$ be eigenvectors and corresponding eigenvalues of $P$ and similarly, let $\hat{w}_1,..., \hat{w}_n$ and $\hat{\lambda}_1\ge\hat{\lambda}_2\ge\cdots\ge \hat{\lambda}_n$ be the eigenvectors and eigenvalues of $A$. Denote $W = ({w}_1,...,{w}_K)$ and $\hat{W} = (\hat{{w}}_1,...,\hat{{w}}_K)$. Suppose that Assumption~\ref{ass:eigen-gap} holds and $d \ge C\log{n}$ for a sufficiently large constant $C>0$. Then for any fixed unit vector ${v}$, with high probability,  
\begin{eqnarray}\label{eq: projection bound}
\|(\hat{W}\hat{W}^T - WW^T){v}\| \le  \frac{2\sqrt{K\log n}}{d}.
\end{eqnarray}
\end{thm}

Notice that the bound is for a given deterministic vector ${v}$ instead of all unit vectors. The latter would be equivalent to a bound on $\norm{\hat{W}\hat{W}^T - WW^T}$ and is too large for our inference purpose. 
By restricting the scope of applicability, our result trades off for the tighter bound \eqref{eq: projection bound}, which is crucial for our inference to work in relatively sparse network settings. 

\begin{coro}[Small projection perturbation, adjacency matrix case]\label{coro:small-perturbation-A}
Assume that $\hat{P} = A$ is used in Algorithm~\ref{algo:estimation-most-general}. If Assumption~\ref{ass:eigen-gap} holds and $d$ satisfies
$$\frac{d\cdot\min\left\{(1-\sigma_{r+1})^3,\sigma_{r+s}^3\right\}}{\left(Knp\log{n}\right)^{1/2}} \rightarrow \infty$$
then Assumption~\ref{ass:small-perturbation} holds with high probability for sufficiently large $n$.
\end{coro}

Corollary~\ref{coro:small-perturbation-A} provides a sufficient condition for valid inference. If $\min\left\{1-\sigma_{r+1},\sigma_{r+s}\right\}$ is bounded away from zero while $K$ and $p$ are fixed then $d$ needs to grow faster than $\sqrt{n\log n}$. This is arguably a strong assumption, although it does allow for moderately sparse networks. A natural question is whether it can be relaxed. Unfortunately, the answer to this is negative. We now show that under the current assumptions, the bound \eqref{eq: projection bound} is rate optimal up to a logarithm factor and we cannot guarantee valid inference if $d$ is of the order $\sqrt{n}$.

\begin{thm}[Tightness of concentration and degree requirements]\label{thm:necessary}
Assume $C\sqrt{n/\log n}\le d\le n^{1-\xi}$ for some constant $\xi\in(0,1/2]$ and some sufficiently large constant $C>0$. The following statements hold.
\begin{enumerate}[label=(\roman*)] 
\item There exists a configuration of $(K,{v}, P)$ satisfying the condition of Theorem~\ref{thm:projection-concentration} with $K= O(1)$, under which, for a sufficiently large $n$, 
$$\norm{(\hat{W}\hat{W}^T-WW^T){v}} \ge c/d$$
holds with high probability and a constant $c>0$.
\item There exists a configuration of $(K, p, X, P, \beta, \theta, \alpha)$ satisfying  Assumptions~\ref{ass:signal-scale}, \ref{ass:weak-dependence} and the conditions of Theorem~\ref{thm:projection-concentration} with $K,p = O(1)$ and $\min\left\{1-\sigma_{r+1},\sigma_{r+s}\right\} \ge 1/4$,  under which, for a sufficiently large $n$, we have
$$\norm{\e\hat{\beta}-\beta}_{\infty} \ge c'/d \text{~~ and ~~} \max_j Var(\hat{\beta}_j) \le C'/n$$
with high probability for some constant $C',c'>0$.
\end{enumerate}
\end{thm}

Notice that the two statements of Theorem~\ref{thm:necessary} are different and, in general,  neither implies the other. The first statement is about the general concentration bound of Theorem~\ref{thm:projection-concentration}. It indicates that for the given range of $d$, the bound of \eqref{eq: projection bound} is rate optimal (up to a logarithm order). The second statement is about the necessary condition of the inference problem, under the model-free setting. When $d = \sqrt{n}$, the bias is at least in the order of $1/\sqrt{n}$. In contrast, the standard deviation is at most of order $1/\sqrt{n}$. Therefore, we will not be able to give asymptotically correct testing or confidence intervals under these circumstances.  Meanwhile, under the same condition, Corollary~\ref{coro:small-perturbation-A} shows that if  $d$ grows faster than $\sqrt{n\log n}$, valid inference could be achieved. In combination, we observe a phase transition at the network degree of $\sqrt{n}$. 

To conclude this section, we address the problem of selecting the correct dimension $r$, which is used in Algorithm~\ref{algo:estimation-most-general}, as another application of Theorem~\ref{thm:projection-concentration}. Specifically, denote 
$\hat{d} = \frac{1}{n}\sum_{i,j=1}^n A_{ij}$ and $\bar{d} = \frac{1}{n}\sum_{i,j=1}^n P_{ij}.$
The following rule is used to select $r$:
\begin{equation}\label{eq:r-selection}
\hat{r}  = \max\left\{i: \hat{\sigma}_i \ge 1- \frac{4\sqrt{pK\log{n}}}{\hat{d}}\right\}. 
\end{equation}
This estimate can be shown to give the correct dimension with high probability.

\begin{coro}[Consistency of estimating the dimension of $\rcal$]\label{coro:r-consistency}
Assume that the conditions of Theorem~\ref{thm:projection-concentration} hold and 
\begin{equation}\label{eq:same-scale}
\bar{d} \ge \frac{16\sqrt{pK\log n}}{1-\sigma_{r+1}}.
\end{equation}
Then for a sufficiently large $n$, the dimension estimate $\hat{r}$ defined in \eqref{eq:r-selection} satisfies $\hat{r}=r$ with high probability.
\end{coro}


\subsection{The informative situation: small projection perturbation for parametric models} \label{secsec:parametric-setting}

In the previous section, it is shown that if $A$ is used as $\hat{P}$ in Algorithm~\ref{algo:estimation-most-general}, the network degree needs to grow faster than $\sqrt{n}$ for the small perturbation assumption to hold. In many cases, it may be reasonable to assume that $P$ satisfies certain structural conditions. Leveraging such additional information can provide a more accurate estimate of $P$ and, in turn, ensure the validity of inference for potentially much sparser networks. 
In this section, we investigate a few special cases for which the edge density requirement can be substantially relaxed, highlighting the benefit of knowing the correct model. 

The  discussion begins with the stochastic block model (SBM) as the true network generating model. Specifically, assume that  there exists a vector $g \in [K]^n$ specifying the community node labels and a symmetric matrix $B \in [0,1]^{K\times K}$ such that $P_{ij} = B_{g_ig_j}$. The SBM has been widely used to model community structures \citep{holland1983stochastic}, and its theoretical properties have been well-understood  thanks to intensive research in recent years
. For more detail, refer to the review paper of \cite{abbe2017community} and the references therein. Since the SBM is used in this section only  to illustrate how the degree requirement can be relaxed, we will focus on the following special configuration.

\begin{ass}[Stochastic block model]\label{ass:SBM}
Let $n_k$ be the number of nodes in the $k$th community. Assume that $(1-t)n/K \le n_k \le (1+t)n/K$ for some constant $t$ and all $1\le k\le K$. Moreover, assume that $B = \kappa_n B_0$, where $B_0 \in [0,1]^{K\times K}$ is a fixed matrix and $\kappa_n$ controls the dependence of network density on $n$  and $K=o(n^{1/4})$.  
\end{ass}

Another example is the degree-corrected block model (DCBM) proposed by \cite{karrer2011stochastic}, which generalizes the SBM by allowing the degree heterogeneity of the nodes. Specifically, in addition to the SBM parameters, the model has a degree parameter $\nu \in \bR^n$ such that $P_{ij} = \nu_i\nu_jB_{g_ig_j}$. Notice that $\nu_i$ is only identifiable up to a scale so additional constraints are needed. The following simplified assumption is made about it.

\begin{ass}[Identifiability of degree parameter]\label{ass:DCBM}
Assume $\sum_{g_i = k}\nu_i = n_k$ and there exists a constant $\zeta$ such that $1/\zeta \le \nu_i \le \zeta$.
\end{ass}

Notice that under Assumptions~\ref{ass:SBM} and \ref{ass:DCBM}, both the SBM and DCBM satisfy Assumption~\ref{ass:eigen-gap} with $d = n\kappa_n$. In particular, under the SBM, the subspace of the $K$ leading eigenvectors of $P$ has the same component for all nodes within the same community. Therefore, the model under the SBM reduces to a linear regression model with fixed group effects according to communities.  The following assumption is further made for the two models. 

\begin{ass}[Exact recovery of community labels]\label{ass:strong-consistency}
Assume the community $g$ is known with $n\kappa_n/\log{n} \to \infty$.
\end{ass}

Although $g$ is seldom known in practice, this assumption is based on the fact that in many cases $g$ can be exactly recovered from the network. Indeed, there exist many polynomial-time algorithms that ensure the exact recovery of $g$ such as \cite{gao2015achieving,lei2017generic,Abbe2017entrywise,li2018hierarchical,lei2020consistency} for the SBM and \cite{lei2017generic,chen2018convexified, gao2018community} for the DCBM, to name a few. 
The corresponding regularity condition needed for such strong consistency of community detection is already reflected in the degree requirement in Assumption~\ref{ass:strong-consistency}. Given the community labels, the commonly used estimator of $P$ would be the MLE (see Section~\ref{app:blockmodel} of the supplementary material for details). 
 When such parametric estimates are used as $\hat{P}$ in  Algorithm~\ref{algo:estimation-most-general}, the corresponding estimation and inference procedure are referred to as the \emph{parametric version} of our method, in contrast to the model-free version when $A$ is used as $\hat{P}$. The following theorem shows that the parametric version delivers valid inference under much weaker network density assumptions compared to its nonparametric counterpart.

\begin{thm}[Small projection perturbation under block models]\label{thm:blockmodel-small-perturbation}
Let $d = n\kappa_n$. Let $\hat{P}$ be the parametric estimator of $P$ introduced in Section~\ref{app:blockmodel}. Under either of the following settings:
\begin{enumerate}[label=(\roman*)] 
\item $A$ is generated from $P$ according to the SBM satisfying Assumptions~\ref{ass:SBM} and  \ref{ass:strong-consistency} with
$$
\min\left\{(1-\sigma_{r+1})^6,\sigma_{r+s}^6\right\}\cdot \frac{d}{p\log n} \rightarrow\infty; 
$$
\item $A$ is generated from $P$ according to the DCBM satisfying Assumptions~\ref{ass:SBM}, \ref{ass:DCBM} and \ref{ass:strong-consistency} with
$$
\min\left\{(1-\sigma_{r+1})^6,\sigma_{r+s}^6\right\}\cdot \frac{d}{K^2p\log n} \rightarrow\infty; 
$$
\end{enumerate}
the estimator $\hat{P}$ satisfies Assumption~\ref{ass:small-perturbation} with high probability.
\end{thm}

Again, if $p, K$ and $\min\left\{(1-\sigma_{r+1}),\sigma_{r+s}\right\}$ are fixed, then the degree requirement of the parametric version for the small projection perturbation to hold is $d/\log{n}\to \infty$ for either of two block models. This is much weaker than the degree requirement for the nonparametric estimation, indicating the benefit of knowing the underlying model for $P$. Although the above result is only about two special classes of models for the network, it clearly shows the tradeoff between the model assumption and statistical efficiency: without any information about the true model, the model-free version of the proposed method requires the average degree to grow faster than $\sqrt{n}$ to be effective; in contrast, if the class that the true model belongs to is known then the parametric version of our method performs well for much sparser networks.

\section{Simulation}\label{sec:sim}
In this section, we present a simulation study to evaluate the proposed method and compare it with other benchmark methods for linear regression on networks. First, we evaluate the validity of the inference framework and demonstrate the predicted phase transition of the model-free method, as well as the advantage of the parametric version in the situation when the true network model is known. After that, we introduce comparisons with other benchmark methods under several network effect models.

\subsection{Inference validation}\label{secsec:Fixed-Design-Sim}

We will evaluate our theoretical claims about the inference and phase transition in this subsection. For this purpose, we will again assume $K$ to be known to exactly match our theoretical framework in this subsection.  However, this does not lose generality because in all of our experiment settings, because the methods we mentioned before  \citep{chen2014network, le2015estimating,li2016network} can accurately identify the $K$ almost perfectly in our settings. We fix $p=K=4$ in all configurations. The network is generated by the SBM with $K=4$ communities. In Section~\ref{sec:additional-sim} of the supplementary material, we provide additional results and also a study when the network is generated by the more general DCBM. Using the parameterization in Section~\ref{secsec:nonparametric-setting}, we set $B = 0.2 \mbone_4\mbone_4^T + 0.8I_4$, where $\mbone_k$ denotes the all-one vector of length $k$. Three levels of sample size $n= 300, 500, 1000, 2000, 4000$, and three levels of average expected degree  $\varphi_n = 2\log{n}, \sqrt{n}$, and $n^{2/3}$ for each level of $n$ are compared. Given the eigenvectors $w_1,...,w_n$ from $P$, we generate $X \in \bR^{n\times p}$ in the following way: Set $X_{1}/\sqrt{n} = w_1$; Set $X_{j}/\sqrt{n} = \sqrt{\frac{1}{25}}w_j + \sqrt{\frac{24}{25}}w_{j+3}, j = 2, 3, 4$. This configuration gives a design with $r=1, s=3$ and $(\sigma_1, \sigma_2, \sigma_3, \sigma_4) = (1, 0.2, 0.2, 0.2)$. In particular, the setup with $\beta = (0,1,1,1)^T$ and $\theta = (1,0,0,0)^T$ satisfies the proposed model and is fixed  in all settings. Similarly, the $\gamma$ can be any vector with the first coordinate being zero. We consider two cases: $\gamma = (0,1,1,1)^T$ and $\gamma = (0,0,0,0)^T$. As discussed, the case of $\gamma = (0,0,0,0)^T$ indicates that there is no network effect. The two settings do not give any difference on inference of $\beta$ while the latter is used to verify the validity of the $\chi^2$ test.  Also, notice that $X_1$ violates assumption~\eqref{eq:nontrivial-partial-correlation} because it is completely orthogonal to the space of $\col(\hat{Z}_{r+1:p})$.  Therefore, valid inference cannot be made for $\beta_1$. This setup highlights that having a degenerate parameter space on $\beta_1$ does not impact the validity of our inference on the other parameters $\beta_2, \beta_3, \beta_4$, as shown by both the theoretical and numerical analyses. Both the model-free and parametric versions of the subspace projection (SP) methods are used to demonstrate the empirical difference, denoted by ``SP" and ``SP-SBM", respectively. All results are taken as the average of the 50 independent replications.

Table~\ref{tab:Controlled-SBM-Ratio} shows the ratio between the absolute bias and standard deviation (SD) averaged across $\hat{\beta}_2, \hat{\beta}_3$ and $\hat{\beta}_4$ for different $n$ and network densities, when the network is generated from the SBM.  For average degree of order $\sqrt{n}$ or below, the ratio does not vanish with increasing $n$ for the SP, indicating that valid statistical inference for the parameters cannot be achieved. For denser networks, the bias becomes vanishing; thus, the asymptotic inference becomes valid. These observations coincide with the prediction of the phase transition at degree $\sqrt{n}$. In contrast, the SP-SBM results in much smaller bias-SD ratios and achieves vanishing ratio even at degree $\sqrt{n}$. At the borderline case $2\log{n}$, such vanishing pattern is not clearly observed, even for the SP-SBM.
%

\begin{table}
\caption{\label{tab:Controlled-SBM-Ratio}Average bias-SD ratios for $\beta_2, \beta_3$ and $\beta_4$ when the network is generated from the SBM.}
\centering
\fbox{
\begin{tabular}{l|rrrr|rrr}
\hline
 \multirow{2}{*}{}&  & \multicolumn{3}{c}{bias-SD ratio}  & \multicolumn{3}{c}{coverage probability} \\  
 \hline
\multirow{2}{*}{Method} & \multirow{2}{*}{$n$} & \multicolumn{3}{c}{average expected degree}   & \multicolumn{3}{c}{average expected degree}\\ 
 &  & $2\log{n}$ & $\sqrt{n}$ & $n^{2/3}$ & $2\log{n}$ & $\sqrt{n}$ & $n^{2/3}$  \\ 
  \hline
 \multirow{5}{*}{SP} &300&  1.497 & 0.796 & 0.393 & 0.813 & 0.908 & 0.941 \\
 &  500 & 1.310 & 0.720 & 0.255 & 0.845 & 0.910 & 0.949 \\ 
& 1000 & 1.929 & 0.791 & 0.232  & 0.717 & 0.888 & 0.948 \\ 
 & 2000 & 2.273 & 0.723 & 0.254& 0.585 & 0.891 & 0.943 \\ 
 & 4000 & 2.818 & 0.790 & 0.213& 0.554 & 0.884 & 0.949 \\ 
\hline
 \multirow{5}{*}{SP-SBM} &300& 1.257 & 0.281 & $<10^{-4}$ & 0.844 & 0.945 & 0.950 \\
 & 500 & 0.826 & 0.157 & $<10^{-4}$ & 0.911 & 0.949 & 0.949 \\ 
 & 1000  & 0.887 & 0.012 & $<10^{-4}$ & 0.895 & 0.949 & 0.950 \\ 
 & 2000 & 0.934 &$<10^{-4}$ & $<10^{-4}$   & 0.856 & 0.950 & 0.950 \\ 
 & 4000 & 0.898 & $<10^{-4}$ & $<10^{-4}$  & 0.838 & 0.950 & 0.951 \\ 
   \hline
\end{tabular}
}
\end{table}

Correspondingly, Table~\ref{tab:Controlled-SBM-Ratio} also shows the resulting average coverage probabilities $p_{\text{cov}}$ of the 95\% confidence intervals for $\beta_2, \beta_3$ and $\beta_4$. Notice that, differently from Table~\ref{tab:Controlled-SBM-Ratio}, the coverage probability is calculated as the average of 500 Monte Carlo runs, thereby taking the errors of $\hat{\sigma}$ into account. It can be seen that the confidence interval for the SP is not good enough for degree of $\sqrt{n}$ or lower, but does give valid results for denser graphs. The parametric version is always more accurate than the model-free version, and in particular, remains valid for the order of $\sqrt{n}$.

%
%

\subsection{Comparison with other benchmarks}\label{secsec:Random-Design-Compare}

In this subsection, we will evaluate the practical performance of our method in different settings of network effects and compare it with other benchmark methods. 
The first benchmark in the comparison is the ordinary least squares method (OLS). The other two are the SIM of \cite{manski1993identification} and the RNC of \cite{li2016prediction} introduced before. It has been shown that the RNC is a better modeling option than both OLS and the SIM \citep{li2016prediction}. However, the RNC method cannot make inference of the parameters, and it is computationally more intensive than our method. In this experiment, the RNC is always tuned by 10-fold cross-validation. For our method, the number $K$ is determined by the method of \cite{le2015estimating}. The theory-driven tuning of $r$ introduced in \eqref{eq:r-selection} is based on large-sample properties and may be conservative for small samples. We also propose a bootstrapping method in Section~\ref{sec:r-tuning} of the supplementary material. This method is used to select $r$ for the numerical and data examples from now on and works very well.  Lastly, a systematic model fitting procedure is used for SP to take advantage of its available inference framework.  Specifically, the SP methods would check the $\chi^2$ test result of $\gamma$. If the p-value for the $\chi^2$ test is more significant than 0.05 (any other reasonable level can be chosen), the method would instead return OLS fit, which is equivalent to constraining $\gamma = 0$ in our estimation procedure.

\begin{table}
\caption{\label{tab:SBM-compare} Mean squared error (MSE) $\times 10^2$ of $\e Y$ when the network is generated from the SBM with $p=4, K=4$ with three types of individual effects, in small sample case ($n=300$) and large sample case ($n=1000$), respectively. } 
\centering
\fbox{{\small
\begin{tabular}{l|l|l|rrrrrrr}
  \hline
$n$& indiv. effects &   avg. degree & SP & SP-SBM & OLS & SIM &RNC & SP-L & SP-SBM-L\\ 
  \hline
\multirow{9}{*}{$300$}& \multirow{3}{*}{eigenspace} & $2\log{n}$  & 18.20 & 13.59 & 40.65 & 246.05 & 12.11 & 40.40 & 27.13 \\ 
&& $\sqrt{n}$ & 11.13 & 2.36 & 40.65 & 40.16 & 11.57 & 39.00 & 21.50 \\ 
& & $n^{2/3}$ & 3.66 & 0.30 & 40.65 & 36.35 & 9.53 & 22.93 & 20.47 \\  \cline{2-10}
&\multirow{3}{*}{$\gamma = 0$} & $2\log{n}$ & 0.44 & 0.43 & 0.37 & 0.70 & 1.11 & 0.43 & 0.42 \\ 
& & $\sqrt{n}$ & 0.42 & 0.42 & 0.37 & 0.68 & 1.73 & 0.38 & 0.44 \\ 
& & $n^{2/3}$ & 0.42 & 0.41 & 0.37 & 0.69 & 2.31 & 0.40 & 0.42 \\  \cline{2-10}
&\multirow{3}{*}{smooth} & $2\log{n}$ & 25.39 & 25.40 & 25.36 & 189.83 & 17.83 & 16.98 & 25.38 \\ 
&  & $\sqrt{n}$ & 15.67 & 14.81 & 27.09 & 25.82 & 23.16 & 13.87 & 20.94 \\ 
&  & $n^{2/3}$ & 2.55 & 2.97 & 24.07 & 23.24 & 16.61 & 12.30 & 13.52 \\ 
   \hline
\multirow{9}{*}{$1000$}&\multirow{3}{*}{eigenspace} & $2\log{n}$  & 13.60 & 8.67 & 38.47 & 38.68 & 12.07 & 39.38 & 23.57 \\ 
&& $\sqrt{n}$ & 5.24 & 0.24 & 38.47 & 37.90 & 10.52 & 23.13 & 19.35 \\ 
& & $n^{2/3}$ & 1.44 & 0.11 & 38.47 & 35.82 & 7.83 & 20.16 & 19.29 \\  \cline{2-10}
&\multirow{3}{*}{$\gamma = 0$} & $2\log{n}$ & 0.14 & 0.14 & 0.13 & 0.43 & 0.56 & 0.13 & 0.14 \\ 
& & $\sqrt{n}$ & 0.14 & 0.14 & 0.13 & 0.22 & 0.80 & 0.14 & 0.13 \\ 
& & $n^{2/3}$ & 0.14 & 0.14 & 0.13 & 0.24 & 2.19 & 0.13 & 0.14 \\  \cline{2-10}
&\multirow{3}{*}{smooth} & $2\log{n}$ & 25.34 & 25.26 & 25.33 & 26.40 & 21.73 & 12.75 & 25.29 \\ 
 & & $\sqrt{n}$ & 14.69 & 14.75 & 17.61 & 16.54 & 18.79 & 8.90 & 16.18 \\ 
  && $n^{2/3}$ & 1.16 & 1.23 & 3.30 & 2.55 & 11.95 & 1.74 & 2.26 \\ 
   \hline
\end{tabular}
}}
\end{table}

As before, the network is generated by the SBM with $K=4$  with the same configurations. We take $n=300$ and $n=1000$, as representative cases for small-sample and large-sample settings. The results for the more general situation of DCBM can be found in Section~\ref{sec:additional-sim} of the supplementary material. However, to avoid overly tailoring the setup for our methods, $X$'s are generated differently. Specifically, $X_2, X_3, X_4$ are randomly generated from Gaussian distribution $N(0,1)$, uniform distribution $U(0,1)$ and exponential distribution $\text{Exp}(1)$ before standardization is applied. $X_1$ is directly set to be $\sqrt{n}w_1$ as before, where $w_1$ is the first eigenvector of $P$.  Meanwhile, three different schemes are used to generate the individual effect $\alpha$, so that the model misspecification situations can also be tested. In particular, the first scenario is exactly the same setting as in the last section, which corresponds to the assumed model. This individual effect setting is called ``eigenspace". In the second scenario,  we set $\gamma = 0$, thus the model becomes the standard linear regression model for which OLS is designed. In the last scenario, we set $\alpha$ to be the average of the three eigenvectors corresponding to the three smallest nonzero eigenvalues of the observed Laplacian matrix $L$. This $\alpha$ gives a small value of $\sum_{i\sim j}(\alpha_i - \alpha_j)^2$, and thus matches the assumption of the RNC framework  \citep{li2016prediction}.  As discussed in Appendix~\ref{sec:Laplacian}, our method should still work by using $\hat{P} = L$ when $\alpha$ is defined in this way. Therefore, for the evaluation in this section, the Laplacian version of the SP estimator is also included and labeled as ``SP-L''.

Since the data generating models have different parameterizations, we directly compare the relative mean squared error (MSE) of the expected response $\e Y$. All of the results are averaged over 50 independent replications and are shown in Table~\ref{tab:SBM-compare}.  The overall pattern remains the same for small-sample and large-sample problems. Under the model with eigenspace individual effects (the proposed model), it can be seen that OLS never renders competitive performance while the SIM is only slightly better than OLS. The RNC gives much better results than OLS and the SIM because it effectively incorporates the network information. However, since the network is noisy, its performance is far from adequate. Both versions of our method (SP and SP-SBM) significantly outperform the other methods, with the parametric version performing better than the model-free version, as expected. When there are no individual effects ($\gamma = 0$), OLS becomes the correct model and gives optimal results, as expected. The RNC and our methods all are correctly adaptive to this setting, but the SP methods are still better than the RNC.  The performance under the model with smooth individual effects (the RNC model) is noisier as it depends on the perturbation of the network. Overall, neither OLS nor the SIM works. However, under the SBM with average degree $n^{2/3}$, the network is too dense, such that the smooth individual effects become almost identical everywhere, and OLS becomes effective. In this setting, using the Laplacian matrix in our method still models the correct form of individual effects, so SP-L still works and is more effective than the RNC.  Note that the parametric versions (SP-SBM-L) are no longer better than the SP in this situation because the individual effects depend on the observed network instead of the population signal; thus, the parametric methods are not using the correct model.

In conclusion, when there are network effects, our method outperforms the other benchmark methods in the experiments. When there are no network effects, it is adaptive to the simpler model and gives a similar result to OLS. Compared to the model-free SP, the parametric SP method is less robust to model misspecification. 

\subsection{Timing evaluation}\label{secsec:timing}
The proposed method is generally computationally efficient. We include the timing evaluation in this section. All the computations are on a Linux system Intel(R) Xeon(R) CPU E5-2630 v3 \@ 2.40GHz CPU and 2G memory. Our implementation of the SP method is completely done in R by taking advantage of sparse matrices and the efficient partial eigen-decomposition algorithm in \cite{RSpectra}.  The SIM fitting is based on the two-stage least square method in the R package \texttt{spatialreg} \citep{spatialreg} and the RNC method is implemented in the R package \texttt{netcoh} \citep{li2016netcoh}.  In Table~\ref{tab:TimingSBM} we include the average computational time for model fitting in the same settings of Section~\ref{secsec:Fixed-Design-Sim}.

The computational efficiency of our method depends on the network density. The computation is slightly faster on a sparse network, but the difference is marginal. This is because while sparse networks result in more efficient spectrum decomposition, the other part of model fitting involving $X$ and the eigenvectors do not benefit from sparsity. Overall, the model fitting procedure takes about 1.3 seconds for a network with 4000 nodes. All the other methods are also reasonably fast. The SIM is similar to our method on sparse networks, but it becomes much slower when the network is denser. The RNC is generally the slowest, taking about 2.5 times that of the SP method.

Notice that the timing only accounts for modeling fitting without prior tuning because these prior procedures can be flexibly specified by users. Among the methods we recommend for selecting $K$ in the SP method, the USVT of \cite{chatterjee2015matrix} is the fastest while the Beth-Hession method of \cite{le2015estimating} is slightly slower; both of them are cheaper than our model fitting procedure itself. The cross-validation method of \cite{li2016network} is much more computationally intensive, but still remains feasible for moderately large networks. The $r$ selection based on \eqref{eq:r-selection} does not introduce an additional cost. Alternatively, if the bootstrap is used to select $r$, the computational cost is upper bounded by that of the single model fitting in Table~\ref{tab:TimingSBM} times the bootstrap cost, if no distributed computing is involved. Similarly, for the RNC method, the cross-validation tuning time is not included, which depends on how many tuning parameter values are examined and the number of data-splitting for each value. Empirically, RNC needs cross-validation for a wide range of tuning parameter values so its computational disadvantage will be even more significant than Table~\ref{tab:TimingSBM}.

%

\begin{table}
\caption{\label{tab:TimingSBM} Average timing (in seconds) of model fitting procedures in the settings of Table~\ref{tab:Controlled-SBM-Ratio}. The standard deviation is below 3\% of the timing in all settings. }
\centering
\fbox{
\begin{tabular}{l|r|rrrr}
  \hline
$n$ &avg. degree & SP & OLS & SIM & RNC \\ 
  \hline
\multirow{3}{*}{$300$} &$2\log{n}$ & 0.01 & $<0.01$ & 0.01 & 0.01 \\ 
 &$\sqrt{n}$  & 0.01 & $<0.01$ & 0.02 & 0.01 \\ 
 &$n^{2/3}$  & 0.01 & $<0.01$ & 0.02 & 0.01 \\ 
\hline
\multirow{3}{*}{$500$} &$2\log{n}$  & 0.01 & $<0.01$ & 0.03 & 0.02 \\ 
 &$\sqrt{n}$    & 0.01 & $<0.01$ & 0.03 & 0.02 \\ 
 &$n^{2/3}$& 0.01 & $<0.01$ & 0.05 & 0.02 \\ 
\hline
\multirow{3}{*}{$1000$} &$2\log{n}$  & 0.04 & $<0.01$ & 0.12 & 0.11 \\ 
 &$\sqrt{n}$    & 0.04 & $<0.01$ & 0.09 & 0.10 \\
 &$n^{2/3}$& 0.04 & $<0.01$ & 0.20 & 0.09 \\ 
\hline
\multirow{3}{*}{$2000$} &$2\log{n}$& 0.19 & $<0.01$ & 0.24 & 0.52 \\ 
 &$\sqrt{n}$  & 0.18 & $<0.01$ & 0.29 & 0.51 \\ 
&$n^{2/3}$& 0.18 & $<0.01$& 0.79 & 0.51 \\ 
\hline
\multirow{3}{*}{$4000$} &$2\log{n}$& 1.30 & $<0.01$ & 1.22 & 3.15 \\ 
 &$\sqrt{n}$& 1.33 & $<0.01$ & 1.43 & 3.10 \\ 
&$n^{2/3}$& 1.34 & $<0.01$& 3.94 & 3.12 \\ 
   \hline
\end{tabular}
}
\end{table}

\section{ School conflict reduction study}\label{sec:conflict}
In \cite{paluck2016changing}, an experiment is conducted to study the impacts of educational workshops on reducing conflicts in schools. The experimenters randomly selected 26 middle schools in New Jersey in which they held educational workshops.  In each school, the experimenters first determined a group of eligible students. A small proportion of them was selected (randomly after gender and race blocking) to participate in bi-monthly educational workshops about school conflicts. Students were also asked to name their friends (with whom they spend time) in school. The friend nomination may not be symmetric in this case. Following the standard approach in previous  studies \citep{bramoulle2009identification,goldsmith2013social,paluck2016changing}, we ignore the directions of edges and treat two students as connected as long as either one identifies the other as a connection. This approach is widely applied with the assumption that as long as one side nominates the association, it reliably indicates that connections reasonably exist.

One interesting aspect of the data set is that the social network information is collected in two waves of surveys, one at the early stage of the school year and one towards the end of the year. The nominated edge sets are very different from the two surveys in all schools. On average, across all schools, 66\% of edges only appear in one survey. Such a mismatch is commonly observed in social studies \citep{bramoulle2009identification,goldsmith2013social}.  It also highlights the fact that the observed social network relation is noisy. First, it is questionable whether one can use a “true network”. Second, it is not immediately clear how to construct a social network for further analysis. These difficulties reveal that a suitable model in practice should be robust to network variations, a property that our proposed model has been designed for. In Sections~\ref{secsec:full-model} and \ref{secsec:reduced-model}, we will use the weighted network based on both surveys. If an edge is nominated in both surveys, the edge weight is 1; if the edge only appears in one survey, the edge weight is 0.5; otherwise, there is no connection. In Section~\ref{secsec:network-versions}, we will study the impacts of different ways to construct the network and show that our proposed SP model gives consistent conclusions under these different constructions.

At the end of the school year, students answered the 13 questions about their rating of the school's ``friendliness" atmosphere, such as ``How many students at this school think it is good to be friendly and nice with all students no matter who." For each of such questions, the student would give a score ranging from 0 to 5, representing their estimate of the proportion of students satisfying the condition in the question, ranging from ``Almost nobody" (0) to ``Almost everyone" (5). Overall, a higher rating indicates a more friendly environment.\footnote{Two of the questions were about negative atmosphere. So we transform then by the $5-$ the original scores before combining with other questions to ensure all scores are measuring positive atmosphere.} We use the average of the 13 scores as the overall friendliness rating of each student. The score is further raised by a power of 1.2 to make the data more Gaussian based on the Box-Cox test. In this example, our target is to learn the workshop's impact on students’ perception of the school atmosphere and conflict situations, and which of the other demographic attributes are also strongly related to the response. The demographic attributes include race (white, black, Hispanic, Asian, other), gender (M, F), grade, whether the student is a returning student from the previous year, and whether the student lives with both parents. The involvement of the workshop is indicated by a single binary vector (treatment). For each school, after removing observations with missing values in the response and predictors, we take the largest connected component of the social network as the final data. In total, 9026 students from 26 schools are included in our analysis. Each school will be assigned an individual effect parameter to account for potential school-level differences. Figure~\ref{fig:reg-hist} shows distributions of non-binary variables in the data set. The distributions of the binary variables are included in Appendix~\ref{sec:more_school_conflicts}.

\begin{figure}[h]
\centering
\includegraphics[width=\textwidth]{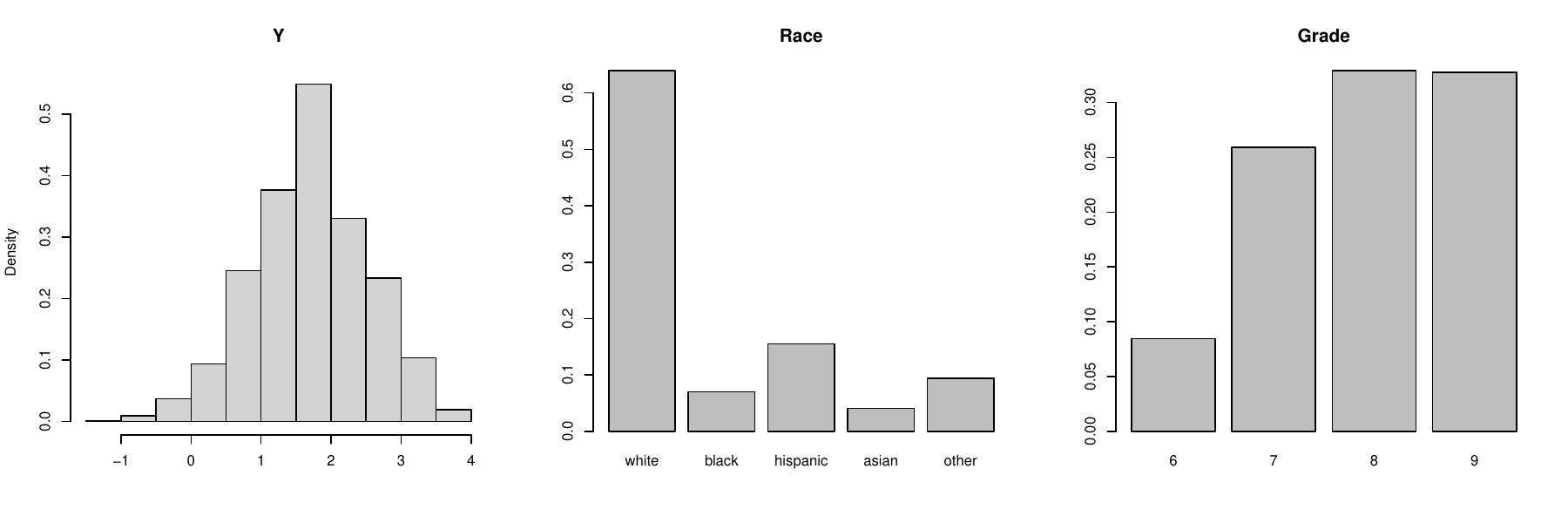}
\caption{Marginal distributions of the response variable, race and grade in the data set. }
\label{fig:reg-hist}
\end{figure}

Notice that the treatment is randomized by design, uncorrelated with other covariates. We expect all reasonable estimation procedures to give a similar treatment effect estimate. However, the variance of the estimation (and the validity of the inference) would depend on modeling effectiveness. This will be reflected in our results to follow.

\subsection{Full models}\label{secsec:full-model}

We first use all predictors to fit a regression model based on our proposed SP and the three benchmark methods (OLS, SIM, and RNC). Due to the redundancy of insignificant variables, these may not serve as final models. However, we can have an evaluation of four methods on common ground.

\paragraph{Network effects and model significance.} In the SP estimation, $r$ is detected to be 0, so there is no confounding observed between the covariates and network space. The $\chi^2$ test of our SP method gives a very small p-value ($<10^{-6}$), representing strong evidence of network effects. Appendix~\ref{sec:more_school_conflicts} includes visualizations of our model fitting residuals, and the residuals follow a Gaussian pattern reasonably well. The fitted $\alpha$ values are shown in Figure~\ref{fig:prediction_error}-(a). Recall that our model includes the standard linear regression model as a special case. That means, if the OLS inference is correct, our inference should also be correct. Therefore, the rejection in the $\chi^2$ test also indicates that the OLS model-fitting without network effects is not proper for the data.  The test of network autoregressive effect in the SIM model gives a p-value of 0.92, suggesting no strong evidence of the endogenous correlation. Furthermore, the exogenous coefficients are not significantly different from zero either, as shown in Table~\ref{tab:fullcovariates}. Together, these indicate that the SIM autoregressive pattern is supported. The RNC model does not come with an inference framework. However, the cross-validation procedure in the model fitting selects a penalty that gives a very different fitting from the OLS, which implicitly indicates the potential network effects. 
\paragraph{Parameter estimation} Table~\ref{tab:fullcovariates} displays the estimated coefficients (excluding schools and intercepts). The SP fitting and OLS fitting agree on the overall magnitude in most coefficients with a 5\%-10\% difference in the parameters with small p-values. In particular, on the treatment effect, the SP estimates a negative effect of -0.084 while the OLS gives an estimate of -0.089. The similarity is expected, as discussed before. However, the effectiveness of incorporating other factors would influence the variance for the inference. In this example, while the SP estimates a slightly smaller treatment effect, the corresponding p-value turns out to be smaller, indicating a much smaller estimation variance. The SIM model has very different estimates due to incorporating all local averages, except for the treatment effect due to the randomized design. The coefficients are not directly comparable with SP or OLS, but most of them come with larger p-values, especially considering the number of tests. Overall, we consider the current SIM fitting as non-informative. RNC method does not allow for the intercept and the school effects. Therefore, the estimated parameters are not meaningfully comparable to the other methods. We treat this as a limitation because it eliminates the straightforward interpretation of school-level fixed effects. The drawback of lacking an inference framework becomes more critical in this case. The table shows that the RNC renders larger estimated effects, but it is unclear how significant they are.

\begin{table}
\caption{\label{tab:fullcovariates} Estimates and p-values of the common covariates (excluding the intercept and schools). For the SIM,  exogenous coefficients (for local averages of covariates) are also included. ``--" indicates that the quantity is not available in the setting.} 
\centering
\fbox{\begin{tabular}{r|rr|rr|rr|rr}
  \hline
&\multicolumn{2}{c}{SP} & \multicolumn{2}{c}{OLS} &\multicolumn{2}{c}{SIM} & \multicolumn{2}{c}{RNC}\\
  \hline
 & coef. & p-value & coef. & p-value & coef. & p-value & coef. & p-value \\ 
  \hline
Race: white & 0.0260 & 0.2432 & 0.0261 & 0.4831 & 0.0152 & 0.6941 & 0.2131 & --  \\ 
Race: black & -0.1193 & 0.0121 & -0.1288 & 0.0149 & -0.1079 & 0.0564 & 0.194 &  -- \\ 
 Race: hispanic  & -0.0313 & 0.2261 & -0.0244 & 0.5580 & -0.0207 & 0.6272 & 0.1102 &  --\\ 
 Race: asian & 0.0433 & 0.2477 & 0.0394 & 0.5358 & 0.0297 & 0.6449 & 0.1637 &  -- \\ 
 Grade & -0.1663 & $<10^{-4}$ & -0.1784 & $<10^{-4}$ & 0.0148 & 0.8574 & 0.3672 & -- \\ 
  Gender & -0.0091 & 0.3557 & -0.0099 & 0.6842 & 0.0842 & 0.1458 & 0.1207 &  -- \\ 
  Return student & -0.1429 & 0.0001 & -0.1296 & 0.0002 & -0.0816 & 0.0859 & -0.1379 &  -- \\ 
  Live w/ both parents& 0.0715 & 0.0066 & 0.0767 & 0.0078 & 0.0610 & 0.0376 & 0.1710 & --  \\ 
  Treatment & -0.0839 & 0.0442 & -0.0889 & 0.0709 & -0.0895 & 0.0702 & -0.0768 &  -- \\ 
   \hline
local avg.: white & -- & -- & -- & -- & 0.0719 & 0.5769 & -- & -- \\ 
local avg.: black  & -- & -- & -- & -- & 0.0144 & 0.9321 & -- & -- \\ 
local avg.: hispanic & -- & -- & -- & -- & -0.0023 & 0.987 & -- & -- \\ 
local avg.: asian & -- & -- & -- & -- & -0.0152 & 0.9426 & -- & -- \\ 
local avg.: grade & -- & -- & -- & -- & -0.1776 & 0.2977 & -- & -- \\ 
local avg.: gender  & -- & -- & -- & -- & -0.1389 & 0.1267 & -- & -- \\ 
local avg.: return & -- & -- & -- & -- & -0.1041 & 0.2623 & -- & -- \\ 
local avg.: w. parents & -- & -- & -- & -- & 0.2654 & 0.0064 & -- & -- \\ 
local avg.: treatment & -- & -- & -- & -- & 0.1982 & 0.2734 & -- & -- \\ 
\end{tabular}}
\vspace{-0.5cm}
\end{table}

\begin{figure}
\begin{center}
\subfigure[  Fitted $\alpha$ by the SP method  ]{\includegraphics[width = 0.8\textwidth]{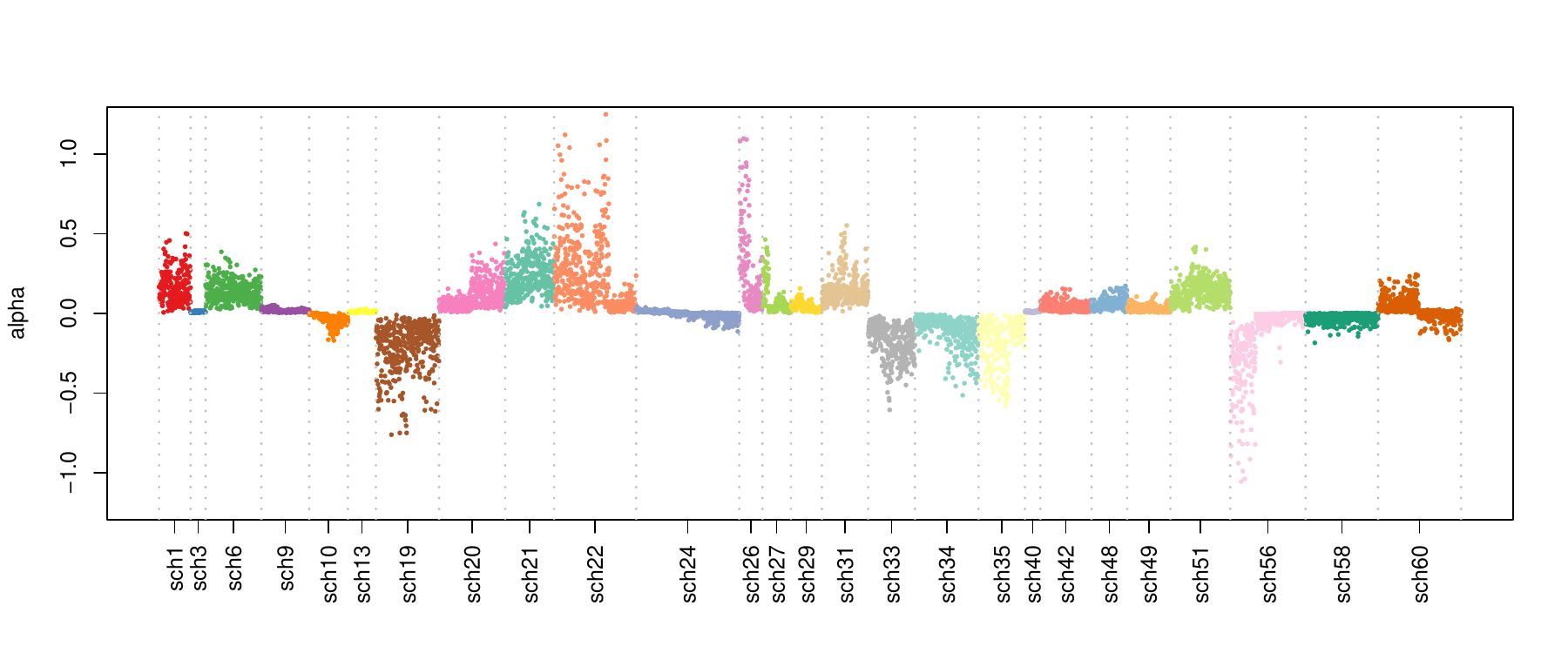}}
\subfigure[  Mean squared prediction error at each school: SP vs. OLS ]{\includegraphics[width = 0.8\textwidth]{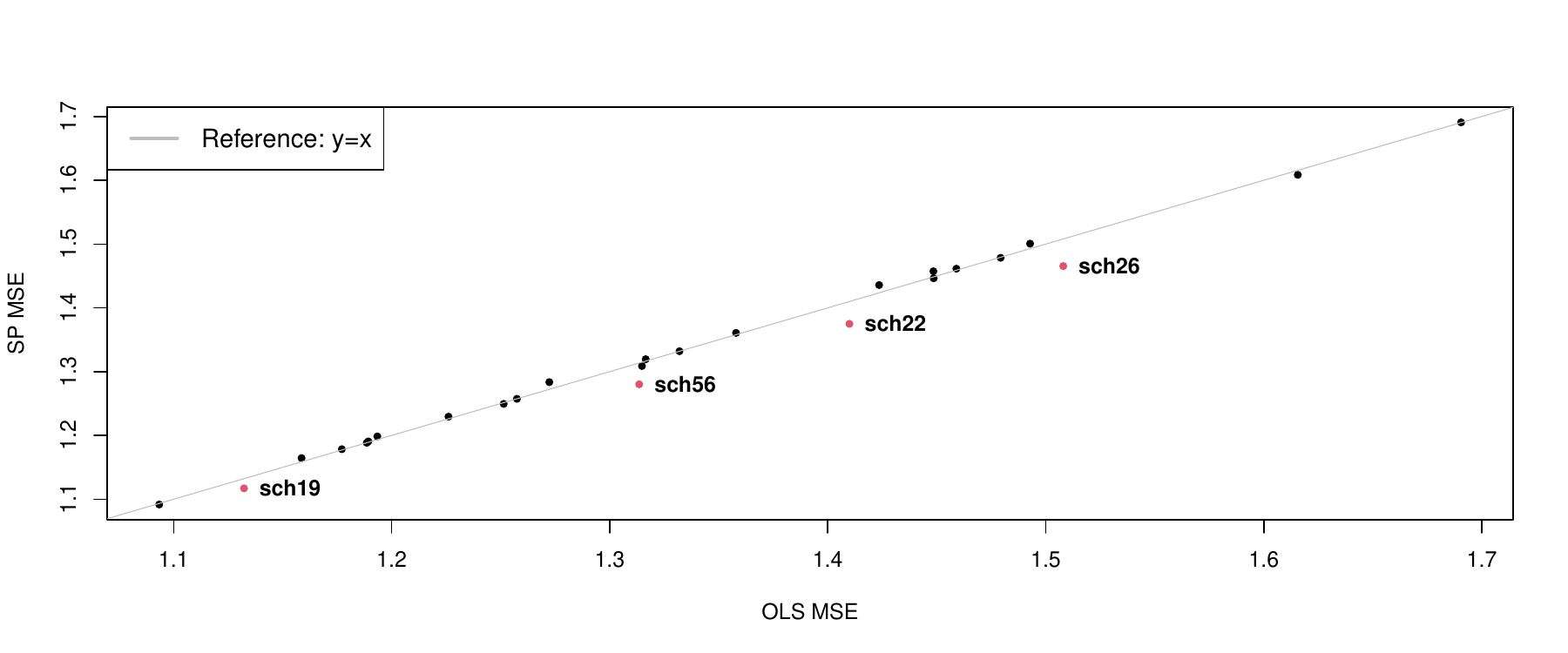}}
\subfigure[  Mean squared prediction error at each school: SP vs. SIM  ]{\includegraphics[width = 0.8\textwidth]{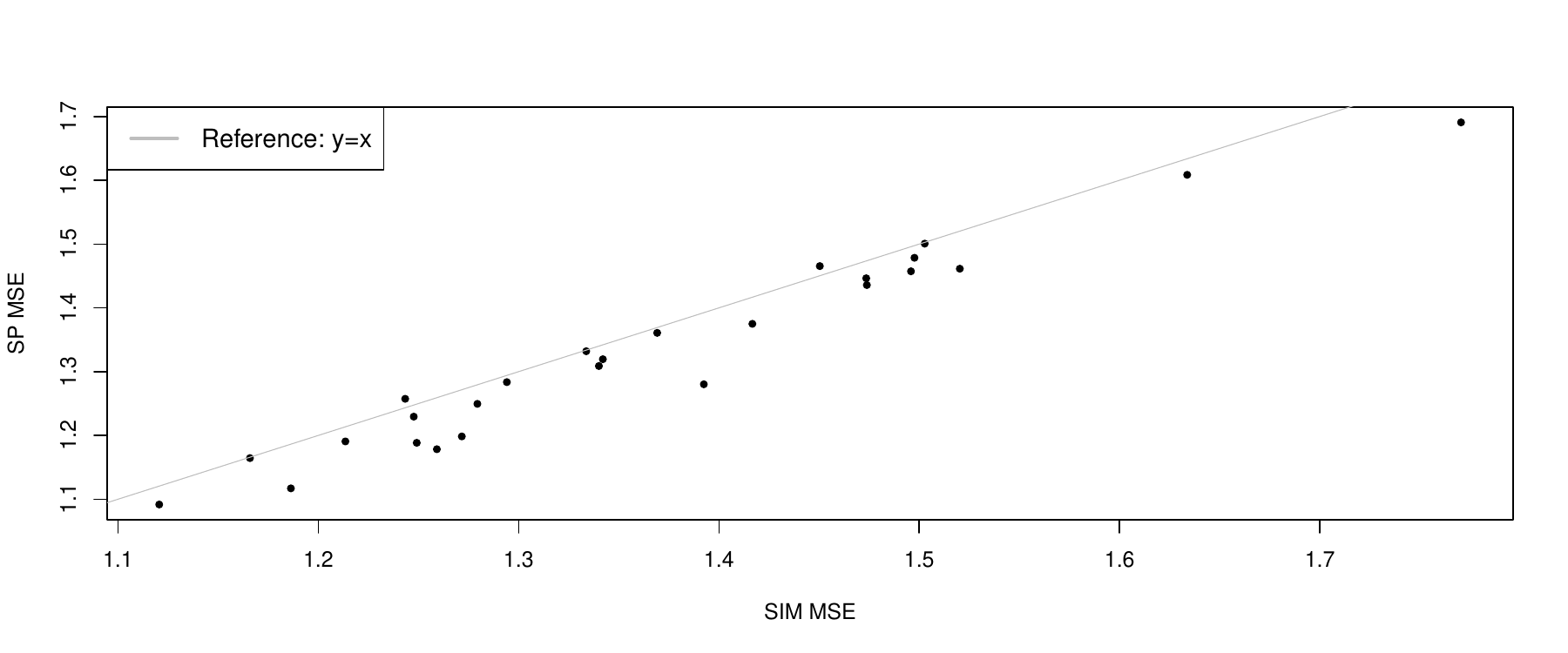}}
\subfigure[  Mean squared prediction error at each school: SP vs. RNC  ]{\includegraphics[width = 0.8\textwidth]{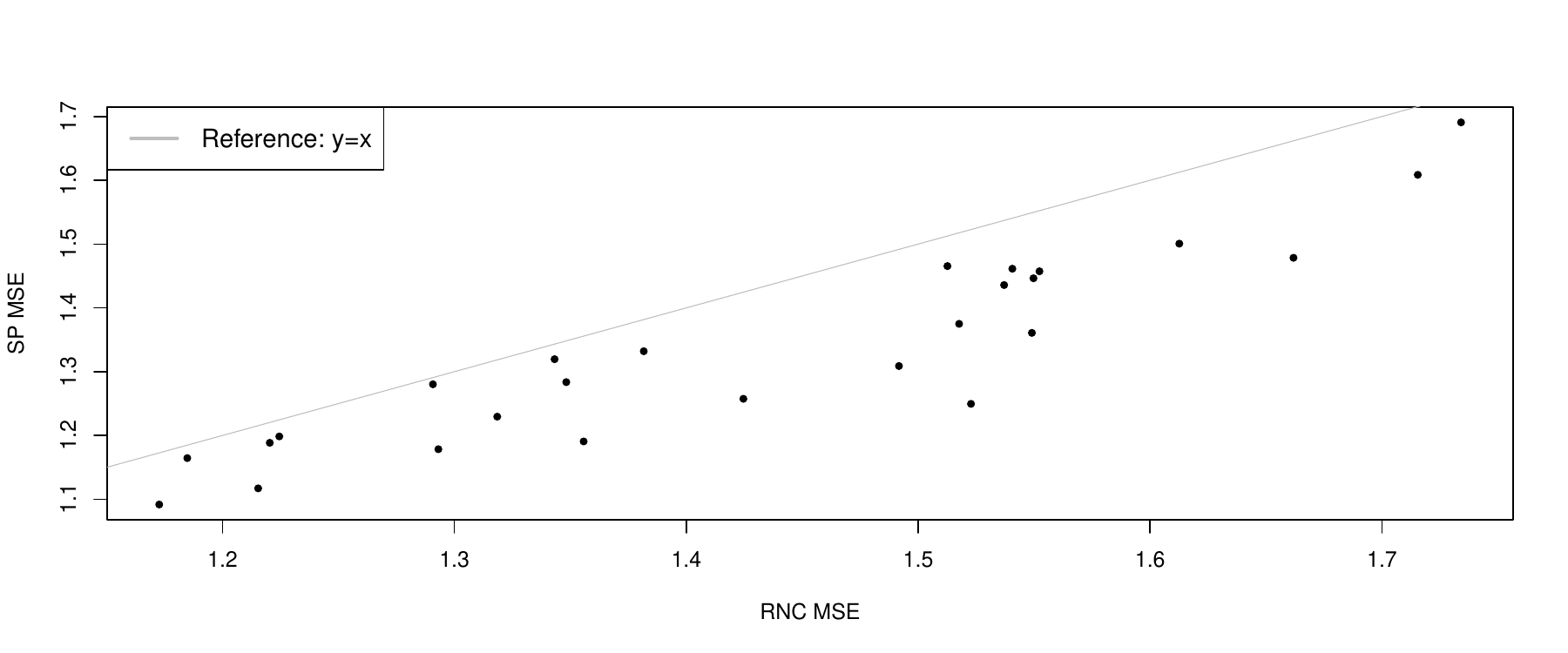}}
\caption{ Fitted $\alpha_i$'s in the SP method and the mean square prediction error comparison for each school.}
\end{center}
\label{fig:prediction_error}
\end{figure}

\paragraph{Predictive performance} We can also compare the methods by their predictive performance. The predictive performance is evaluated in a cross-validation manner. We randomly split the 9026 students into 200 folds. One fold of students' responses is held out each time, and we estimate the models using the remaining 199 folds of responses with the full set of predictors and the network. This procedure is repeated for all 200 folds. As schools exhibit significant variations in social effects, we calculate the mean squared prediction errors within each school $\sum_{i \in \text{sch}_k}(y_i - \hat{y}_i)^2/|\text{sch}_k|$. The comparisons between the SP method with the other three are shown in Figure~\ref{fig:prediction_error}.  Recall that when the $\alpha_i$'s are uniformly close to zero, our model would be similar to the standard linear regression, while when the  $\alpha_i$'s exhibit large deviations from zero, the SP and OLP tend to have very different results. This phenomenon is verified in the example. Figure~\ref{fig:prediction_error}-(b) shows that the SP and OLS have similar performance for most schools, but the SP has more accurate predictions for schools sch19, sch22, sch26 and sch56. This observation matches the fitted $\alpha$ values in Figure~\ref{fig:prediction_error}-(a), as these four schools have substantial variations in $\alpha$'s. In most of the other schools, the estimated $\alpha$'s have a smaller and more uniform magnitude and the difference between the two methods is small. Compared with the SIM and RNC, the SP renders better predictive performance in most schools. These results show that the SP model is more proper for the data than the others.

\subsection{Interpretation with refined models}\label{secsec:reduced-model}

In the full models of the previous section, many schools do not have significantly different effects from the reference level, and Table~\ref{tab:fullcovariates}  also suggests that many covariates are not significant. To better understand and interpret the data, we will resort to statistical inference for model selection. The RNC is not suitable for this model refinement procedure due to the lack of an inference framework. We will focus only on SP, OLS and SIM.

Starting from the full model, we first check the significance of the categorical variables, School (26 categories) and Race (5 categories). Bonferroni correction is applied to the multiple-level tests. Insignificant levels (adjusted p-value  $\ge$ 0.05) are merged into the reference level (School: sch1, Race: other). After this step, the model selection follows backward elimination \citep{halinski1970selection}. The variable with the largest p-value that exceeds 0.05 (after Bonferroni correction) is removed at each step until no further elimination is possible. Throughout this procedure, we always keep the treatment variable from elimination. For the SIM, before the backward elimination step, the elimination of the exogenous parameters is applied.  The final models from the three methods are given in Table~\ref{tab:reducedcovariates}. In the SP model, the $\chi^2$ test for the network effect gives a p-value smaller than $10^{-10}$, indicating strong evidence of network effects. However, the SIM model results in a p-value of 0.168, showing no firm evidence of network autoregressive correlation. The SIM model does have the local average of ``live with both parents" as a significant covariate effect. However, it fails to identify the impacts of the Race:black and the Returning Student, different from SP and OLS.

\begin{table}
\caption{\label{tab:reducedcovariates} Estimates and p-values of variables in the reduced models. ``--" indicates that the quantity is not available in the setting. A blank means the variable is excluded by model selection.}  
\centering
\fbox{
\begin{tabular}{r|ll|ll|ll}
  \hline
&\multicolumn{2}{c}{SP} & \multicolumn{2}{c}{OLS} &\multicolumn{2}{c}{SIM} \\ 
\hline
 & coef. & p-value & coef. & p-value & coef. & p-value  \\ 
  \hline
  (Intercept) & 4.1221 & $<10^{-4}$ & 4.3117 & $<10^{-4}$ & 3.4427 & $<10^{-4}$ \\ 
  Treatment & -0.0856 & 0.0411 & -0.0812 & 0.0993 & -0.079 & 0.1204 \\ 
  Race: black & -0.136 & 0.0033 & -0.1786 & 0.0003 &  &  \\  
  Grade & -0.1828 & $<10^{-4}$ & -0.1948 & $<10^{-4}$ &  &  \\ 
  Return & -0.1162 & 0.0004 & -0.0987 & 0.0016 & -0.1967 & 0.0001 \\ 
  Live w/ both parents & 0.0767 & 0.0038 &  &  & 0.0927 & 0.0019 \\ 
\hline
  School 3 & 0.5421 & $<10^{-4}$ & 0.4437 & 0.0001 & 0.3904 & 0.001 \\ 
  School 19 &  &  & -0.2901 & $<10^{-4}$ & -0.4269 & $<10^{-4}$ \\ 
  School 20 & 0.3503 & $<10^{-4}$ & 0.3281 & $<10^{-4}$ & 0.4512 & $<10^{-4}$ \\ 
  School 27 & -0.3371 & 0.0005 & -0.3954 & $<10^{-4}$ &  &  \\ 
  School 34 &  &  & -0.3017 & $<10^{-4}$ & &  \\ 
  School 35 & 0.8085 & $<10^{-4}$ & 0.5303 & $<10^{-4}$ &  &  \\ 
  School 40 & 0.4912 & $<10^{-4}$ & 0.3802 &0.0008 & 0.4674 & 0.0002\\ 
  School 42 & 0.4274 & $<10^{-4}$ & 0.3678 & $<10^{-4}$ &  &  \\ 
  School 48 &  &  & -0.3634 & $<10^{-4}$ &  &  \\ 
  School 49 &  &  & -0.3869 & $<10^{-4}$ &  &  \\ 
  School 51 & -0.5481 & $<10^{-4}$ & -0.5234 & $<10^{-4}$ & -0.5032 & $<10^{-4}$ \\ 
  School 56 & -0.3192 & 0.0002 & -0.5437 & $<10^{-4}$ & -0.606 & $<10^{-4}$ \\ 
\hline
  Local avg.: live w/ both parents  & -- & -- & -- & -- & 0.6465 & $<10^{-4}$ \\ 
   \hline
\end{tabular}}
\vspace{-0.5cm}
\end{table}

Comparing the OLS and SP, the high-level messages from the two models coincide in many aspects. Both models identify that black students tend to have more negative ratings of the school atmosphere, suggesting potential racial effects in school conflicts. Return students and students from higher grades also tend to have more negative perceptions. The two models also have differences in their conclusions. Based on the SP model, students living with both parents tend to have a more positive perception of the school atmosphere, while the OLS fails to detect this phenomenon. Though the ground truth about the data set is unknown, the SP’s finding agrees with a few previous social studies on experiments \citep{musick2010both,anderson2014impact}. When it comes to the treatment effect, as expected, all three methods give similar estimates but due to the more effective modeling, the SP delivers a smaller p-value of 0.04, suggesting potentially weak effects. Note that the treatment effect is negative. This is reasonable since the education workshops help introduce more information about school conflicts, so students may be aware of many conflicts they did not know about before. Our conclusion is implicitly supported by the original analysis of experimenters \citep{paluck2016changing}, who showed that students involved in the workshops were more likely to wear orange waistbands as a public sign that they stand against school conflicts. However, the analysis of \cite{paluck2016changing} considers the experiment design with additional assumptions about the spill-over effects of the educational workshops and is based on complicated causal inference methods. Therefore, their conclusion is about the stronger causal relation. In contrast, our method does not take advantage of the design, nor do we know whether the network cohesion assumption explicitly corresponds to the widely used spill-over effects assumptions. So our conclusion is not causal. We leave the development of causal inference under the current framework for future work.

\subsection{Robustness to network perturbation}\label{secsec:network-versions}

In the previous sections, we use the information of both nomination surveys to construct a weighted network for the regression analysis. In practice, researchers seldom know whether such a construction is the best one or if an optimal construction exists. It may also be more natural for some people to use the Wave II survey (or the Wave I survey) as the primary information. However, a valuable analysis that can reflect the true nature of the data should deliver consistent results as long as one uses some reasonable construction of the network. In this section, we examine how different ways of constructing the network would change the resulting models from the previous analysis. OLS does not use the network information, so it would not be affected. We will focus on comparing SP and the SIM.

The first alternative construction is the undirected network by only the Wave II survey. The same model fitting procedures of the previous section are applied, where model selection is done by back elimination with p-values and multiple comparison correction.  The fitted SP and SIM with the Wave II network are given in Table~\ref{tab:W2-reduced-covariates}. The SP gives a very small p-value ($<10^{-10}$) for the $\chi^2$ test, indicating strong network effects. Moreover, comparing Table~\ref{tab:W2-reduced-covariates} and Table~\ref{tab:reducedcovariates}, we can see that the high-level message of the SP model remains consistent: the Race:black, Grade, and Return Student variable have negative effects while ``live with parents" has a positive effect. The parameter estimates are slightly different but the changes are marginal. The school effects are also very similar, with the only difference on the school sch49. So overall, the change of the network construction method does not lead to material changes in the SP inference results. The SIM, in contrast, delivers very different messages from Table~\ref{tab:reducedcovariates} and Table~\ref{tab:W2-reduced-covariates}. For example, on the Wave II network, the SIM model identifies a strong network autoregressive correlation (p-value = 0.0003). However, the local average of ``live with parents" is no longer included in the model. Both Gender and Race:black are now identified as strong effects in the model. These are all different from the previous result.

\begin{table}
\caption{\label{tab:W2-reduced-covariates} Estimates and p-values of reduced models based on the Wave II network. ``--" indicates that the quantity is not available in the setting. A blank means the variable is excluded in the selection procedure.} 
\centering
\fbox{
\begin{tabular}{rrrrr}
  \hline
&\multicolumn{2}{c}{SP} &\multicolumn{2}{c}{SIM} \\ 
\hline
 & coef. & p-value & coef. & p-value  \\ 
  \hline
 \hline
(Intercept) & 4.0926 & $<10^{-4}$ & 4.9619 & $<10^{-4}$ \\ 
  Treatment & -0.0838 & 0.0439 & -0.0904 & 0.1673 \\ 
  Race: black  & -0.1362 & 0.0032 & -0.2084 & 0.0022 \\ 
  Grade& -0.1795 & $<10^{-4}$ &  &  \\ 
  Return & -0.0761 & 0.0175 & -0.1965 & 0.0005 \\ 
  Live with both parents & 0.0765 & 0.0038 & 0.1392 & 0.0004 \\ 
  Gender &  &  & 0.0850 & 0.0451 \\ 
\hline
  sch3 & 0.5213 & $<10^{-4}$ &  &  \\ 
  sch20 & 0.3621 & $<10^{-4}$ &  &  \\ 
  sch19 &  &  & -0.6291 & $<10^{-4}$ \\ 
  sch22 &  &  & -0.3637 & $<10^{-4}$ \\ 
  sch24 &  &  & -0.4440 & $<10^{-4}$ \\ 
  sch26 &  &  & -0.4029 & 0.0015 \\ 
  sch27 & -0.3483 & 0.0004 & -0.6313 & $<10^{-4}$ \\ 
  sch29 &  &  & -0.5946 & $<10^{-4}$ \\ 
  sch34 &  &  & -0.3376 & $<10^{-4}$ \\ 
  sch35 & 0.7134 & $<10^{-4}$ & 0.2783 & 0.0019 \\ 
  sch40 & 0.4733 & $<10^{-4}$ &  &  \\ 
  sch42 & 0.4121 & $<10^{-4}$ &  &  \\ 
  sch48 &  &  & -0.7667 & $<10^{-4}$ \\ 
  sch49 & -0.3490 & 0.0003 & -0.6358 & $<10^{-4}$ \\ 
  sch51 & -0.4448 & 0.0001 & -1.1455 & $<10^{-4}$ \\ 
  sch56 & -0.3561 & $<10^{-4}$ & -0.9281 & $<10^{-4}$ \\ 
  sch58 &  &  & -0.4807 & $<10^{-4}$ \\ 
  sch60 &  &  & -0.7074 & $<10^{-4}$ \\ 
   \hline
\end{tabular}}
\end{table}

Naturally, the second alternative network is the unweighted network based on the Wave I survey.   The model fitting results are shown in Table~\ref{tab:W1-reduced-covariates}. The conclusion remains the same. Though the estimates of the SP model change numerically, the significance and the selected model remain the same. The $\chi^2$ test indicates a strong network effect in all three cases. The SIM again selects different results compared to either Table~\ref{tab:reducedcovariates} or Table~\ref{tab:W2-reduced-covariates}. It also fails to identify the autoregressive correlation. In summary, it is evident that the SIM heavily relies on how one constructs the network. Since it is unclear which is the best way to construct the network, the SIM fails provide a reliable analysis in this case.

The above difference in robustness highlights the crucial advantage of our model over the SIM framework. Figure~\ref{fig:school40} displays the ``sch40" networks from the two surveys. The two networks have only about 60\% overlapping edges. So a statistical model focusing on local edges may deliver very different results, as reflected in the SIM case. Our model, in contrast, relies on the subspace assumption that is more robust to these changes. In other words, the eigenspaces remain stable under a certain amount of perturbations. Recall that the alignment of two subspaces is measured by the cosine values of their principal angles. A perfect alignment would result in cosine values of 1. Figure~\ref{fig:school40} shows the principle angle cosine values for the pairwise alignment between the 3-dimensional leading eigenspaces of Wave I network, combined (and weighted) Wave I+II network, and Wave II network. The eigenspace alignment is close to perfect ($>0.95$) for the leading two dimensions and remains high for the 3rd dimension. Therefore, even though the two networks are very different in their edge sets, the eigenspace remains stable.  In Figure~\ref{fig:school-other} of Appendix~\ref{sec:more_school_conflicts}, we include the principle angle plots for the other 25 schools in the data set. Such stable alignment of the eigenspace holds in most schools. This observation explains the robustness of our model.

%
%

\begin{table}
\caption{\label{tab:W1-reduced-covariates} Estimates and p-values of reduced models based on the Wave I network. ``--" indicates that the quantity is not available in the setting. A blank means the variable is excluded in the selection procedure.} 
\centering
\fbox{
\begin{tabular}{rrrrr}
  \hline
&\multicolumn{2}{c}{SP} &\multicolumn{2}{c}{SIM} \\ 
\hline
 & coef. & p-value & coef. & p-value  \\ 
  \hline
(Intercept) & 4.1216 & $<10^{-4}$ & 3.5797 & $<10^{-4}$ \\ 
  Treatment & -0.0860 & 0.0401 & -0.0691 & 0.1687 \\ 
  Race: black   & -0.1333 & 0.0037 & -0.1409 & 0.0063 \\ 
  Grade & -0.1797 & $<10^{-4}$ &  &  \\ 
  Return & -0.0966 & 0.0017 & -0.2237 & $<10^{-4}$ \\ 
  Live with both parents & 0.0868 & 0.0012 & 0.0928 & 0.0019 \\ 
\hline
 sch3 & 0.5135 & $<10^{-4}$ &  &  \\ 
  sch20 & 0.4070 & $<10^{-4}$ & 0.5619 & $<10^{-4}$ \\ 
  sch27 & -0.3530 & 0.0002 &  &  \\ 
  sch33 & 0.3088 & 0.0032 &  &  \\ 
  sch35 & 0.7650 & $<10^{-4}$ &  &  \\ 
  sch40 & 0.4672 & $<10^{-4}$ &  &  \\ 
  sch42 & 0.3836 & 0.0001 & 0.2750 & 0.0384 \\ 
  sch49 & -0.3465 & 0.0004 &  &  \\ 
  sch51 & -0.5468 & $<10^{-4}$ & -0.4526 & $<10^{-4}$ \\ 
  sch56 & -0.3861 & $<10^{-4}$ & -0.5451 & $<10^{-4}$ \\ 
  sch19 &  &  & -0.3956 & $<10^{-4}$ \\ 
  sch31 &  &  & 0.3844 & $<10^{-4}$ \\ 
\hline
 Local average: live with parents &  &  & 0.3617 & $<10^{-4}$ \\ 
    \hline
\end{tabular}}
\end{table}

\begin{figure}[h]
\centering
\includegraphics[width=1\textwidth]{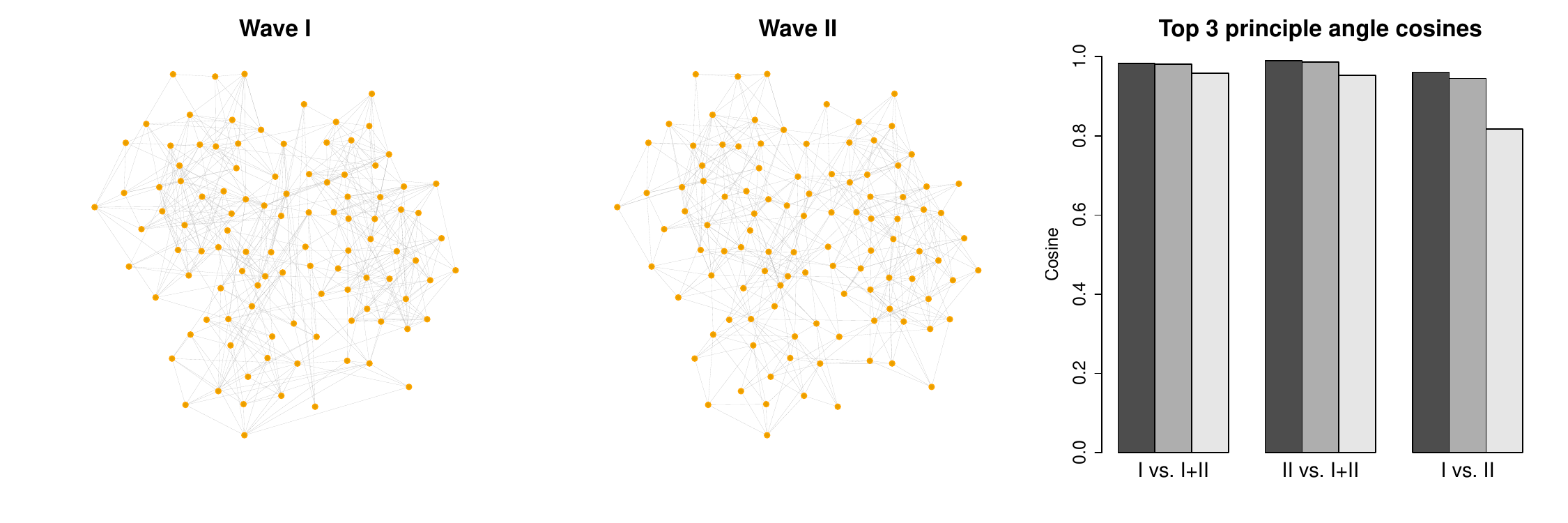}
\caption{Demonstration of network perturbation in two waves of surveys in one school. About 40\% edges in one network disappear in the other one.  However, the leading 3-dimensional eigenspaces of the networks still align well, indicated by the large cosine values of the principle angles. }
\label{fig:school40}
\end{figure}

\section{Conclusion}\label{sec:conclusion}
We have introduced a linear regression model on observations linked by a network. The model comes with a computationally efficient inference algorithm that can tolerate network observational errors. The study in this paper focuses on using Assumption~\ref{ass:small-perturbation} to control the estimation bias and deliver valid inference. It makes important progress toward the inference problem for the network-linked model.   

The next step is to explore whether a certain bias correction can be applied so that the small perturbation assumption can be further relaxed. Developing such a technique will require an accurate estimate of the perturbation $\hat{P} - P$ and, in turn, new tools for characterizing such random matrices. Meanwhile, the current focus is on the fixed design problem of $X$ and $P$;  the framework remains valid if we condition on $X$ and $P$ while assuming that they are independent. Exploring the model with dependence between $X$ and $P$ is another promising research direction. One such situation is when network evolution is observed over time.  In  \cite{goldsmith2013social} and  \cite{mcfowland2021estimating}, assuming the network is perfectly observed, the generative models between $X$ and $P$ lead to valid estimations of homophily effects.  It would be very interesting to study if embedding our subspace regression strategies in such settings would lead to robust inference framework of homophily.

\section*{Acknowledgement}
C. M. Le is supported in part by the NSF grant DMS-2015134. T. Li is supported in part by the NSF grant DMS-2015298 and the 3-Caverliers Award from the University of Virginia.

\bibliographystyle{rss}
\bibliography{CommonBib}{}

\newpage

\setcounter{page}{1}
\begin{appendix}

\begin{center}
\textbf{\Large Supplement to ``Linear regression and its inference on noisy network-linked data"}
\end{center}

\section{Extension to Laplacian individual effects}\label{sec:Laplacian} 
So far, we have treated $A$ as a perturbed version of $P$. A natural and popular alternative is to use the Laplacian matrix as the perturbed version of the relational information $P$ \citep{smola2003kernels,sadhanala2016graph,li2016prediction}. Specifically, we assume that $\alpha$ lies in the subspace determined by $K$ eigenvectors of $P = \e L = \e D - \e A$ corresponding to the smallest eigenvalues, while the estimation procedure is based on the perturbed version $\hat{P} = L=D-A$ of $P$, where $D$ is the diagonal matrix with node degrees $d_i$ on the diagonal. Similar to Assumption~\ref{ass:eigen-gap} about the spectrum of $\e A$, the following assumption about the spectrum of $\e L$ is made.  

\begin{ass}[Eigenvalue gap of the expected Laplacian]\label{ass:eigen-gap laplacian} 
Let $L=D-A$ be the Laplacian of a random network generated from the inhomogeneous Erd\H{o}s-R\'{e}nyi model and $P=\e L$. Denote by $\lambda_1\le\lambda_2\le \cdots \le\lambda_n$ the eigenvalues of $P$. Assume that the $K$ smallest eigenvalues of $P$
are well separated from the remaining eigenvalues and their range is not too large:
$$
\min_{i\le K, \ i'>K} |\lambda_i - \lambda_{i'}| \ge \rho' d, \qquad \max_{i, i'\le K} |\lambda_i - \lambda_{i'}| \le d/\rho',
$$
where $\rho'>0$ is a constant and $d =n\cdot \max_{ij}P_{ij}$.
\end{ass}

We show that, under the perturbation mechanism of Assumption~\ref{ass:eigen-gap laplacian}, the eigenvectors of the Laplacian corresponding to the $K$ smallest eigenvalues satisfy the small projection perturbation requirement, and therefore the inference framework based on  Algorithm~\ref{algo:estimation-most-general} with $\hat{P}=L$ remains valid.          

\begin{thm}[Concentration of perturbed projection for the Laplacian]\label{thm:projection-concentration-Laplacian}
Let ${w}_1,..., w_n$ and $\lambda_1\le\lambda_2\le \cdots \le\lambda_n$ be eigenvectors and corresponding eigenvalues of $\e L = \e D-\e A$ and similarly, let $\hat{w}_1,..., \hat{ w}_n$ and $\hat{\lambda}_1\le\hat{\lambda}_2\le\cdots\le \hat{\lambda}_n$ be the eigenvectors and eigenvalues of $L=D-A$. Denote $W = ({w}_1,...,{w}_K)$ and $\hat{W} = (\hat{{w}}_1,...,\hat{{w}}_K)$. Assume that Assumption~\ref{ass:eigen-gap laplacian} holds and $d \ge C\log{n}$ for a sufficiently large constant $C$. Then  for any fixed unit vector ${v}$, with high probability we have 
\begin{eqnarray*}
\|(\hat{W}\hat{W}^T - WW^T){v}\| \le  \frac{C\big[K\left(1+n\|W\|_{\infty}^2\right)\big]^{1/2}\log n}{d},
\end{eqnarray*}
where $\|W\|_{\infty}$ denotes the largest absolute value of entries in $W$. 
\end{thm}

Compared to the bound for the adjacency matrix in Theorem~\ref{thm:projection-concentration}, the bound in Theorem~\ref{thm:projection-concentration-Laplacian} contains a new term $n\|W\|_\infty^2$. When $\e L$ is incoherent, $\|W\|_\infty^2$ is of order $1/n$ (for example under stochastic block model) and therefore the upper bound is in the order of $\sqrt{K\log^2n}/d$. Both the term $n\|W\|_\infty^2$ and the extra $\sqrt{\log n}$ factor are results of the deviation of node degrees from their expected values.

\begin{coro}[Small projection perturbation, Laplacian case]\label{coro:small-perturbation-L}
Assume that $L$ is used as $\hat{P}$ in Algorithm~\ref{algo:estimation-most-general} and $P = \e L$. If Assumption~\ref{ass:eigen-gap laplacian} holds and $d$ is sufficiently large so that
\begin{eqnarray*}
\frac{d\cdot\min\left\{(1-\sigma_{r+1})^3,\sigma_{r+s}^3\right\}}{\left[Knp\big(1+n\|W\|_{\infty}^2\right)\big]^{1/2}\log n} \rightarrow \infty,
\end{eqnarray*}
then Assumption~\ref{ass:small-perturbation} holds with high probability. In particular, the proposed inference framework in Section~\ref{secsec:generic-inference} is valid with high probability.
\end{coro}

\section{Basic tools and properties for generic inference}

We start with the following lemma, which among other things allows us to rewrite $\mathcal{P}_{C}$ and $\mathcal{P}_{N}$ in a more convenient way for theoretical analysis. 

\begin{lem}[Properties of population projections]\label{lem:unorth proj}
Let $\mathcal{P}_{\ccal}$ and $\mathcal{P}_{\ncal}$ be the matrices defined by \eqref{eq:PC} and \eqref{eq:PN}, respectively. 
Then the following statements hold.
\begin{enumerate}[label=(\roman*)]
\item $\mathcal{P}_{\ccal}+\mathcal{P}_{\ncal} = M(M^TM)^{-1}M^T$ is the orthogonal projection onto the column space of $M$. 
\item $\mathcal{P}_{\ccal}$ and $\mathcal{P}_{\ncal}$ can be written as
\begin{eqnarray}
\label{eq:PC alt}\mathcal{P}_{C} &=& \sum_{i=r+1}^{r+s} \frac{1}{1-\sigma_i^2} \tilde{Z}_{i} \tilde{Z}^T_{i} - \sum_{i=r+1}^{r+s} \frac{\sigma_i}{1-\sigma_i^2} \tilde{Z}_{i} \tilde{W}^T_{i} + \sum_{i=r+s+1}^{p} \tilde{Z}_{i} \tilde{Z}^T_{i},\\
\label{eq:PN alt}\mathcal{P}_{N} &=& \sum_{i=r+1}^{r+s} \frac{1}{1-\sigma_i^2} \tilde{W}_{i} \tilde{W}^T_{i} - \sum_{i=r+1}^{r+s} \frac{\sigma_i}{1-\sigma_i^2} \tilde{W}_{i} \tilde{Z}^T_{i} + \sum_{i=r+s+1}^{K} \tilde{W}_{i} \tilde{W}^T_{i}.
\end{eqnarray}
\item For any $x\in \text{col}(\tilde{Z}_{(r+1):p})$ and $y\in \text{col}(\tilde{W}_{(r+1):K})$, we have
\begin{eqnarray*}
\mathcal{P}_{\ccal}(x+y) = x, \quad \mathcal{P}_{\ncal}(x+y) = y.
\end{eqnarray*}
\item For any $x\in\mathbb{R}^n$, we have $\|\mathcal{P}_{\ccal}x\|\ge \|\tilde{Z}_{(r+1):p}^Tx\|$.
\end{enumerate}
\end{lem}

A crucial step in proving our results is to bound the differences $\hat{\mathcal{P}}_C-\mathcal{P}_{\ccal}$ and $\hat{\mathcal{P}}_N-\mathcal{P}_{\ncal}$; see \eqref{eq:PChat}, \eqref{eq:PNhat}, \eqref{eq:PC} and \eqref{eq:PN} for the definition of these matrices. Since they depend
on $M$ and $\hat{M}$, it is tempting to bound $\hat{M}-M$ first and then use that bound to show that
$\hat{\mathcal{P}}_C-\mathcal{P}_{\ccal}$ and $\hat{\mathcal{P}}_N-\mathcal{P}_{\ncal}$ are small in magnitude. Unfortunately, $\hat{M}-M$ may not be small because in general singular vectors are unique only up to multiplying by a certain orthogonal matrix. Nevertheless, we will still be able to show that $\hat{\mathcal{P}}_C-\mathcal{P}_{\ccal}$ and $\hat{\mathcal{P}}_N-\mathcal{P}_{\ncal}$ are small
by leveraging the projection-type properties of $\hat{\mathcal{P}}_C$, $\mathcal{P}_{\ccal}$, $\hat{\mathcal{P}}_N$ and $\mathcal{P}_{\ncal}$, although the fact that
they may not be orthogonal projections makes the analysis much more involved. 

In light of
Lemma~\ref{lem:unorth proj}, it is convenient to introduce the following notations:
\begin{eqnarray}
\label{eq:PCtilde}\tilde{\mathcal{P}}_C
&=& \sum_{i=r+1}^{r+s} \frac{1}{1-\hat{\sigma}_i^2} \hat{Z}_{i} \hat{Z}^T_{i} - \sum_{i=r+1}^{r+s} \frac{\hat{\sigma}_i}{1-\hat{\sigma}_i^2} \hat{Z}_{i} \breve{W}^T_{i} + \sum_{i=r+s+1}^{p} \hat{Z}_{i} \hat{Z}^T_{i},\\
\label{eq:PNtilde}\tilde{\mathcal{P}}_N&=& \sum_{i=r+1}^{r+s} \frac{1}{1-\hat{\sigma}_i^2} \breve{W}_{i} \breve{W}^T_{i} - \sum_{i=r+1}^{r+s} \frac{\hat{\sigma}_i}{1-\hat{\sigma}_i^2} \breve{W}_{i} \hat{Z}^T_{i} + \sum_{i=r+s+1}^{K} \breve{W}_{i} \breve{W}^T_{i}.
\end{eqnarray}
The following lemma bound $\hat{\mathcal{P}}_C-\tilde{\mathcal{P}}_C$ and $\hat{\mathcal{P}}_N-\tilde{\mathcal{P}}_N$.     

\begin{lem}[Pertubation of projections, first step]\label{lem:comparing Phat and Ptilde}
If condition \eqref{eq: singular value gap condition} holds then
\begin{eqnarray*}
\max\{\|\hat{\mathcal{P}}_C-\tilde{\mathcal{P}}_C\|,\|\hat{\mathcal{P}}_N-\tilde{\mathcal{P}}_N\|\} \le  \frac{16\tau_n}{(1-\sigma_{r+1})^2}.
\end{eqnarray*} 
\end{lem}

We will bound the differences $\hat{\mathcal{P}}_C-\mathcal{P}_{\ccal}$ and $\hat{\mathcal{P}}_N-\mathcal{P}_{\ncal}$ by bounding $\tilde{\mathcal{P}}_C-\mathcal{P}_{\ccal}$ and $\tilde{\mathcal{P}}_N-\mathcal{P}_{\ncal}$, and then combining them with Lemma~\ref{lem:comparing Phat and Ptilde}. For that purpose, denote
\begin{eqnarray}
\delta_1 &=& \sum_{i=r+1}^{r+s} \frac{1}{1-\hat{\sigma}_i^2} \hat{Z}_{i} \hat{Z}^T_{i} -\sum_{i=r+1}^{r+s} \frac{1}{1-\sigma_i^2} \tilde{Z}_{i} \tilde{Z}^T_{i},\\
\delta_2 &=&\sum_{i=r+1}^{r+s} \frac{\hat{\sigma}_i}{1-\hat{\sigma}_i^2} \hat{Z}_{i} \breve{W}^T_{i} - \sum_{i=r+1}^{r+s} \frac{\sigma_i}{1-\sigma_i^2} \tilde{Z}_{i} \tilde{W}^T_{i},\\ \delta_3 &=&  \sum_{i=r+s+1}^{p} \hat{Z}_{i} \hat{Z}^T_{i} - \sum_{i=r+s+1}^{p} \tilde{Z}_{i} \tilde{Z}^T_{i},\\
\delta_4 &=& \sum_{i=r+1}^{r+s} \frac{1}{1-\hat{\sigma}_i^2} \breve{W}_{i} \breve{W}^T_{i} - \sum_{i=r+1}^{r+s} \frac{1}{1-\sigma_i^2} \tilde{W}_{i} \tilde{W}^T_{i},\\
\delta_5 &=& \sum_{i=r+1}^{r+s} \frac{\hat{\sigma}_i}{1-\hat{\sigma}_i^2} \breve{W}_{i} \hat{Z}^T_{i} - \sum_{i=r+1}^{r+s} \frac{\sigma_i}{1-\sigma_i^2} \tilde{W}_{i} \tilde{Z}^T_{i},\\
\delta_6 &=& \sum_{i=r+s+1}^{K} \breve{W}_{i} \breve{W}^T_{i} - \sum_{i=r+s+1}^{K} \tilde{W}_{i} \tilde{W}^T_{i}.
\end{eqnarray}
By Lemma~\ref{lem:unorth proj}, $\tilde{\mathcal{P}}_C-\mathcal{P}_{\ccal} = \delta_1-\delta_2+\delta_3$ and $\tilde{\mathcal{P}}_N-\mathcal{P}_{\ncal} = \delta_4-\delta_5+\delta_6$. The following lemma provides upper bounds for all $\|\delta_i\|$.

\begin{lem}[Perturbation of projections, second step]\label{lem:bounding perturbation unorthogonal proj}
If condition \eqref{eq: singular value gap condition} holds then 
\begin{eqnarray*}
\max_{1\le i\le 6}\|\delta_i\| \le \frac{60\tau_n}{\min\{(1-\sigma_{r+1})^3,\sigma_{r+s}^3\}}.
\end{eqnarray*}
\end{lem}

The proofs of Lemmas~\ref{lem:unorth proj}, \ref{lem:comparing Phat and Ptilde} and \ref{lem:bounding perturbation unorthogonal proj} are postponed to later in this section. As a direct consequence of these results and the triangle inequality, we have the following bound. 

\begin{coro}[Pertubation of projections]\label{cor:pertubation of projections}
If condition \eqref{eq: singular value gap condition} holds then there exists a constant $C>0$ such that
\begin{eqnarray*}
\max\{\|\hat{\mathcal{P}}_C-{\mathcal{P}}_C\|,\|\hat{\mathcal{P}}_N-{\mathcal{P}}_N\|\} \le \frac{C\tau_n}{\min\{(1-\sigma_{r+1})^3,\sigma_{r+s}^3\}}.
\end{eqnarray*}
\end{coro}

We are now ready to prove Proposition~\ref{prop:beta-bias} about the bias of $\hat{\beta}$. 

\begin{proof}[Proof of Proposition~\ref{prop:beta-bias}]
Recall the estimate $\hat{\beta}$ in \eqref{eq:alpha beta theta hat}, identity of $\beta$ in \eqref{eq:alpha beta proj form}, and the model $Y=X\beta+X\theta+\alpha+\epsilon$ from \eqref{eq:RNC}.   
Taking the expectation with respect to the randomness of the noise $\epsilon$, we get
\begin{eqnarray}\label{eq:expected alpha hat beta hat}
\mathbb{E}\hat{\beta} = (X^TX)^{-1}X^T \hat{\mathcal{P}}_C(X\beta+X\theta+\alpha).  
\end{eqnarray}  
By \eqref{eq:alpha beta proj form}, \eqref{eq:expected alpha hat beta hat} and Corollary~\ref{cor:pertubation of projections}, we have
\begin{eqnarray*}
\|\mathbb{E}\hat{\beta}-\beta\| &=& \|(X^TX)^{-1}X^T(\hat{\mathcal{P}}_C-\mathcal{P}_{\ccal})(X\beta+X\theta+\alpha)\|\\
&\le& \|(X^TX)^{-1}\|\cdot\|X\|\cdot\|X\beta+X\theta+\alpha\|\cdot\|\hat{\mathcal{P}}_C-\mathcal{P}_{\ccal}\|\\
&\le& \|(X^TX)^{-1}\|\cdot\|X\|\cdot\|X\beta+X\theta+\alpha\|\cdot\frac{C\tau_n}{\left(\min\{1-\sigma_{r+1},\sigma_{r+s}\}\right)^3}.
\end{eqnarray*}
The proof is complete.
\end{proof}

Next, we prove Proposition~\ref{prop: bias variance of gamma hat} about the consistency of $\hat{\alpha}$ and the bias and variance of $\hat{\gamma}$.

\begin{proof}[Proof of Proposition~\ref{prop: bias variance of gamma hat}]

\noindent{\bf Bounding $\|\hat{\alpha}-\alpha\|$.}  From \eqref{eq:alpha beta proj form}, \eqref{eq:alpha beta theta hat} and \eqref{eq:RNC} we have 
$$\hat{\alpha}-\alpha = (\hat{\mathcal{P}}_N-\mathcal{P}_{\ncal})(X\beta+X\theta+\alpha+\epsilon)+\mathcal{P}_{\ncal}\epsilon.$$ 
By Assumption~\ref{ass:signal-scale} and the fact that $\|\epsilon\| = O(\sqrt{n\sigma^2})$ with high probability,
$$
\|X\beta+X\theta+\alpha+\epsilon\| \le C\sqrt{n(p+\sigma^2)}.
$$
Using the triangle inequality and Corollary~\ref{cor:pertubation of projections}, we get
\begin{eqnarray}
\nonumber\|\hat{\alpha}-\alpha\| &\le& \big\|\hat{\mathcal{P}}_N-\mathcal{P}_{\ncal}\big\|\cdot\|X\beta+X\theta+\alpha+\epsilon\|+\big\|{\mathcal{P}}_N\epsilon\big\|\\
\label{eq:alpha-alphahat}&\le& \frac{C\tau_n\sqrt{n(p+\sigma^2)}}{\min\{(1-\sigma_{r+1})^3,\sigma_{r+s}^3\}} +\big\|{\mathcal{P}}_N\epsilon\big\|.
\end{eqnarray}
To bound $\big\|{\mathcal{P}}_N\epsilon\big\|$, we first use \eqref{eq:tilde Z W} to rewrite $\tilde{W}_i$ as $\tilde{W}_i = WV_i$ for $i\le K$. Then from \eqref{eq:PN}, we can rewrite $\mathcal{P}_{\ncal}$ as follows:   
\begin{eqnarray*}
\mathcal{P}_{N} &=& W\left(\sum_{i=r+1}^{r+s} \frac{1}{1-\sigma_i^2} V_{i} \tilde{W}^T_{i} - \sum_{i=r+1}^{r+s} \frac{\sigma_i}{1-\sigma_i^2} V_{i} \tilde{Z}^T_{i} + \sum_{i=r+s+1}^{K} V_{i} \tilde{W}^T_{i}\right) = :WQ,
\end{eqnarray*}
where $Q$ is a fixed matrix of rank at most $2K$. Moreover, for any unit vector $x$,
\begin{eqnarray*}
\|Qx\| &\le& \left\|\sum_{i=r+1}^{r+s} \frac{1}{1-\sigma_i^2} V_{i} \tilde{W}^T_{i}x + \sum_{i=r+s+1}^{K} V_{i} \tilde{W}^T_{i}x\right\| + \left\|\sum_{i=r+1}^{r+s} \frac{\sigma_i}{1-\sigma_i^2} V_{i} \tilde{Z}^T_{i}x \right\|\\
&=&\left(\sum_{i=r+1}^{r+s} \frac{(\tilde{W}^T_{i}x)^2}{(1-\sigma_i^2)^2}  + \sum_{i=r+s+1}^{K} (\tilde{W}^T_{i}x)^2\right)^{1/2}+\left(\sum_{i=r+1}^{r+s} \frac{\sigma_i^2(\tilde{Z}^T_{i}x)^2}{(1-\sigma_i^2)^2} \right)^{1/2}\\
&\le& \frac{\|\tilde{W}^Tx\|}{1-\sigma_{r+1}^2}+\frac{\|\tilde{Z}x\|}{1-\sigma_{r+1}^2}\\
&\le& \frac{2}{1-\sigma_{r+1}^2}.
\end{eqnarray*}
Thus, $\|Q\|\le 2/(1-\sigma_{r+1}^2)$ and $Q\epsilon$ is a Gaussian vector of dimension at most $2K$ and covariance matrix with spectral norm at most $4\sigma^2/(1-\sigma_{r+1}^2)^2$. Therefore with high probability 
$$\|Q\epsilon\| \le \frac{C\sqrt{K\sigma^2}\cdot\log n}{1-\sigma_{r+1}^2}.$$  
In particular,
\begin{eqnarray*}
\|\mathcal{P}_{\ncal}\epsilon\| \le \|W\|_{2\to\infty}\cdot\|Q\epsilon\| \le \|W\|_{2\to\infty}\cdot\frac{C\sqrt{K\sigma^2}\cdot\log n}{1-\sigma_{r+1}^2}. 
\end{eqnarray*}
Together with \eqref{eq:alpha-alphahat}, this implies the required bound for $\|\hat{\alpha}-\alpha\|$.

\noindent{\bf Bounding the bias of $\hat{\gamma}$.} By \eqref{eq:alpha beta proj form}, 
$$
\gamma = \tilde{W}_{(r+1):K}^T\alpha = \tilde{W}_{(r+1):K}^T\mathcal{P}_{\ncal}(X\beta+X\theta+\alpha).
$$
Similarly, by \eqref{eq:gamma hat}, \eqref{eq:RNC} and taking the expectation with respect to the randomness of the noise $\epsilon$, we have
$$
\mathbb{E}\hat{\gamma} = \breve{W}_{(r+1):K} \hat{\mathcal{P}}_N\mathbb{E} Y = \breve{W}_{(r+1):K} \hat{\mathcal{P}}_N (X\beta+X\theta+\alpha).
$$
Let $O = \breve{W}^T_{(r+1):K}\tilde{W}_{(r+1):K}$, then 
\begin{equation}\label{eq:bias gammahat}
\mathbb{E}\hat{\gamma}-O\gamma = \left(\breve{W}^T_{(r+1):K}\hat{\mathcal{P}}_N-O\tilde{W}_{(r+1):K}^T\mathcal{P}_{\ncal}\right)(X\beta+X\theta +\alpha).
\end{equation}
To bound this bias, we first replace $\hat{\mathcal{P}}_N$ with $\mathcal{P}_{\ncal}$ and then estimate the resulting expression. By Corollary~\ref{cor:pertubation of projections}, there exists a constant $C>0$ such that
\begin{equation*}
\|\hat{\mathcal{P}}_N-\mathcal{P}_{\ncal}\| \le \frac{C\tau_n}{\min\{(1-\sigma_{r+1})^3,\sigma_{r+s}^3\}}.
\end{equation*}
Together with Assumption~\ref{ass:signal-scale}, this implies
\begin{eqnarray*}
\big\|\breve{W}^T_{(r+1):K}\big(\hat{\mathcal{P}}_N-\mathcal{P}_{\ncal}\big)(X\beta+X\theta +\alpha)\big\| &\le& \frac{C\tau_n\sqrt{np}}{\min\{(1-\sigma_{r+1})^3,\sigma_{r+s}^3\}}.
\end{eqnarray*}
Therefore from \eqref{eq:bias gammahat}, \eqref{eq:alpha beta proj form} and the triangle inequality,
\begin{eqnarray*}
\|\mathbb{E}\hat{\gamma}-O\gamma\|  &\le& \Big\|\left(\breve{W}^T_{(r+1):K}-O\tilde{W}_{(r+1):K}^T\right)P_N(X\beta+X\theta +\alpha)\Big\|+\frac{C\tau_n\sqrt{np}}{\min\{(1-\sigma_{r+1})^3,\sigma_{r+s}^3\}}\\
&=& \Big\| \left(\breve{W}^T_{(r+1):K}-O\tilde{W}_{(r+1):K}^T\right)\alpha\Big\|+\frac{C\tau_n\sqrt{np}}{\min\{(1-\sigma_{r+1})^3,\sigma_{r+s}^3\}}\\
&=&\Big\|\breve{W}^T_{(r+1):K}\left(I - \tilde{W}_{(r+1):K}\tilde{W}_{(r+1):K}^T\right)\alpha\Big\|+\frac{C\tau_n\sqrt{np}}{\min\{(1-\sigma_{r+1})^3,\sigma_{r+s}^3\}}\\
&=& \frac{C\tau_n\sqrt{np}}{\min\{(1-\sigma_{r+1})^3,\sigma_{r+s}^3\}}.
\end{eqnarray*}
The last equality holds because $\alpha\in\text{col}(\tilde{W}_{(r+1):K})$ by \eqref{eq:identifiability in detail} and $\tilde{W}_{(r+1):K}\tilde{W}_{(r+1):K}^T$ is the orthogonal projection onto $\text{col}(\tilde{W}_{(r+1):K})$.

\noindent{\bf Bounding the covariance of $\hat{\gamma}$.} Recall $\hat{\gamma} = \breve{W}_{(\hat{r}+1):K} \hat{P}_NY$ from \eqref{eq:gamma hat}.
Since $Y=X\beta+X\theta +\alpha+\epsilon$, it follows that $$\Sigma_{\hat{\gamma}} =  \sigma^2 \breve{W}^T_{(r+1):K}\hat{\mathcal{P}}_N \hat{\mathcal{P}}_N^T\breve{W}_{(r+1):K}.$$
Define $\Gamma = \diag(\sigma_{r+1}, \cdots, \sigma_{r+s})$ and $\hat{\Gamma} = \diag(\hat{\sigma}_{r+1}, \cdots, \hat{\sigma}_{r+s})$. A direct calculation shows that
\begin{eqnarray*}
\breve{W}_{(r+1):K}^T\hat{\mathcal{P}}_N = 
\begin{pmatrix}
(1-\hat{\Gamma}^2)^{-1}\breve{W}^T_{(r+1):(r+s)}-\hat{\Gamma}(1-\hat{\Gamma}^2)^{-1}\hat{Z}^T_{(r+1):(r+s)}\\
\breve{W}^T_{(r+s+1):K}
\end{pmatrix}.
\end{eqnarray*}
Therefore
\begin{eqnarray*}
\Sigma_{\hat{\gamma}} = \sigma^2\cdot
\begin{pmatrix}
(1-\hat{\Gamma}^2)^{-1} & \hat{\Gamma}(1-\hat{\Gamma}^2)^{-1}\breve{Z}^T_{(r+1):(r+s)}\breve{W}_{(r+s+1):K}\\
\breve{W}_{(r+s+1):K}^T\breve{Z}_{(r+1):(r+s)}\hat{\Gamma}(1-\hat{\Gamma}^2)^{-1} & I_{K-r-s} 
\end{pmatrix}.
\end{eqnarray*}
Since $Z^T\hat{W}=\hat{U}\hat{\Sigma}\hat{V}$ by \eqref{eq:svd sample}, it follows from \eqref{eq:Z hat W breve} that $\hat{Z}^T\breve{W}=\hat{\Sigma}$, and therefore 
$$\hat{Z}^T_{(r+1):(r+s)}\breve{W}_{\cdot(r+s+1):K}=0$$ because the left-hand side is an off-diagonal block of $\hat{\Sigma}$. Moreover, by  \eqref{eq: singular value gap condition} and inequality \eqref{eq: Inverse Gamma} of Lemma~\ref{lem: projection bound},
\begin{equation*}
\|(1-\hat{\Gamma}^2)^{-1}-(1-\Gamma^2)^{-1}\| \le \frac{4\tau_n}{\min\{(1-\sigma_{r+1})^2,\sigma_{r+s}^2\}},
\end{equation*}
and the bound for $\Sigma_{\hat{\gamma}}$ follows.
\end{proof}

\begin{proof}[Proof of Lemma~\ref{lem:unorth proj}] 
Since 
$$M =  \big(\tilde{Z}_{(r+1):p},0_{n\times (K-r)}\big)+\big(0_{n\times (p-r)},\tilde{W}_{(r+1):K}\big),$$ 
it follows directly from \eqref{eq:PC} and \eqref{eq:PN} that $\mathcal{P}_{\ccal}+\mathcal{P}_{\ncal} = M(M^TM)^{-1}M$, and (i) is proved.

Next, we verify \eqref{eq:PC alt} and \eqref{eq:PN alt} in (ii).  
Recall from \eqref{eq:subspace angles} that
$Z^TW = U\Sigma V^T$ and the nonzero singular values of $Z^TW$ are
$$\sigma_1 = \sigma_2= \cdots =\sigma_r =1 >\sigma_{r+1} \ge \cdots \ge \sigma_{r+s}>0.$$
Multiplying both sides of $Z^TW = U\Sigma V^T$ with $U^T$ on the left and $V$ on the right, and then using \eqref{eq:tilde Z W}, we get
$$
\tilde{Z}^T \tilde{W} = \Sigma = 
\begin{pmatrix}
I_r & 0_{r\times s} & 0_{r\times (K-r-s)} \\
0_{s\times r} & \Gamma & 0_{s\times (K-r-s)} \\
0_{(p-r-s)\times r} & 0_{(p-r-s)\times s}& 0_{(p-r-s)\times (K-r-s)}
\end{pmatrix}, 
$$
where
\begin{equation}\label{eq:Gamma def}
\Gamma = \text{diag}(\sigma_{r+1},...,\sigma_{r+s}).
\end{equation}
It follows that
\begin{equation}\label{eq:Z^TW}
\tilde{Z}_{(r+1):p}^T \tilde{W}_{(r+1):K} = 
\begin{pmatrix}
\Gamma & 0_{s\times (K-r-s)} \\
0_{(p-r-s)\times s}& 0_{(p-r-s)\times (K-r-s)}
\end{pmatrix}.
\end{equation}
Since $M=\big(\tilde{Z}_{(r+1):p}, \tilde{W}_{(r+1):K}\big)$, the above identity implies
\begin{eqnarray}
\label{eq:MTM} M^T M &=& 
\begin{pmatrix}
I_s & 0 & \Gamma & 0 \\
0& I_{p-r-s} &0 & 0\\
\Gamma & 0& I_s & 0\\
0 & 0 &0 & I_{K-r-s} 
\end{pmatrix}, \\ 
\label{eq:MTMinverse}
\left(M^T M\right)^{-1} &=& 
\begin{pmatrix}
(I-\Gamma^2)^{-1} & 0 & -(I-\Gamma^2)^{-1}\Gamma & 0 \\
0& I_{p-r-s} & 0 & 0 \\
-(I-\Gamma^2)^{-1}\Gamma &0& (I-\Gamma^2)^{-1} & 0\\ 
0 & 0 & 0& I_{K-r-s} 
\end{pmatrix}.
\end{eqnarray}
Therefore 
\begin{eqnarray*}\label{eq: Z projection}
\nonumber \left(M^T M\right)^{-1} M^T &=& 
\begin{pmatrix}
(I-\Gamma^2)^{-1} & 0 & -(I-\Gamma^2)^{-1}\Gamma & 0 \\
0& I_{p-r-s} & 0 & 0 \\
-(I-\Gamma^2)^{-1}\Gamma &0& (I-\Gamma^2)^{-1} & 0\\ 
0 & 0 & 0& I_{K-r-s} 
\end{pmatrix}
\begin{pmatrix}
\tilde{Z}^T_{(r+1):(r+s)} \\
\tilde{Z}^T_{(r+s+1):p}\\
\tilde{W}_{(r+1):(r+s)}^T\\
\tilde{W}_{(r+s+1):K}^T
\end{pmatrix}\\
&=& 
\nonumber 
\begin{pmatrix}
(I-\Gamma^2)^{-1}\tilde{Z}^T_{(r+1):(r+s)} -(I-\Gamma^2)^{-1}\Gamma \tilde{W}_{(r+1):(r+s)}^T\\
\tilde{Z}^T_{(r+s+1):p}\\
-(I-\Gamma^2)^{-1}\Gamma\tilde{Z}^T_{(r+1):(r+s)} +(I-\Gamma^2)^{-1}\tilde{W}_{(r+1):(r+s)}^T\\
\tilde{W}_{(r+s+1):K}^T
\end{pmatrix}.
\end{eqnarray*}
Multiplying this matrix on the left by either $\big(\tilde{Z}_{(r+1):p},0_{n\times (K-r)}\big)$ or $\big(0_{n\times (p-r)},\tilde{W}_{(r+1):K}\big)$, we see that \eqref{eq:PC alt} and \eqref{eq:PN alt} hold and (ii) is proved. 

For (iii), we will show that $\mathcal{P}_{\ccal}(x+y) = x$; the identity $\mathcal{P}_{\ncal}(x+y) = y$ can be proved by using the same argument. Since $x\in \text{col}(\tilde{Z}_{(r+1):p})$ and $y\in \text{col}(\tilde{W}_{(r+1):K})$, it is enough to show that the identity holds if $x$ is any column vector of $\tilde{Z}_{(r+1):p}$ and $y$ is any column vector of $\tilde{W}_{(r+1):K}$. Therefore, assume that $x = \tilde{Z}_{i'}$ and $y = \tilde{W}_{j'}$ for some $r+1\le i'\le p$ and $r+1\le j'\le K$. If $r+1\le i'\le r+s$ then by \eqref{eq:Z^TW},
\begin{eqnarray*}
\mathcal{P}_{C}x &=& \sum_{i=r+1}^{r+s} \frac{1}{1-\sigma_i^2} \tilde{Z}_{i} \tilde{Z}^T_{i}\tilde{Z}_{i'} - \sum_{i=r+1}^{r+s} \frac{\sigma_i}{1-\sigma_i^2} \tilde{Z}_{i} \tilde{W}^T_{i}\tilde{Z}_{i'} + \sum_{i=r+s+1}^{p} \tilde{Z}_{i} \tilde{Z}^T_{i}\tilde{Z}_{i'}\\
&=&\frac{1}{1-\sigma_{i'}^2}\tilde{Z}_{i'}-\frac{\sigma_{i'}^2}{1-\sigma_{i'}^2} \tilde{Z}_{i'} \\
&=& x. 
\end{eqnarray*}  
If $i'\ge r+s+1$ then again by \eqref{eq:Z^TW}, 
\begin{eqnarray*}
\mathcal{P}_{C}x = \sum_{i=r+1}^{r+s} \frac{1}{1-\sigma_i^2} \tilde{Z}_{i} \tilde{Z}^T_{i}\tilde{Z}_{i'} - \sum_{i=r+1}^{r+s} \frac{\sigma_i}{1-\sigma_i^2} \tilde{Z}_{i} \tilde{W}^T_{i}\tilde{Z}_{i'} + \sum_{i=r+s+1}^{p} \tilde{Z}_{i} \tilde{Z}^T_{i}\tilde{Z}_{i'} =\tilde{Z}_{i'} = x. 
\end{eqnarray*} 
We have shown that $\mathcal{P}_{\ccal}x=x$ for all $x\in \text{col}(\tilde{Z}_{(r+1):p})$. Now, if $y = \tilde{W}_{j'}$ and $r+1\le j'\le r+s$ then by \eqref{eq:Z^TW},  
\begin{eqnarray*}
\mathcal{P}_{C}y &=& \sum_{i=r+1}^{r+s} \frac{1}{1-\sigma_i^2} \tilde{Z}_{i} \tilde{Z}^T_{i}\tilde{W}_{j'} - \sum_{i=r+1}^{r+s} \frac{\sigma_i}{1-\sigma_i^2} \tilde{Z}_{i} \tilde{W}^T_{i}\tilde{W}_{j'} + \sum_{i=r+s+1}^{p} \tilde{Z}_{i} \tilde{Z}^T_{i}\tilde{W}_{j'}\\
&=&\frac{\sigma_{j'}}{1-\sigma_{j'}^2}\tilde{Z}_{j'}-\frac{\sigma_{i'}}{1-\sigma_{i'}^2} \tilde{Z}_{j'} \\
&=& 0. 
\end{eqnarray*}  
If $j'\ge r+s+1$ then also by \eqref{eq:Z^TW},  
\begin{eqnarray*}
\mathcal{P}_{C}y = \sum_{i=r+1}^{r+s} \frac{1}{1-\sigma_i^2} \tilde{Z}_{i} \tilde{Z}^T_{i}\tilde{W}_{j'} - \sum_{i=r+1}^{r+s} \frac{\sigma_i}{1-\sigma_i^2} \tilde{Z}_{i} \tilde{W}^T_{i}\tilde{W}_{j'} + \sum_{i=r+s+1}^{p} \tilde{Z}_{i} \tilde{Z}^T_{i}\tilde{W}_{j'}
= 0. 
\end{eqnarray*}  
Thus, $\mathcal{P}_{\ccal}y=0$ for all $y\in \text{col}(\tilde{W}_{(r+1):K})$ and (iii) is proved.

It remains to show (iv). By \eqref{eq:PC alt} we have
\begin{eqnarray*}
\mathcal{P}_{\ccal}^Tx &=&  \sum_{i=r+1}^{r+s} \frac{1}{1-\sigma_i^2} \tilde{Z}_{i} \tilde{Z}^T_{i}x - \sum_{i=r+1}^{r+s} \frac{\sigma_i}{1-\sigma_i^2}  \tilde{W}_{i}\tilde{Z}_{i}^Tx + \sum_{i=r+s+1}^{p} \tilde{Z}_{i} \tilde{Z}^T_{i}x\\
&=& \sum_{i=r+1}^{r+s} \left( \frac{\tilde{Z}^T_{i}x}{1-\sigma_i^2}\cdot \tilde{Z}_{i}  - \frac{\sigma_i\tilde{Z}_{i}^Tx}{1-\sigma_i^2} \cdot \tilde{W}_{i}\right) + \sum_{i=r+s+1}^{p}\tilde{Z}^T_{i}x\cdot \tilde{Z}_{i}.
\end{eqnarray*}
By  \eqref{eq:tilde Z W}, we have $\tilde{Z}^T\tilde{W}=\Sigma$. This implies that $\tilde{Z}_{r+s+1},...,\tilde{Z}_{p}$ are orthogonal to $\tilde{Z}_{r+1},...,\tilde{Z}_{r+s}$ and $\tilde{W}_{r+1},...,\tilde{W}_{r+s}$. Moreover, $\tilde{Z}_i$ and $\tilde{W}_j$ are orthogonal for any pair of $i$ and $j$ such that $i\neq j$ and $r+1\le i,j\le r+s$. Also, $\tilde{Z}_i^T\tilde{W}_i = \sigma_i$ if $r+1\le i\le r+s$.    
Therefore,
\begin{eqnarray*}
\|\mathcal{P}_{\ccal}^Tx\|^2 &=& \Big\|\sum_{i=r+1}^{r+s} \left( \frac{\tilde{Z}^T_{i}x}{1-\sigma_i^2} \cdot\tilde{Z}_{i}  - \frac{\sigma_i\tilde{Z}_{i}^Tx}{1-\sigma_i^2} \cdot \tilde{W}_{i}\right)\Big\|^2 + \Big\|\sum_{i=r+s+1}^{p}\tilde{Z}^T_{i}x \cdot\tilde{Z}_{i}\Big\|^2\\
&=& \sum_{i=r+1}^{r+s} \left(\frac{(\tilde{Z}^T_{i}x)^2}{(1-\sigma_i^2)^2}+\frac{\sigma_i^2(\tilde{Z}^T_{i}x)^2}{(1-\sigma_i^2)^2}-2\frac{\sigma_i^2(\tilde{Z}^T_{i}x)^2}{(1-\sigma_i^2)^2}\right) +\sum_{i=r+s+1}^p (\tilde{Z}^T_{i}x)^2\\
&=& \sum_{i=r+1}^{r+s}\frac{(\tilde{Z}^T_{i}x)^2}{1-\sigma_i^2} + \sum_{i=r+s+1}^p (\tilde{Z}^T_{i}x)^2\\
&\ge& \|\tilde{Z}_{(r+1):p}x\|^2,
\end{eqnarray*}
and (iv) is proved.
\end{proof}

\begin{proof}[Proof of Lemma~\ref{lem:comparing Phat and Ptilde}]
According to \eqref{eq:svd sample} and \eqref{eq:Z hat W breve}, the singular value decomposition of $\hat{Z}^T\breve{W}$ admits the following form:
$$\hat{Z}^T\breve{W} = \hat{\Sigma}=
\begin{pmatrix}
\hat{I} & 0 & 0 \\
0& \hat{\Gamma}& 0\\
0&0&\hat{\Phi}
\end{pmatrix}
\in\mathbb{R}^{p\times K},
$$
where $\hat{I}\in\mathbb{R}^{r\times r}$ and $\hat{\Gamma}\in\mathbb{R}^{s\times s}$ are square diagonal matrices, and $\hat{\Phi}\in\mathbb{R}^{(p-r-s)\times(K-r-s)}$ is a rectangular diagonal matrix. 
It then follows from the definition of $\hat{M}$ in Algorithm~\ref{algo:estimation-most-general} that
\begin{eqnarray*}
\hat{M}^T\hat{M} = 
\begin{pmatrix}
I_s & 0 & \hat{\Gamma} & 0 \\
0& I_{p-r-s} &0 & \hat{\Phi}\\
\hat{\Gamma} & 0& I_s & 0\\
0 & \hat{\Phi}^T &0 & I_{K-r-s} 
\end{pmatrix}.
\end{eqnarray*}
Let $Q$ be the matrix obtained from $\hat{M}^T\hat{M}$ by replacing $\hat{\Phi}$ and $\hat{\Phi}^T$ with zero matrices of the same dimensions. 
Using \eqref{eq:MTM} and \eqref{eq:singular value bound}, we have
\begin{eqnarray*}
\max\{\|Q-\hat{M}^T\hat{M}\|,\|\hat{M}^T\hat{M}-M^TM\|,\|Q-M^TM\|\} \le \|\hat{\Gamma}-\Gamma\|+\|\hat{\Phi}\| \le 2\|\hat{\Sigma}-\Sigma\|\le 2\tau_n.
\end{eqnarray*}  
Also, from \eqref{eq:MTM} and a direct calculation we see that $M^TM$ has $p+K-2r$ eigenvalues, of which $p+K-2r-2s$ are identically one and the remaining $2s$ eigenvalues are $1\pm\sigma_i$ where $\sigma_i$ are diagonal elements of $\Gamma$ and $r+1\le i\le r+s$. Therefore by Weyl's inequality,
\begin{eqnarray*}
1-\sigma_{r+1}-2\tau_n\le \|\hat{M}^T\hat{M}\|,\|Q\| \le 1+\sigma_{r+1}+2\tau_n,
\end{eqnarray*}
which implies that both $(\hat{M}^T\hat{M})^{-1}$ and $Q^{-1}$ exist by condition \eqref{eq: singular value gap condition}. 
A direct calculation shows that 
\begin{eqnarray*}
Q^{-1} &=& 
\begin{pmatrix}
(I-\hat{\Gamma}^2)^{-1} & 0 & -(I-\hat{\Gamma}^2)^{-1}\hat{\Gamma} & 0 \\
0& I_{p-r-s} & 0 & 0 \\
-(I-\hat{\Gamma}^2)^{-1}\hat{\Gamma} &0& (I-\hat{\Gamma}^2)^{-1} & 0\\ 
0 & 0 & 0& I_{K-r-s} 
\end{pmatrix}.
\end{eqnarray*}
Moreover, $\tilde{\mathcal{P}}_C$ and $\tilde{\mathcal{P}}_N$ defined in \eqref{eq:PCtilde} and \eqref{eq:PNtilde}, respectively, can be written as follows: 
\begin{eqnarray}
\label{eq:PCtilde Alt}\tilde{\mathcal{P}}_C &=& (\hat{Z}_{(r+1):p},0_{n\times (K-r)})Q^{-1}\hat{M}^T,\\
\label{eq:PNtilde Alt}\tilde{\mathcal{P}}_N &=& (0_{n\times (p-r)},\breve{W}_{(r+1):K})Q^{-1}\hat{M}^T.
\end{eqnarray}
Therefore by \eqref{eq: singular value gap condition},
\begin{eqnarray*}
\|(\hat{M}^T\hat{M})^{-1}-Q^{-1}\| &\le& \|\hat{M}^T\hat{M}-Q\|\cdot \|(\hat{M}^T\hat{M})^{-1}\|\cdot\|Q^{-1}\|\\
&\le& \frac{2\tau_n}{(1-\sigma_{r+1}-2\tau_n)^2}\\
&\le& \frac{4\tau_n}{(1-\sigma_{r+1})^2}.
\end{eqnarray*}
The claim of the lemma then follows from \eqref{eq:PChat}, \eqref{eq:PNhat}, \eqref{eq:PCtilde Alt}, \eqref{eq:PNtilde Alt}, the inequality above and the fact that $\|\hat{M}\|\le 2$.
\end{proof}

To prove Lemma~\ref{lem:bounding perturbation unorthogonal proj}, we need the following two lemmas, which will be proved at the end of this section.

\begin{lem}[Pertubation of singular vectors and partial singular decomposition]\label{lem: projection bound}
Let $X, \hat{X}$ be $p\times k$ matrices with singular decompositions
$$
X = U\Sigma V^T, \quad \hat{X} = \hat{U}\hat{\Sigma}\hat{V}^T.
$$
Let $\eta$ be a set of singular values of $X$, $\eta^c$ be the set of remaining singular values of $X$ and 
$\Delta = \min_{\sigma\in\eta,\sigma'\in\eta^c}|\sigma-\sigma'|.$
Denote by $U_\eta$ the matrix whose columns are left singular vectors of $X$ with corresponding singular values in $\eta$; define $V_{\eta}, \hat{U}_\eta$ and $\hat{V}_\eta$ similarly. Also, let $\Gamma$ and $\hat{\Gamma}$ be diaognal matrices with $\{\sigma_i,i\in\eta\}$ and $\{\hat{\sigma}_i,i\in\eta\}$ on the diagonal, respectively.  
Assume that $\|X\|\le 1$, $\|\hat{X}\|\le 1$ and $\Delta>4\|\hat{X}-X\|$. Then
\begin{eqnarray}
\label{eq: projection difference bound}
\max\left\{\|\hat{U}_\eta\hat{U}_\eta^T-U_\eta U_\eta^T\|,\|\hat{V}_\eta\hat{V}_\eta^T-V_\eta V_\eta^T\|\right\}  &\le& \frac{2\|\hat{X}-X\|}{\Delta-4\|\hat{X}-X\|}, \\
\label{eq: singular decomposition restriction bound}
\Big\|\sum_{i\in \eta} \sigma_i U_{\cdot i} V_{\cdot i}^T - \sum_{i\in \eta} \hat{\sigma}_i \hat{U}_{\cdot i} \hat{V}_{\cdot i}^T\Big\| &\le& \frac{3\|\hat{X}-X\|}{\Delta-4\|\hat{X}-X\|}. 
\end{eqnarray} 
Furthermore, if $1-\max_{i\in\eta}\sigma_i^2> 2\|\hat{X}-X\|$, we have
\begin{equation}\label{eq: Inverse Gamma}
\|(1-\hat{\Gamma}^2)^{-1}-(1-\Gamma^2)^{-1}\| \le \frac{2\|\hat{X}-X\|}{(1-\max_{i\in\eta}\sigma_i^2-2\|\hat{X}-X\|)^2}, 
\end{equation}
and 
\begin{eqnarray}
\nonumber&&\max\left\{ \|\hat{U}_\eta(1-\hat{\Gamma}^2)^{-1}\hat{U}_\eta^T-U_\eta(1-\Gamma^2)^{-1}U_\eta^T \|, \|\hat{V}_\eta(1-\hat{\Gamma}^2)^{-1}\hat{V}_\eta^T-V_\eta(1-\Gamma^2)^{-1}V_\eta^T \|\right\} \\
&\le& \frac{6\|\hat{X}-X\|} {(1-\max_{i\in\eta}\sigma_i^2-2\|\hat{X}-X\|)^2(\Delta - 4\norm{\hat{X}-X})}. 
\label{eq: InverseGammaPluspProjectionDifference}
\end{eqnarray}
\end{lem}

\begin{lem}[Unorthogonal projection bound]\label{lem: unorthogonal projection bound}
Let $X, \hat{X}$ be two matrices of the same size with singular decompositions
$$
X = U\Sigma V^T, \quad \hat{X} = \hat{U}\hat{\Sigma}\hat{V}^T.
$$
Assume that $\|X\|<1$, $\|\hat{X}\|<1$ and $\|X\|^2 +2\|\hat{X}-X\|<1$. Then
\begin{eqnarray*}
\left\|U\Sigma(1-\Sigma^2)^{-1} V^T - \hat{U}\hat{\Sigma}(1-\hat{\Sigma}^2)^{-1} \hat{V}^T\right\| &\le& \frac{3\|X-\hat{X}\|}{(1-\|X\|^2 -2\|\hat{X}-X\|)^2}.
\end{eqnarray*}
\end{lem}

\begin{proof}[Proof of Lemma~\ref{lem:bounding perturbation unorthogonal proj}] 
Recall that \eqref{eq: RUU} provides a bound on the difference between two matrices $Z^T\hat{W}\hat{W}^T$ and $Z^TWW^T$, the singular value decompositions of which are given by \eqref{eq: RUU decomposition}. Since $\delta_1,...,\delta_6$ depend on singular values and singular vectors of these matrices, we will use Lemma~\ref{lem: projection bound} and Lemma~\ref{lem: unorthogonal projection bound} to bound them.  

\noindent{\bf Bounding $\delta_1$.} Since $\tilde{Z} = ZU$ and $\hat{Z} = Z\hat{U}$ by \eqref{eq:tilde Z W} and \eqref{eq:Z hat W breve}, respectively, 
\begin{eqnarray*}
\delta_1 &=& \sum_{i=r+1}^{r+s} \left(\frac{1}{1-\hat{\sigma}_i^2} Z\hat{U}_{i}\hat{U}_{i}^T Z^T-\frac{1}{1-{\sigma}_i^2}ZU_{i}U_{i}^T Z^T\right) \\
&=& Z\left[\hat{U}_{(r+1):(r+s)}(1-\hat{\Gamma}^2)^{-1}\hat{U}_{(r+1):(r+s)}^T-U_{(r+1):(r+s)}(1-\Gamma^2)^{-1}U_{(r+1):(r+s)}^T\right]Z^T.
\end{eqnarray*}
By \eqref{eq: RUU decomposition}, column vectors of $U$ and  $\hat{U}$ are left singular vectors of $Z^TWW^T$ and $Z^T\hat{W}\hat{W}^T$, respectively. Also, $\|Z^TWW^T-Z^T\hat{W}\hat{W}^T\|\le \tau_n$ by \eqref{eq: RUU} and $\min\{1-\sigma_{r+1},\sigma_{r+s}\}\ge \tau_n/40$ by \eqref{eq: singular value gap condition}. It then follows from 
inequality \eqref{eq: InverseGammaPluspProjectionDifference} of Lemma~\ref{lem: projection bound} that
$$
\|\delta_1\| \le  \frac{12\tau_n}{\left(\min\{1-\sigma_{r+1},\sigma_{r+s}\}\right)^3}.
$$

\noindent{\bf Bounding $\delta_2$.} The same argument for bounding $\delta_1$ can be applied here to bound $\delta_2$. Since $\tilde{Z} = ZU$ and $\hat{Z} = Z\hat{U}$ by \eqref{eq:tilde Z W} and \eqref{eq:Z hat W breve}, we have
\begin{eqnarray*}
\delta_2 = Z \left[\hat{U}_{(r+1):(r+s)}\hat{\Gamma}(1-\hat{\Gamma}^2)^{-1} \breve{W}_{(r+1):(r+s)}^T-U_{(r+1):(r+s)}\Gamma(1-\Gamma^2)^{-1} \tilde{W}_{(r+1):(r+s)}^T\right].
\end{eqnarray*}

By \eqref{eq: singular value gap condition}, \eqref{eq: RUU decomposition}, \eqref{eq: RUU} and the bound \eqref{eq: singular decomposition restriction bound} of Lemma~\ref{lem: projection bound},  
\begin{eqnarray*}
\left\|\hat{U}_{(r+1):(r+s)}\hat{\Gamma} \breve{W}_{(r+1):(r+s)}^T-U_{(r+1):(r+s)}\Gamma \tilde{W}_{(r+1):(r+s)}^T\right\| &\le& \frac{3\tau_n}{\min\{1-\sigma_{r+1},\sigma_{r+s}\}-4\tau_n}\\
&\le& \frac{5\tau_n}{\min\{1-\sigma_{r+1},\sigma_{r+s}\}}.
\end{eqnarray*}
It remains to apply Lemma~\ref{lem: unorthogonal projection bound}, for which the following condition must be satisfied:
$$
\sigma_{r+1} + \frac{6\tau_n}{\min\{1-\sigma_{r+1},\sigma_{r+s}\}-4\tau_n} < 1.
$$
This follows from \eqref{eq: singular value gap condition} because $\min\{1-\sigma_{r+1},\sigma_{r+s}\}\ge \tau_n/c\ge 20\tau_n$ and
\begin{eqnarray*}
\frac{6\tau_n}{\min\{1-\sigma_{r+1},\sigma_{r+s}\}-4\tau_n} \le \frac{10\tau_n}{\min\{1-\sigma_{r+1},\sigma_{r+s}\}} \le \min\{1-\sigma_{r+1},\sigma_{r+s}\} \le 1-\sigma_{r+1}.
\end{eqnarray*}
Therefore by Lemma~\ref{lem: unorthogonal projection bound}, 
\begin{eqnarray*}
\|\delta_2\| &\le& \left\|\hat{U}_{(r+1):(r+s)}\hat{\Gamma}(1-\hat{\Gamma}^2)^{-1} \breve{W}_{(r+1):(r+s)}^T-U_{(r+1):(r+s)}\Gamma(1-\Gamma^2)^{-1} \tilde{W}_{(r+1):(r+s)}^T\right\| \\
&\le& 
\frac{3\cdot\frac{5\tau_n}{\min\{1-\sigma_{r+1},\sigma_{r+s}\}}}{\left(1-\sigma_{r+1}^2 -2\cdot\frac{5\tau_n}{\min\{1-\sigma_{r+1},\sigma_{r+s}\}}\right)^2} \\
&\le& 
\frac{\frac{15\tau_n}{\min\{1-\sigma_{r+1},\sigma_{r+s}\}}}{\left(1-\sigma_{r+1} -\frac{10\tau_n}{\min\{1-\sigma_{r+1},\sigma_{r+s}\}}\right)^2} \\
&\le& \frac{60\tau_n}{\left(\min\{1-\sigma_{r+1},\sigma_{r+s}\}\right)^3}.
\end{eqnarray*}

\noindent{\bf Bounding $\delta_3$.} Since $\tilde{Z} = ZU$ and $\hat{Z} = Z\hat{U}$ by \eqref{eq:tilde Z W} and \eqref{eq:Z hat W breve},
\begin{eqnarray*}
\delta_3 = Z\left[\hat{U}_{(r+s+1):p}\hat{U}_{(r+s+1):p}^T-U_{(r+s+1):p}U_{(r+s+1):p}^T\right]Z^T.
\end{eqnarray*}
By \eqref{eq: singular value gap condition}, \eqref{eq: RUU decomposition}, \eqref{eq: RUU} and the bound \eqref{eq: projection difference bound} of Lemma~\ref{lem: projection bound}, we have 
\begin{eqnarray*}
\|\delta_3\| \le \left\|\hat{U}_{(r+s+1):p}\hat{U}_{(r+s+1):p}^T-U_{(r+s+1):p}U_{(r+s+1):p}^T\right\|
\le \frac{2\tau_n}{\sigma_{r+s}-4\tau_n} \le \frac{4\tau_n}{\min\{1-\sigma_{r+1},\sigma_{r+s}\}}.
\end{eqnarray*}

\noindent{\bf Bounding $\delta_4$.} We have
\begin{eqnarray*}
\delta_4 = \breve{W}_{(r+s):(r+s)} (1-\hat{\Gamma}^2)^{-1} \breve{W}_{(r+s):(r+s)}^T - \tilde{W}_{(r+s):(r+s)} (1-\Gamma^2)^{-1} \tilde{W}_{(r+s):(r+s)}^T.
\end{eqnarray*}
By \eqref{eq: singular value gap condition}, \eqref{eq: RUU decomposition}, \eqref{eq: RUU} and the bound \eqref{eq: InverseGammaPluspProjectionDifference} of Lemma~\ref{lem: projection bound}, we have 
$$
\|\delta_4\| \le   \frac{12\tau_n}{\left(\min\{1-\sigma_{r+1},\sigma_{r+s}\}\right)^3}.
$$

\noindent{\bf Bounding $\delta_5$.}
Since $\tilde{Z} = ZU$ and $\hat{Z} = Z\hat{U}$ by \eqref{eq:tilde Z W} and \eqref{eq:Z hat W breve},
\begin{eqnarray*}
\delta_5 = \left[\breve{W}_{(r+1):(r+s)} \hat{\Gamma}(1-\hat{\Gamma}^2)^{-1} \hat{U}_{(r+1):(r+s)}^T- \tilde{W}_{(r+1):(r+s)} \Gamma(1-\Gamma^2)^{-1} U_{(r+1):(r+s)}^T\right]Z^T
=\delta_2^T,
\end{eqnarray*}
and the bound for $\delta_2$ applies.

\noindent{\bf Bounding $\delta_6$.} We have
$$
\delta_6 = \breve{W}_{(r+s+1):K}\breve{W}_{(r+s+1):K}^T-\tilde{W}_{(r+s+1):K}\tilde{W}_{(r+s+1):K}^T.
$$
By \eqref{eq: singular value gap condition}, \eqref{eq: RUU decomposition}, \eqref{eq: RUU} and the bound \eqref{eq: projection difference bound} of Lemma~\ref{lem: projection bound}, we have 
\begin{eqnarray*}
\|\delta_6\| \le \frac{2\tau_n}{\sigma_{r+s}-4\tau_n} \le \frac{4\tau_n}{\min\{1-\sigma_{r+1},\sigma_{r+s}\}}.
\end{eqnarray*}
The proof of Lemma~\ref{lem:bounding perturbation unorthogonal proj} is complete by combining the bounds for $\delta_1$ to $\delta_6$.
\end{proof}

\begin{proof}[Proof of Lemma~\ref{lem: projection bound}]
Note that $X^TX$ and $\hat{X}^T\hat{X}$ are symmetric matrices with eigenvectors $V_{\cdot i}$, $\hat{V}_{\cdot i}$ and eigenvalues being squares of singular values of $X$ and $\hat{X}$, respectively. Also, by the triangle inequality, 
\begin{eqnarray*}
\|\hat{X}^T\hat{X}-X^TX\| = \|(\hat{X}-X)^T\hat{X}-X^T(\hat{X}-X)\| \le 2\|\hat{X}-X\|. 
\end{eqnarray*}
Therefore by Theorem 7.3.1 of \cite{Bhatia1996},
\begin{eqnarray*}
\|\hat{V}_\eta\hat{V}_\eta^T-V_\eta V_\eta^T\| &\le& \frac{\|\hat{X}^T\hat{X}-X^TX\|}{\Delta - 2\|\hat{X}^T\hat{X}-X^TX\|} \le \frac{2\|\hat{X}-X\|}{\Delta-4\|\hat{X}-X\|}.
\end{eqnarray*}
A similar bound holds for $\|\hat{U}_\eta\hat{U}_\eta^T-U_\eta U_\eta^T\|$ if we consider $XX^T$ instead of $X^TX$, and \eqref{eq: projection difference bound} is proved. 

To show \eqref{eq: Inverse Gamma}, we use Weyl's inequality
\begin{eqnarray*}
\|(1-\hat{\Gamma}^2)^{-1}-(1-\Gamma^2)^{-1}\| &=& \max_{i\in\eta} \frac{|\hat{\sigma}_i^2-\sigma_i^2|}{(1-\hat{\sigma}_i^2)(1-\sigma_i^2)}\\
&\le& \frac{2\|\hat{X}-X\|}{(1-\max_{i\in\eta}\sigma_i^2)(1-\max_{i\in\eta}\sigma_i^2-2\|\hat{X}-X\|)}\\
&\le& \frac{2\|\hat{X}-X\|}{(1-\max_{i\in\eta}\sigma_i^2-2\|\hat{X}-X\|)^2}.
\end{eqnarray*}

To show \eqref{eq: singular decomposition restriction bound}, let $S$ and $\hat{S}$ be diagonal matrices with $\sigma_i,i\in \eta$ and $\hat{\sigma}_i,i\in \eta$ on the diagonal, respectively. For every unit vector $x$, we need to bound 
$\|U\eta S V_\eta^Tx - \hat{U}\eta \hat{S} \hat{V}_\eta^Tx\|$. Let $H=V_\eta V_\eta^T, H^\perp = I - H$, $\hat{H} = \hat{V}_\eta\hat{V}_\eta^T$ and $\hat{H}^\perp = I - \hat{H}$. Since
$$
x = Hx + H^\perp x = \hat{H}x+\hat{H}^\perp x,
$$
it follows that
$$
U\eta S V_\eta^Tx - \hat{U}\eta \hat{S} \hat{V}_\eta^Tx = U\eta S V_\eta^T H x - \hat{U}\eta \hat{S} \hat{V}_\eta^T \hat{H}x = XHx - \hat{X}\hat{H}x.
$$
By \eqref{eq: projection difference bound}, we have
\begin{eqnarray*}
\|\hat{X} \hat{H}x - \hat{X} Hx\| \le \|(\hat{H}-H)x\| \le \frac{2\|\hat{X}-X\|}{\Delta-4\|\hat{X}-X\|}.
\end{eqnarray*}
Therefore
\begin{eqnarray*}
\|U\eta S V_\eta^Tx - \hat{U}\eta \hat{S} \hat{V}_\eta^Tx\| &\le& \|XHx - \hat{X}Hx\| + \|\hat{X} Hx-\hat{X} \hat{H}x\|\\
&\le& \|X - \hat{X}\| + \frac{2\|\hat{X}-X\|}{\Delta-4\|\hat{X}-X\|}\\
&\le& \frac{3\|\hat{X}-X\|}{\Delta-4\|\hat{X}-X\|}.
\end{eqnarray*} 
Since $x$ is an arbitrary unit vector, \eqref{eq: singular decomposition restriction bound} is proved.

Finally, to show \eqref{eq: InverseGammaPluspProjectionDifference}, define $\hat{Z} = \hat{U}_{\eta}\hat{\Gamma}\hat{U}_{\eta}^T$ and $Z = U_{\eta} \Gamma U_{\eta}^T$. We have
\begin{align*}
\norm{\hat{U}_\eta(I-\hat{\Gamma}^2)^{-1}\hat{U}_\eta^T-U_\eta(I-\Gamma^2)^{-1}U_\eta^T} & = \norm{(I-\hat{Z}^2)^{-1} - (I-Z^2)^{-1}}\\
&= \norm{(I-\hat{Z}^2)^{-1}(\hat{Z}^2-Z^2)(I-Z^2)^{-1}}\\
& \le \norm{\hat{Z}^2-Z^2}\norm{(I-\hat{Z}^2)^{-1}}\norm{(I-Z^2)^{-1}}.
\end{align*}
It is easy to see that 
$$\norm{(I-Z^2)^{-1}} = \frac{1}{1-\max_{i\in \eta}\sigma_i^2}$$
 and
$$\norm{(I-\hat{Z}^2)^{-1}}= \frac{1}{1-\max_{i\in \eta}\hat{\sigma}_i^2} \le  \frac{1}{1-\max_{i\in \eta}\sigma_i^2 - 2\norm{\hat{X}-X}}.$$
Furthermore, by \eqref{eq: singular decomposition restriction bound},
\begin{align*}
\norm{\hat{Z}^2-Z^2}  &= \norm{\hat{U}_{\eta}\hat{\Gamma}\hat{V}^T_{\eta}\hat{V}_{\eta}\hat{\Gamma}\hat{U}_{\eta}^T - U_{\eta}\Gamma V^T_{\eta}V_{\eta}\Gamma U_{\eta}^T}\\
& \le \norm{\hat{U}_{\eta}\hat{\Gamma}\hat{V}_{\eta}^T(\hat{V}_{\eta}\hat{\Gamma}\hat{U}_{\eta}^T - V_{\eta}\Gamma U_{\eta}^T)} + \norm{(\hat{U}_{\eta}\hat{\Gamma}\hat{V}_{\eta}^T - U_{\eta}\Gamma V^T_{\eta})V_{\eta}\Gamma U_{\eta}^T}\\
& \le 2 \norm{\hat{U}_{\eta}\hat{\Gamma}\hat{V}_{\eta}^T - U_{\eta}\Gamma V^T_{\eta}}\\
& \le \frac{6\|\hat{X}-X\|}{\Delta-4\|\hat{X}-X\|}.
\end{align*}

Combining the above inequalities gives the stated result. The counterpart for $V$ is proved in the same way.
\end{proof}

\begin{proof}[Proof of Lemma~\ref{lem: unorthogonal projection bound}]
Since $\|\Sigma\|=\|X\|<1$ and $\|\hat{\Sigma}\|=\|\hat{X}\|<1$, we have
$$
\Sigma(1-\Sigma^2)^{-1} = \sum_{i=0}^\infty \Sigma^{2i+1},\qquad \hat{\Sigma}(1-\hat{\Sigma}^2)^{-1} = \sum_{i=0}^\infty \hat{\Sigma}^{2i+1}.  
$$
Therefore 
\begin{eqnarray*}
U\Sigma(1-\Sigma^2)^{-1} V^T - \hat{U}\hat{\Sigma}(1-\hat{\Sigma}^2)^{-1} \hat{V}^T &=&
\sum_{i=0}^\infty \left(U\Sigma^{2i+1}V^T-\hat{U}\hat{\Sigma}^{2i+1}\hat{V}^T\right)\\
&=& \sum_{i=0}^\infty \left(U\Sigma V^T V \Sigma^{2i}V^T-\hat{U}\hat{\Sigma}\hat{V}^T\hat{V}\hat{\Sigma}^{2i}\hat{V}^T\right)\\
&=& \sum_{i=0}^\infty \left(X V \Sigma^{2i}V^T-\hat{X}\hat{V}\hat{\Sigma}^{2i}\hat{V}^T\right)\\
&=& \sum_{i=0}^\infty \left(X Y^{2i}-\hat{X}\hat{Y}^{2i}\right),
\end{eqnarray*}
where 
$Y = V\Sigma V^T$ and $\hat{Y} = \hat{V}\hat{\Sigma}\hat{V}^T$. Since $X^TX = Y^2$ and $\hat{X}^T\hat{X}=\hat{Y}^2$, it follows that $\|Y^2\|<1$ and
\begin{eqnarray*}
\|\hat{Y}^2 -Y^2\|\le\|\hat{X}^T(\hat{X}-X)\|+\|(\hat{X}-X)X\| \le 2\|\hat{X}-X\|.
\end{eqnarray*}
By the triangle inequality and the assumption of the lemma,
$$
\|\hat{Y}^2\|\le\|Y^2\|+ 2\|\hat{X}-X\| = \|X\|^2+ 2\|\hat{X}-X\| <1.
$$
Therefore $\sum_{i=0}^\infty Y^{2i}$ and $\sum_{i=0}^\infty \hat{Y}^{2i}$ are convergent series. By the triangle inequality, we have
\begin{eqnarray*}
\Big\|\sum_{i=0}^\infty \left(X Y^{2i}-\hat{X}\hat{Y}^{2i}\right)\Big\| &=&\Big\| \sum_{i=0}^n (X-\hat{X})Y^{2i} + \sum_{i=0}^\infty \hat{X}(Y^{2i}-\hat{Y}^{2i}) \Big\| \\
&\le& \|X-\hat{X}\|\cdot \Big\|\sum_{i=0}^\infty Y^{2i}\Big\| + \Big\|\sum_{i=0}^\infty Y^{2i}-\sum_{i=0}^\infty \hat{Y}^{2i}\Big\|\\
&=& \|X-\hat{X}\|\cdot\|(1-Y^2)^{-1}\| + \Big\|(1-\hat{Y}^2)^{-1}-(1-Y^2)^{-1}\Big\|\\
&=& \|X-\hat{X}\|\cdot\|(1-Y^2)^{-1}\|+\Big\|(1-\hat{Y}^2)^{-1}(\hat{Y}^2-Y^2)(1-Y^2)^{-1}\Big\|\\
&\le& \|X-\hat{X}\|\cdot\|(1-Y^2)^{-1}\|+\|\hat{Y}^2-Y^2\|\cdot \|(1-\hat{Y}^2)^{-1}\|\cdot\|(1-Y^2)^{-1}\|\\
&\le& \|X-\hat{X}\|\cdot\|(1-Y^2)^{-1}\|+2\|\hat{X}-X\|\cdot \|(1-\hat{Y}^2)^{-1}\|\cdot\|(1-Y^2)^{-1}\|.
\end{eqnarray*}
Note that $\|(1-Y^2)^{-1}\| = (1-\|Y^2\|)^{-1}=(1-\|X\|^2)^{-1}$ and
$$
\|(1-\hat{Y}^2)^{-1}\| = \frac{1}{1-\|\hat{Y}^2\|}\le \frac{1}{1-\|Y^2\| -2\|\hat{X}-X\|} = \frac{1}{1-\|X\|^2 -2\|\hat{X}-X\|}.
$$
Putting these inequalities together, we get
\begin{eqnarray*}
\Big\|\sum_{i=0}^\infty \left(X Y^{2i}-\hat{X}\hat{Y}^{2i}\right)\Big\| &\le& \frac{3\|X-\hat{X}\|}{(1-\|X\|^2 -2\|\hat{X}-X\|)^2}
\end{eqnarray*}
and the proof is complete.
\end{proof}

\newpage 

\section{Consistency of variance estimation}

Denote by $H$ and $H^\perp$ the orthogonal projections onto $\text{span}\pare{\col(X),S_K(P)}$ and its orthogonal complement, respectively. Then $H^\perp Y=H^\perp\epsilon$ and therefore
\begin{equation}\label{eq: variance eps}
\mathbb{E} \|H^\perp Y\|^2 =  \mathbb{E} \text{Trace}(H^\perp \epsilon\epsilon^T) = (n-p-K+r)\sigma^2,
\end{equation}
where $p+K-r$ is the dimension of $\Span\pare{\text{col}(X),S_K(P)}$. 
Since $\mathcal{P}_{\rcal}$ is the orthogonal projection onto $\rcal =\col(X)\cap S_K(P)$ by \eqref{PRhat} and $\mathcal{P}_{\ccal}+\mathcal{P}_{\ncal}$ is the orthogonal projection onto the subspace of $\Span\pare{\col(X),S_K(P)}$ that is orthogonal to $\rcal$ by Lemma~\ref{lem:unorth proj}, we have  
\begin{eqnarray*}
H = \mathcal{P}_{\rcal}+\mathcal{P}_{\ccal}+\mathcal{P}_{\ncal}.
\end{eqnarray*}
Recall $\hat{H} = \hat{\mathcal{P}}_R+\hat{\mathcal{P}}_C+\hat{\mathcal{P}}_N$ from Algorithm~\ref{algo:estimation-most-general}. Denote $\hat{H}^\perp = I - \hat{H}$ and $H^\perp = I-H$. 

\begin{proof}[Proof of Proposition~\ref{prop:variance-consistency}]
We first bound the difference $$\hat{H}-H = (\hat{\mathcal{P}}_R-\mathcal{P}_{\rcal})+(\hat{\mathcal{P}}_C-\mathcal{P}_{\ccal})+(\hat{\mathcal{P}}_N-\mathcal{P}_{\ncal}).$$ 
Recall that  $\hat{\mathcal{P}}_R = \hat{Z}_{1:r}\hat{Z}_{1:r}^T$ by \eqref{PRhat} and $\mathcal{P}_{\rcal} = \tilde{Z}_{1:r}\tilde{Z}_{1:r}^T$ by \eqref{eq:PR}. Also, $\tilde{Z}_{1:r} = ZU_{1:r}$ and $\hat{Z}_{1:r} = Z\hat{U}_{1:r}$ by \eqref{eq:tilde Z W} and Algorithm~\ref{algo:estimation-most-general}, respectively. 
By \eqref{eq: RUU decomposition}, column vectors of $U$ and  $\hat{U}$ are left singular vectors of $Z^TWW^T$ and $Z^T\hat{W}\hat{W}^T$, respectively. Also, $\|Z^TWW^T-Z^T\hat{W}\hat{W}^T\|\le \tau_n$ by \eqref{eq: RUU} and $1-\sigma_{r+1}\ge 40\tau_n$ by \eqref{eq: singular value gap condition}. It then follows from 
the bound \eqref{eq: projection difference bound} of Lemma~\ref{lem: projection bound} that
\begin{eqnarray}\label{eq:Pr}
\big\|\hat{\mathcal{P}}_R-\mathcal{P}_{\rcal}\big\| = \big\|Z\left(\hat{U}_{1:r}\hat{U}_{1:r}^T-U_{1:r}U_{1:r}^T\right)Z^T\big\|
\le \frac{2\tau_n}{1-\sigma_{r+1}-4\tau_n}
\le \frac{4\tau_n}{1-\sigma_{r+1}}.
\end{eqnarray}
Together with the bound of Corollary~\ref{cor:pertubation of projections} and the triangle inequality, we see that there exists a constant $C>0$ such that with high probability
\begin{eqnarray*}
\|\hat{H} - H\| &\le& \|\hat{\mathcal{P}}_R-\mathcal{P}_{\rcal}\| + \|\hat{\mathcal{P}}_C-\mathcal{P}_{\ccal}\| +\|\hat{\mathcal{P}}_N-\mathcal{P}_{\ncal}\| \\
&\le& \frac{C\tau_n}{\min\{(1-\sigma_{r+1})^3,\sigma_{r+s}^3\}}.
\end{eqnarray*}
Since $\hat{H}^\perp - H^\perp = \hat{H}-H$, we obtain
\begin{eqnarray*}
|\hat{\sigma}^2-\sigma^2| &=&
\Big|\frac{\|\hat{H}^\perp Y\|^2}{n-p-K+r}-\frac{\|{H}^\perp Y\|^2}{n-p-K+r}\Big| \\
&\le& \|\hat{H}^\perp-H^\perp\|\cdot\frac{\|Y\|^2}{n-p-K+r}\\
&\le& \frac{C\tau_n}{\min\{(1-\sigma_{r+1})^3,\sigma_{r+s}^3\}}\cdot\frac{\|Y\|^2}{n-p-K+r}.
\end{eqnarray*} 
The claim of the lemma then follows because $\|Y\|\le \|X\beta+X\theta+\alpha\|+\|\epsilon\|=O(\sqrt{n(\sigma^2+p)})$ with high probability because of Assumption \ref{ass:signal-scale} and the fact that $\|\epsilon\|^2 = O(\sigma^2n)$ with high probability.
\end{proof}

\newpage 

\section{Inference of $\alpha$ and $\beta$}

\begin{proof}[Proof of Theorem~\ref{thm:general-beta-inference}]
Recall the variance of $\omega^T\hat{\beta}$ in \eqref{eq:true-var}. By the triangle inequality, 
\begin{align*}
\omega^T\Theta \tilde{X}^T\hat{\mathcal{P}}_C\hat{\mathcal{P}}_C^T\tilde{X}\Theta\omega &\ge \omega^T\Theta \tilde{X}^T\mathcal{P}_{\ccal}\mathcal{P}_{\ccal}^T\tilde{X}\Theta\omega - |\omega^T\Theta \tilde{X}^T(\hat{\mathcal{P}}_C\hat{\mathcal{P}}_C^T-\mathcal{P}_{\ccal}\mathcal{P}_{\ccal}^T)\tilde{X}\Theta\omega|
\end{align*}
From part (iv) of Lemma~\ref{lem:unorth proj} we have 
\begin{eqnarray*}
\omega^T\Theta \tilde{X}^T\mathcal{P}_{\ccal}\mathcal{P}_{\ccal}^T\tilde{X}\Theta\omega = \|\mathcal{P}_{\ccal}^T\tilde{X}\Theta\omega\|^2 \ge \|\tilde{Z}_{(r+1):p}^T\tilde{X}\Theta\omega\|^2.
\end{eqnarray*}
It follows from \eqref{eq:PC alt} that $\|\mathcal{P}_{\ccal}\|\le 2/(1-\sigma_{r+1})$. By the triangle inequality, \newline
Corollary~\ref{cor:pertubation of projections} and Assumption~\ref{ass:small-perturbation}, 
\begin{eqnarray*}
\|\hat{\mathcal{P}}_C\| \le \|\hat{\mathcal{P}}_C-\mathcal{P}_{\ccal}\| + \|\mathcal{P}_{\ccal}\| \le \frac{C}{1-\sigma_{r+1}}.
\end{eqnarray*}
Using Corollary~\ref{cor:pertubation of projections} and above inequalities, we get
\begin{eqnarray*}
\|\hat{\mathcal{P}}_C\hat{\mathcal{P}}_C^T-\mathcal{P}_{\ccal}\mathcal{P}_{\ccal}^T\| &\le& \|\hat{\mathcal{P}}_C(\hat{\mathcal{P}}_C^T-\mathcal{P}_{\ccal})\| + \|(\hat{\mathcal{P}}_C-\mathcal{P}_{\ccal})\mathcal{P}_{\ccal}^T\|\\
&\le& \frac{C\tau_n}{\min\{(1-\sigma_{r+1})^4,\sigma_{r+s}^4\}}.
\end{eqnarray*}
Therefore by \eqref{eq:general-var-requirement}, Assumption~\ref{ass:small-perturbation} and \eqref{eq:general-signal-requirement}
\begin{eqnarray*}
\omega^T\Theta \tilde{X}^T\hat{\mathcal{P}}_C\hat{\mathcal{P}}_C^T\tilde{X}\Theta\omega 
&\ge& \|\tilde{Z}_{(r+1):p}^T\tilde{X}\Theta\omega\|^2  - \frac{C\tau_n\norm{\tilde{X}\Theta\omega}^2}{\min\{(1-\sigma_{r+1})^4,\sigma_{r+s}^4\}}\\
&\ge& c^2 - \frac{C\tau_n\norm{\Theta\omega}^2}{\min\{(1-\sigma_{r+1})^4,\sigma_{r+s}^4\}}\\
&\ge& c^2/2.
\end{eqnarray*}
The last inequality implies that the variance of $\omega^T\hat{\beta}$ satisfies $\var(\omega^T\hat{\beta}) \ge c^2\sigma^2/(2n)$. Since the bias of $\omega^T\hat{\beta}$ is of order $o(1/\sqrt{n})$ by Proposition~\ref{prop:beta-bias}, it follows that the term
$$\frac{\omega^T{\hat{\beta}}-\omega^T{\beta}}{\sqrt{\frac{\sigma^2}{n}\omega^T\Theta \tilde{X}^T\hat{\mathcal{P}}_C\hat{\mathcal{P}}_C^T\tilde{X}\Theta\omega }}$$
is asymptotically standard normal. The proof is complete because $\hat{\sigma}^2$ is consistent by Proposition~\ref{prop:variance-consistency}.
\end{proof}

\begin{proof}[Proof of Corollary~\ref{coro:inference-individual-beta}]
Let ${\eta}_j$ be the partial residual from regressing $\tilde{X}_j$ on all other columns in $\tilde{X}$. Therefore, we have $\tilde{X}_{j'}^T{\eta}_j = 0$ if $j' \ne j$. When making inference on $\beta_j$, we have $\omega =e_j$ and in particular, we can write it as
$$\omega = (\tilde{X}_j^T{\eta}_j)^{-1}\tilde{X}^T{\eta}_j.$$ 
Therefore,
\begin{eqnarray*}
\tilde{Z}_{(r+1):p}^T\tilde{X}\Theta\omega = (\tilde{X}_j^T{\eta}_j)^{-1}\tilde{Z}_{(r+1):p}^T\tilde{X} (\tilde{X}^T\tilde{X})^{-1} \tilde{X}^T{\eta}_j.
\end{eqnarray*}
Since $\tilde{X} (\tilde{X}^T\tilde{X})^{-1} \tilde{X}^T$ is the orthogonal projection onto the column space of $\tilde{X}$ and the column vectors of $\tilde{Z}$ form an orthonornal basis of that space, it follows that
\begin{eqnarray*}
\tilde{Z}_{(r+1):p}^T\tilde{X}\Theta\omega=(\tilde{X}_j^T{\eta}_j)^{-1}\tilde{Z}_{(r+1):p}^T\left(\tilde{Z}_{1:r}\tilde{Z}_{1:r}^T+\tilde{Z}_{(r+1):p}\tilde{Z}_{(r+1):p}^T\right){\eta}_i= (\tilde{X}_j^T{\eta}_j)^{-1}\tilde{Z}_{(r+1):p}^T{\eta}_j. 
\end{eqnarray*}
Therefore, we need $\norm{\tilde{Z}_{(r+1):p}^T{\eta}_j}/ |\tilde{X}_j^T{\eta}_j| \ge c$ to use \eqref{eq:general-var-requirement}. To make it more interpretable, notice that
$$|\tilde{X}_j^T{\eta}_j| = \norm{{\eta}_j}^2 = \norm{{\eta}_j}^2/\norm{\tilde{X}_j}^2 = 1- \norm{\tilde{X}_j-{\eta}_j}^2/\norm{\tilde{X}_j}^2 = 1-R_j^2,$$
where $R_j^2$ is the partial $R^2$ of regressing $\tilde{X}_j$ to the other columns in $\tilde{X}$. Thus, \eqref{eq:nontrivial-partial-correlation} implies \eqref{eq:general-var-requirement} and Corollary~\ref{coro:inference-individual-beta} is a consequence of Theorem~\ref{thm:general-beta-inference}.
\end{proof}

\begin{proof}[Proof of Theorem~\ref{thm:chisq-test}]
Let $\text{Bias}(\hat{\gamma})=\mathbb{E}\hat{\gamma}-O\gamma$ be the bias of $\hat{\gamma}$ in estimating $O\gamma$, where $O$ is the orthogonal matrix in Proposition~\ref{prop: bias variance of gamma hat}. Then using the notation from Algorithm~\ref{algo:estimation-most-general},
$$\hat{\gamma} \sim N\left(O\gamma+\text{Bias}(\hat{\gamma}\right) , \Sigma_{\hat{\gamma}}).$$
Therefore, under the null hypothesis $H_0: \alpha = 0$, we have $\gamma = 0$ and
$$\Sigma_{\hat{\gamma}}^{-1/2}\hat{\gamma} \sim N\left(\Sigma_{\hat{\gamma}}^{-1/2}\text{Bias}(\hat{\gamma}),I\right).$$
From Assumption~\ref{ass:small-perturbation}, which implies \eqref{eq: singular value gap condition}, and Proposition~\ref{prop: bias variance of gamma hat}, 
\begin{eqnarray*}
\big\| \Sigma_{\hat{\gamma}}^{-1/2}\text{Bias}(\hat{\gamma}) \big\| 
&\le& \big\|\Sigma_{\hat{\gamma}}\big\|^{-1/2}\cdot\norm{\text{Bias}(\gamma) } \\
&\le& \left(\sigma^2-\frac{C\tau_n}{\min\{(1-\sigma_{r+1})^2,\sigma_{r+s}^2\}}\right)^{-1/2}\cdot \frac{C\tau_n\sqrt{np}}{\min\{(1-\sigma_{r+1})^3,\sigma_{r+s}^3\}}\\
&=& o(1).
\end{eqnarray*}
Therefore, as $n \to \infty$ we have
$$\Sigma_{\hat{\gamma}}^{-1/2}\hat{\gamma} \xrightarrow[]{d} N(0, I).$$
Note that $\|\hat{\Sigma}_{\hat{\gamma}}- {\Sigma}_{\hat{\gamma}}\|\rightarrow 0$ because $\hat{\sigma}^2$ is consistent by Proposition~\ref{prop:variance-consistency}. Using Slusky's theorem and the continuously mapping theorem, we obtain
$$\norm{\hat{\gamma}_0}^2 = \big\|\hat{\Sigma}_{\hat{\gamma}}^{-1/2}\hat{\gamma}\big\| \xrightarrow[]{d} \chi^2_K.$$
The proof is complete.
\end{proof}

\begin{proof}[Proof of Theorem~\ref{thm:theta-inference}]
Since $Y=X\beta+X\theta+\alpha+\epsilon$, it follows from \eqref{eq:alpha beta theta hat} that
$$\hat{\theta} = (X^TX)^{-1}X^T \hat{\mathcal{P}}_{\rcal} Y = (X^TX)^{-1}X^T \hat{\mathcal{P}}_{\rcal}(X\beta+X\theta+\alpha)+(X^TX)^{-1}X^T \hat{\mathcal{P}}_{\rcal}\epsilon.$$ 
Recall that $\epsilon\sim N(0,\sigma^2 I_p)$ and $\tilde{X}=X/\sqrt{n}$. Therefore, $\hat{\theta}$ is a multivariate Gaussian vector with variance 
$$
\var(\hat{\theta}) = \frac{\sigma^2}{n} \Theta\tilde{X}^T \hat{\mathcal{P}}_{\rcal}\hat{\mathcal{P}}_{\rcal}^T \tilde{X}\Theta
$$  
and expectation
\begin{eqnarray*}
\e \hat{\theta} &=& (X^TX)^{-1}X^T \hat{\mathcal{P}}_{\rcal}(X\beta+X\theta+\alpha) \\
&=& \theta + (X^TX)^{-1}X^T(\hat{\mathcal{P}}_{\rcal}-\mathcal{P}_{\rcal})(X\beta+X\theta+\alpha).  
\end{eqnarray*}
By \ref{ass:signal-scale},  \ref{ass:weak-dependence}, \ref{ass:small-perturbation}, and \eqref{eq:Pr} we have   
\begin{eqnarray*}
\|\e \hat{\theta}-\theta\| \le \|(X^TX)^{-1}\|\cdot\|X\|\cdot\|X\beta+X\theta+\alpha\|\cdot \|\hat{\mathcal{P}}_{\rcal}-\mathcal{P}_{\rcal}\| = o(n^{-1/2}).
\end{eqnarray*}
The claim about limiting distributions of $\hat{\theta}_i$ follows from \eqref{eq:nontrivial-partial-correlation-theta} and the same argument used in the proof of Corollary~\ref{coro:inference-individual-beta}.  
\end{proof}

\newpage 

\section{Bounding projection perturbation for adjacency matrix}\label{sec:projection bound adjacency matrix}

In this section we prove Theorem~\ref{thm:projection-concentration}, Corollary~\ref{coro:small-perturbation-A} and Theorem~\ref{thm:necessary}. For the proof of Theorem~\ref{thm:projection-concentration}, we need the following lemma.

\begin{lem}[\cite{lei2014consistency}]
\label{lemma:concentration}
Let $A$ be the adjacency matrix of a random graph on $n$ nodes generated from the inhomogeneous Erd\H{o}s-R\'{e}nyi model and denote $d = n\cdot\max_{ij}\e A_{ij}$. Then there exist some constants $C,c_0>0$ such that if $d \ge C\log n$ then $\norm{A - \e A} \le C\sqrt{d}$ with probability at least $1-n^{-c_0}$.
\end{lem}

\begin{proof}[Proof of Theorem~\ref{thm:projection-concentration}]

Recall the eigen-projection representation (see for example \cite{xia2019data})
\be\label{eq:eigen-projection}
\hat{W}\hat{W}^T - WW^T = \sum_{k=1}^\infty S_k,
\qquad S_k = \frac{1}{2\pi i}\int_\Omega [R(z)E]^k R(z) dz.
\ee
Here $E = A - \e A$, $R(z) = (zI-\e A)^{-1}$ is the resolvent and $\Gamma$ is a contour separating $\{\lambda_1,...,\lambda_K\}$ from the remaining eigenvalues such that $\Omega$ is sufficiently far away from all eigenvalues of $\e A$, that is, for some constant $c_1>0$, 
$$
\min_{j\in [n], \ z\in\Gamma} |z-\lambda_j| \ge \frac{d}{2c_1}. 
$$ 
Such a contour exists because of the eigenvalue gap assumption on $\e A$. This property of $\Omega$ implies that for any $z\in \Omega$, we have
$$
\|R(z)\| \le \frac{2c_1}{d}.
$$

Note that the range of $\{\lambda_1,...,\lambda_K\}$ is of order $O(d)$, therefore the length of $\Omega$ is of order $O(d)$ as well. According to Lemma~\ref{lemma:concentration}, there exists constant $C>0$ such that the event
$$T_0 = \left\{\norm{E} \le C\sqrt{d}\right\}$$
occurs with probability at least $1-n^{-c_0}$. Under $T_0$, the bound on the resolvent yields an upper bound for $S_k$:
\begin{equation}\label{eq:Sk bound k greater one}
\|S_k\| \le \frac{1}{2\pi}\left(\frac{2c_1}{d}\right)^k \cdot (C\sqrt{d})^{k/2} \cdot \frac{2c_1}{d} \cdot \text{length}(\Omega) \le \frac{c}{\pi} \left(\frac{4c_1^2C}{d}\right)^{k/2}.
\end{equation}
We will use this inequality to upper bound $S_k$ for all $k\ge 2$. It remains to show that $S_1$ is small. Since
$$
R(z) = (zI-\e A)^{-1} = \sum_{j=1}^n \frac{1}{z-\lambda_j} {w}_j {w}_j^T,
$$
it follows that
$$
S_1 = \frac{1}{2\pi i}\int_\Omega R(z)E R(z) dz 
= \sum_{j,j' = 1}^n \frac{1}{2\pi i} \int_\Omega \frac{{w}_j^T E {w}_{j'}}{(z-\lambda_j)(z-\lambda_{j'})} dz \cdot {w}_j {w}_{j'}^T.
$$
If $j$ and $j'$ are both either at most $K$ or greater than $K$ then the integral is zero. If one of them is at most $K$ and the other does not, say $j\le K$ and $j'> K$, then the integral may not be zero; the contribution of all such terms is
\begin{eqnarray*}
M := 2\sum_{j\le K, \ j'>K} \frac{1}{2\pi i} \int_\Omega \frac{{w}_j^T E {w}_{j'}}{(z-\lambda_j)(z-\lambda_{j'})} dz \cdot {w}_j {w}_{j'}^T =
\sum_{j\le K} \frac{1}{2\pi i} \sum_{j'>K} \frac{{w}_j^T E {w}_{j'}}{\lambda_j - \lambda_{j'}} \cdot {w}_j {w}_{j'}^T.
\end{eqnarray*} 
Let ${v}$ be a fixed unit vector. Then
\begin{eqnarray}
\nonumber\|M {v}\|^2 &=& \Big\|2 \sum_{j\le K} \frac{1}{2\pi i} {w}_j \sum_{j'>K} \frac{{w}_j^T E {w}_{j'}}{\lambda_j - \lambda_{j'}} \cdot {w}_{j'}^T {v} \Big\|^2 \\
\nonumber&=& \frac{1}{\pi^2}\sum_{j\le K} \Big| \sum_{j'>K} \frac{{w}_j^T E {w}_{j'}}{\lambda_j - \lambda_{j'}} \cdot {w}_{j'}^T {v} \Big|^2 \\
\label{eq:Mv bound}&=& \frac{1}{\pi^2}\sum_{j\le K} \left| {w}_j^TE Q_j {v} \right|^2,
\end{eqnarray}
in which we denote 
$$Q_j = \sum_{j'>K} \frac{1}{\lambda_j-\lambda_{j'}} {w}_{j'}{w}_{j'}^T.$$
Note that $Q_j$ depends only on $P$ and since $|\lambda_j - \lambda_{j'}| \ge cd$ for $j\le K, j' >K$, we have $\|Q_j\| \le 1/(cd)$ for all $j \le K$,  which further implies $\norm{Q_j{v}} \le 1/(c_1d)$. For any fixed vectors ${{x}},{{y}} \in \bR^n$, 
$${{x}}^TE {{y}}  = \sum_{i<j}E_{ij}({ x}_i{ y}_j+{ x}_j{ y}_i) + \sum_iE_{ii}{x}_i{ y}_i,$$
where $E_{ij}, i\le j$ are independent Bernoulli errors. By Hoeffding's inequality (see for example Theorem 2.6.3 in \cite{vershynin2018high}), for some absolute constant $c'$ we have 
\begin{align*}
\p(|{{x}}^TE {{y}} |>t) &\le 2\exp\left(-\frac{c't^2}{\sum_{i<j}({ x}_i{y}_j+{x}_j{ y}_i)^2 + \sum_i{ x}_i^2{y}_i^2}\right)\\
& \le 2\exp\left(-\frac{c't^2}{2\norm{{x}}^2\norm{{ y}}^2}\right).
\end{align*}
Applying the above inequality to each $1\le j\le K$ and using the bound $\|Q_j{v}\|\le 1/(c_1d)$, we have
\begin{align*}
\p\left(\left| {w}_j^TEQ_j {v} \right| > \frac{\sqrt{\log n}}{d}\right) & \le 2\exp\left(-\frac{c'\log n/d^2}{2\norm{{w}_j}^2\norm{Q_j {v}}^2}\right) \\
& \le 2\exp\left(-\frac{c'\log n/d^2}{2/(c_1d)^2}\right)\\
& = 2 n^{-c'/2c_1^2}.
\end{align*}
Denote $T_j = \big\{\big|{w}_j^TE Q_j {v} \big| > \sqrt{\log n}/d\big\}$. Then under $\bigcap_{j=1}^K T_j$, by \eqref{eq:Mv bound} we have
$$\norm{M {v}}^2 \le \frac{K\log n}{\pi^2d^2}.$$ 
Finally, under $\bigcap_{j=0}^K T_j$, using \eqref{eq:eigen-projection} and \eqref{eq:Sk bound k greater one}, we obtain
\begin{eqnarray*}
\big\|(\hat{W}\hat{W}^T - WW^T){v}\big\| &\le& \norm{S_1{v}} + \sum_{k\ge 2}\norm{S_k{v}}\\
&\le& \frac{\sqrt{K\log n}}{d} + \sum_{k\ge 2}\left(\frac{4c^2C}{d}\right)^{k/2}\\
&\le& \frac{\sqrt{K\log n}}{d} + \frac{2}{d}\\
&\le& \frac{2\sqrt{K\log n}}{d}
\end{eqnarray*}
with probability at least $1-(1+2K)n^{-c}$ for $c = \min(c_0, c'/2c_1^2) >0$, in which the summation of the sequence is valid as long as $4c_1^2C/d < 1/2$. Notice that $c$ does not depend on ${v}$.
\end{proof}

\begin{proof}[Proof of Corollary~\ref{coro:small-perturbation-A}]
Notice that $Z$ is fixed matrix. Directly applying Theorem~\ref{thm:projection-concentration} for each column vector of $Z$ yields
$$\norm{(\hat{W}\hat{W}-WW^T)Z_j} \le \frac{2\sqrt{Kn\log n}}{d}$$
for every $1\le j\le p$. This implies
$$\norm{(\hat{W}\hat{W}-WW^T)Z}\le\frac{2\sqrt{pKn\log n}}{d}$$
with probability at least $1-pKn^{-c}$.
The proof is complete.
\end{proof}

\begin{proof}[Proof of Corollary~\ref{coro:r-consistency}]

For sufficiently large $n$, it is easy to see that by Bernstein inequality the following event
$$T_0 = \left\{\bar{d}/2 \le \hat{d} \le 2\bar{d}\right\}$$
happens with probability at least $1-n^{-c}$ for some constant $c>0$. By Theorem~\ref{thm:projection-concentration}, the following event 
$$T_1 = \left\{ \norm{(\hat{W}\hat{W}^T-WW^T)Z} \le \frac{2\sqrt{pK\log n}}{d} \right\}$$
happens with probability at least $1-pKn^{-c}$. Under $T_1$, according to \eqref{eq: RUU}, \eqref{eq: RUU decomposition} and \eqref{eq:singular value bound}, we have for every $1\le i\le \min\{p,K\}$:
\begin{eqnarray*}
|\sigma_i - \hat{\sigma}_i|  \le \norm{\hat{\Sigma}-\Sigma} \le \norm{(\hat{W}\hat{W}^T-WW^T)Z} \le \frac{2\sqrt{pK\log n}}{d}.
\end{eqnarray*}
Since $\sigma_r=1$, this implies that under $T_0 \cap T_1$ we have
$$\hat{\sigma}_r \ge 1-\frac{2\sqrt{pK\log n}}{d}\ge 1-\frac{2\sqrt{pK\log n}}{\bar{d}}\ge 1-\frac{4\sqrt{pK\log n}}{\hat{d}}.$$
Similarly, under $T_0 \cap T_1$ we have 
\begin{eqnarray*}
\hat{\sigma}_{r+1} \le \sigma_{r+1}+\frac{2\sqrt{pK\log n}}{d} \le \sigma_{r+1}+ \frac{4\sqrt{pK\log n}}{\hat{d}}. 
\end{eqnarray*}
The right-hand side of this inequality is strictly smaller than the threshold $1-\frac{4\sqrt{pK\log n}}{\hat{d}}$ if 
$$\hat{d}> \frac{8\sqrt{pK\log n}}{1-\sigma_{r+1}},$$
which holds if $T_0$ occurs and condition \eqref{eq:same-scale} is satisfied.
Therefore, 
$$\p(\hat{r} = r) \ge \p(T_0 \cap T_1) \ge 1-2pKn^{-c}$$
and the proof is complete.
\end{proof}

To prove Theorem~\ref{thm:necessary}, we need the following two lemmas to deal with the special case when $P$ is a matrix of a fixed rank.

\begin{lem}[Theorem 1 of \cite{xia2019data}] \label{lem:Xia2019}
Let $A$ be the adjacency matrix of a random graph generated from the inhomogeneous Erd\H{o}s-R\'{e}nyi model and $E = A-\e A$. 
Assume that $\e A$ is of a fixed rank $K$ and $\|E\|<\lambda_K/2$, where $\lambda_K$ is the smallest nonzero eigenvalue of $\e A$. Denote by $W\in \mathbb{R}^{n\times K}$ the matrix whose columns are eigenvectors of $\e A$ corresponding to the $K$ largest eigenvalues of $\e A$; define $\hat{W}\in\mathbb{R}^{n\times K}$ similarly for $A$.  
Then 
$$\hat{W}\hat{W}^T-WW^T = \sum_{k\ge 1}S_{k},$$
where 
$$S_k = \sum_{s}(-1)^{1+\tau(s)}\tilde{P}^{-s_1}E\tilde{P}^{-s_2}\cdots\tilde{P}^{-s_k}E\tilde{P}^{-s_{k+1}},$$
the sum is over all vectors of non-negative integers $s = (s_1,..., s_{k+1})$ such that $\sum_{i=1}^{k+1}s_i=k$, and $\tau(s)$ is the number of positive entries of $s$. Moreover, 
$$\tilde{P}^0 = W_{\perp}W_{\perp}^T, \quad \tilde{P}^{-k} = W\Lambda^{-k}W^T \text{ for } k\ge 1,$$
where $\Lambda\in\mathbb{R}^{K\times K}$ is the diagonal matrix of $K$ non-zero eigenvalues of $\e A$, and $W_\perp\in\mathbb{R}^{n\times(n-K)}$ is the matrix whose columns are eigenvectors of $\e A$ associated with the eigenvalue zero.
\end{lem}

\begin{lem}[Theorem 6.3.2 of \cite{vershynin2018high}]\label{lem:VectorNorm} Let $R\in\mathbb{R}^{m \times n}$ be a non-random matrix and $X = (X_1,...,X_n)^T\in\mathbb{R}^n$ be a random vector with independent, mean zero,
unit variance, sub-gaussian coordinates. Then
$$\norm{\norm{RX} - \norm{R}_F}_{\psi_2} \le C\max_i\norm{X_i}_{\psi_2}^2\norm{R},$$
in which $\norm{\cdot}_{\psi_2}$ is the sub-gaussian Orlicz norm. 
\end{lem}

\begin{proof}[Proof of Theorem~\ref{thm:necessary}]
\textbf{Part (i).}  It is sufficient to find one configuration satisfying the conditions of Theorem~\ref{thm:projection-concentration}, such that the inequality holds. In particular, consider an Erd\H{o}s-Ren\'{y}i model with $K=1$ and  $P = \rho_n \mbone_n\mbone_n^T$, where $\mbone_n$ denotes the all-one vector. In this case, $d = n\rho_n$. Denote ${v} = \frac{1}{\sqrt{n}}\mbone_n$.  Since
$$|{v}^T(\hat{W}\hat{W}^T-WW^T){v}| \le \norm{(\hat{W}\hat{W}^T-WW^T){v}},$$
it is sufficient to show that 
\begin{equation}\label{eq:lower-bound-projection}
|{v}^T(\hat{W}\hat{W}^T-WW^T){v}| \ge c/d.
\end{equation}
By Lemma~\ref{lem:Xia2019}, we have 
\be\label{eq:docompose}
{v}^T(\hat{W}\hat{W}^T-WW^T){v} = {v}^T\pare{S_1(E)+S_2(E)}{v} + {v}^T\sum_{k\ge 3}S_{k}(E){v},
\ee
where $S_k$'s are defined accordingly in Lemma~\ref{lem:Xia2019}. In particular, if we denote $\tilde{P}^{\perp} = \tilde{P}^0$ then
\begin{eqnarray*}
S_1(E) &=& \tilde{P}^{-1}E\tilde{P}^{\perp}+\tilde{P}^{\perp}E\tilde{P}^{-1},\\
S_2(E) &=& (\tilde{P}^{-2}E\tilde{P}^{\perp}E\tilde{P}^{\perp}+\tilde{P}^{\perp}E\tilde{P}^{-2}E\tilde{P}^{\perp}+\tilde{P}^{\perp}E\tilde{P}^{\perp}E\tilde{P}^{-2})\\
&&-(\tilde{P}^{\perp}E\tilde{P}^{-1}E\tilde{P}^{-1}+\tilde{P}^{-1}E\tilde{P}^{\perp}E\tilde{P}^{-1}+\tilde{P}^{-1}E\tilde{P}^{-1}E\tilde{P}^{\perp}), 
\end{eqnarray*}
where  $\tilde{P}^{-1} = \frac{1}{n\rho_n}{v}{v}^T$. Since $\tilde{P}^{\perp}{v} = 0$, we have
$${v}^T\pare{S_1(E)+S_2(E)}{v} = -{v}^T\tilde{P}^{-1}E\tilde{P}^{\perp}E\tilde{P}^{-1}{v} = - \frac{1}{d^2n}\norm{\tilde{P}^{\perp}E\mbone_n}^2.$$
Note that $\tilde{P}^{\perp}= I-(1/n)\mbone_n\mbone_n^T$ is the centralization operator, therefore 
\be\label{eq:term-decomp}
\norm{\tilde{P}^{\perp}E\mbone_n} \ge \norm{E\mbone_n}-\frac{1}{n}\|\mbone_n\mbone_n^TE\mbone_n\| = \norm{E\mbone_n} - \frac{1}{\sqrt{n}}\Big|\sum_{i,j}E_{ij}\Big|.
\ee

We will use Lemma~\ref{lem:VectorNorm} to estimate the first term. In particular, let $\tilde{E} = E/\sqrt{\rho_n(1-\rho_n)}$ and  let $\Psi \in \bR^{(n(n+1)/2)}$ be the vector of on and above-diagonal entries of $\tilde{E}$, indexed by 
$$\Psi= (\tilde{E}_{11}, \tilde{E}_{12},..., \tilde{E}_{1n}, \tilde{E}_{22}, \tilde{E}_{23},..., \tilde{E}_{2n}, \tilde{E}_{33},...,\tilde{E}_{nn})^T.$$
Notice that coordinates of $\Psi$ are independent random variables with mean zero and unit variance. Then $E\mbone_n$ can be written as $\sqrt{\rho_n(1-\rho_n)}R\Psi$ for the matrix $R \in \bR^{n\times n(n+1)/2}$ defined by
\begin{equation*}
R = \begin{pmatrix}
  \mbone_n^T & 0 & 0 & \cdots & 0 \\
e_{2}^{(n)T} & \mbone_{n-1}^T & 0 & \cdots & 0 \\
e_{3}^{(n)T} & e_{2}^{(n-1)T} & \mbone_{n-2}^T & \cdots & 0 \\
  \vdots  & \vdots & \vdots  & \ddots & \vdots  \\
e_{n}^{(n)T} & e_{n-1}^{(n-1)T} & e_{n-2}^{(n-2)T} &\cdots & 1
 \end{pmatrix},
 \end{equation*}
in which $e_{j}^{(m)}\in\mathbb{R}^m$ is the vector with only its $j$th coordinate being 1 while the rest being zeros. Moreover, it is easy to check that (since there are exactly $n^2$ entries in $E$)
$$\norm{R}_F^2  = n^2.$$
To use Lemma~\ref{lem:VectorNorm}, we will find the spectral norm of $R$ first.  This can be achieved by checking for what unit vector ${u} \in \bR^{n(n+1)/2}$ the norm $\norm{R{u}}^2$  achieves its maximum. If we index ${u}$ in the same way as $\tilde{E}$ (in matrix form), we see that
$$\norm{R{u}}^2 = \sum_{i=1}^n\Big(\sum_{j=1}^n{u}_{ij}\Big)^2,$$
where ${u}_{ij} = {u}_{ji}$ for $i > j$. It is easy to see that the maximum is achieved when the mass 1 is equally assigned to all the upper triangular coordinates, as well as the diagonal coordinates but the ratio is two to one such that 
$${u}_{ij} = \frac{2}{\sqrt{(2n-1)n}},~~~~ i< j \text{~~and~~} {u}_{ii} = \frac{1}{\sqrt{(2n-1)n}}.$$
In this case, $$\norm{R{u}}^2 = 2n-1.$$
So the spectral norm of $R$ is $\norm{R} = \sqrt{2n-1}$. Now by using Lemma~\ref{lem:VectorNorm}, we have
\begin{eqnarray*}
\Big\|\norm{E\mbone_n}-\sqrt{\rho_n(1-\rho_n)}n\Big\|_{\psi_2} &=& \left\|\sqrt{\rho_n(1-\rho_n)}\norm{R\Psi}-\sqrt{\rho_n(1-\rho_n)}n\right\|_{\psi_2} \\
&\le& C\sqrt{n\rho_n(1-\rho_n)}\|\tilde{E}_{12}\|_{\psi_2}^2\\
&=& C\left(\frac{n}{\rho_n(1-\rho_n)}\right)^{1/2}\|E_{12}\|_{\psi_2}^2.
\end{eqnarray*}
Since $E_{12}$ is a centered Bernoulli random variable with success probability $\rho_n$, according to \cite{Buldygin2013}, the squared Gaussian norm of $E_{12}$ for $\rho_n= o(1)$ is
\begin{eqnarray*}
\|E_{12}\|_{\psi_2}^2 = \frac{1-2\rho_n}{2(\log(1-\rho_n)-\log\rho_n)}\approx \frac{1}{\log 1/\rho_n}.
\end{eqnarray*}
That means, for some constant $c$,
$$\p\left(|\norm{E\mbone_n}-\sqrt{\rho_n(1-\rho_n)}n| > t\right)\le 2\exp\left(-\frac{ct^2\rho_n\log^2 1/\rho_n}{n}\right).$$
In particular, for $\sqrt{n/\log n} \ll d \le n^{1-\varepsilon} $, with high probability we have
$$\norm{E\mbone_n} > \sqrt{dn}/2.$$
We turn to upper bound the second term of \eqref{eq:term-decomp}. By the Bernstein inequality, 
$$ \frac{1}{\sqrt{n}}\Big|\sum_{i,j}E_{ij}\Big| \le \sqrt{2d{\log n}}$$
with high probability. Therefore for sufficiently large $n$, the following happens with high probability 
\begin{eqnarray}\label{eq:bound1}
\big|{v}^T\pare{S_1(E)+S_2(E)}{v}\big|=\frac{\norm{\tilde{P}^{\perp}E\mbone_n}^2 }{d^2n}\ge \frac{1}{d^2n}\left(\sqrt{dn}/2-\sqrt{2d{\log n}}\right)^2\ge \frac{1}{16d}.
\end{eqnarray}

Next, we look at the term ${v}^T\sum_{k\ge 3}S_{k}(E){v}$. By the definition of $S_k$, we have
\begin{align}\label{eq:bound2}
\Big|{v}^T\sum_{k\ge 3}S_{k}(E){v}\Big| &\le \sum_{k\ge 3}\norm{S_{k}(E)} \le \sum_{k\ge 3}\left(\frac{4\norm{E}}{d}\right)^k \le \frac{C}{d^{3/2}}
\end{align}
with high probability according to Lemma~\ref{lemma:concentration}. Combining \eqref{eq:bound1}, \eqref{eq:bound2} and \eqref{eq:docompose} yields part (i).

\medskip

\textbf{Part (ii).}   Notice that part (i) does not directly indicate the correctness of part (ii). But it motivates the configuration we will use. Specifically, let $p=K=1$ and $P = \rho_n \mbone_n\mbone_n^T$. Notice that in this case, $W ={w}= \frac{1}{\sqrt{n}}\mbone_n$. Also, let ${w}_{\perp} \in \bR^n$ be a unit vector that is orthogonal to ${w}$, the exact value of which will be chosen later. We choose 
$$X = x = \sqrt{n}\pare{\frac{1}{\sqrt{2}}{w} + \frac{1}{\sqrt{2}}{w}_{\perp}}.$$
So $Z = x/\sqrt{n}$. Moreover, set $\alpha = \sqrt{n}{w}$. Denote by $\hat{{w}}$ the leading eigenvector of $A$.

It is easy to see that $r=0$ and $s=1$. Since  $p=K=1$, we have $\hat{Z} = \tilde{Z} = Z =: z$, $\tilde{W} = W = {w}$ and $\breve{W} = \hat{W} = \hat{{w}}.$  In particular, a simple calculation shows that $\hat{\mathcal{P}}_\mathcal{C}$ in \eqref{eq:PChat} takes the form
\begin{equation}\label{eq:PChat special case}
\hat{\mathcal{P}}_\mathcal{C} = \frac{1}{1-\hat{\sigma}_1^2} zz^T - \frac{\hat{\sigma}_1}{1-\hat{\sigma}_1^2} z\hat{{w}}^T,
\end{equation}
where $\hat{\sigma}_1 = \hat{{w}}^Tz$. It then follows from \eqref{eq:alpha beta theta hat} that
\begin{eqnarray*}
\e\hat{\beta} &=& (x^Tx)^{-1}x^T\hat{\mathcal{P}}_{\ccal}(x\beta+\alpha)\\
&=& z^T\left(\frac{1}{1-\hat{\sigma}_1^2} zz^T - \frac{\hat{\sigma}_1}{1-\hat{\sigma}_1^2} z\hat{{w}}^T\right)(z\beta+ {w})\\
&=& \left(\frac{1}{1-\hat{\sigma}_1^2} z^T - \frac{\hat{\sigma}_1}{1-\hat{\sigma}_1^2} \hat{{w}}^T\right)(z\beta+ {w})\\
&=& \left(\frac{1}{1-\hat{\sigma}_1^2} - \frac{\hat{\sigma}_1}{1-\hat{\sigma}_1^2} \hat{{w}}^Tz\right)\beta + \frac{1}{\sqrt{2}(1-\hat{\sigma}_1^2)} - \frac{\hat{\sigma}_1}{1-\hat{\sigma}_1^2}\hat{{w}}^T{w}\\
&=& \beta + \frac{1}{\sqrt{2}(1-\hat{\sigma}_1^2)} - \frac{\hat{\sigma}_1}{1-\hat{\sigma}_1^2}\hat{{w}}^T{w}.
\end{eqnarray*}
This indicates that the bias of $\hat{\beta}$ satisfies 
\begin{equation*}
|B(\hat{\beta})| = \frac{1}{1-\hat{\sigma}_1^2}\Big| \hat{\sigma}_1\hat{{w}}^T{w} - \frac{1}{\sqrt{2}}\Big| \ge \Big| \hat{\sigma}_1\hat{{w}}^T{w} - \frac{1}{\sqrt{2}}\Big|.
\end{equation*}
To show that this is further lower bounded by $c/d$, notice that 
\begin{eqnarray*}
\Big| \hat{\sigma}_1\hat{{w}}^T{w} - \frac{1}{\sqrt{2}}\Big| &=& \Big|z^T\hat{{w}}\hat{{w}}^T{w} - \frac{1}{\sqrt{2}}\Big|\\
&=&\Big| \frac{1}{\sqrt{2}}{w}^T\hat{{w}}\hat{{w}}^T{w}  - \frac{1}{\sqrt{2}} +  \frac{1}{\sqrt{2}}{w}_{\perp}^T\hat{{w}}\hat{{w}}^T{w} \Big| \\
&=&\Big| \frac{1}{\sqrt{2}}{w}^T\hat{{w}}\hat{{w}}^T{w}  - \frac{1}{\sqrt{2}}{w}^T{w}{w}^T{w} +  \frac{1}{\sqrt{2}}{w}_{\perp}^T\hat{{w}}\hat{{w}}^T{w} \Big| \\
&\ge& \frac{1}{\sqrt{2}} \Big| {w}^T\pare{\hat{{w}}\hat{{w}}^T-{w}{w}^T}{w} \Big| - \frac{1}{\sqrt{2}}\Big|{w}_{\perp}^T\hat{{w}}\hat{{w}}^T{w}\Big|\\
&\ge& \frac{1}{\sqrt{2}} \Big| {w}^T\pare{\hat{{w}}\hat{{w}}^T-{w}{w}^T}{w} \Big| - \frac{1}{\sqrt{2}}\Big|{w}_{\perp}^T\hat{{w}}\Big|.
\end{eqnarray*}
By part (i), with high probability, we already have 
$$\frac{1}{\sqrt{2}}\Big| {w}^T\pare{\hat{{w}}\hat{{w}}^T-{w}{w}^T}{w} \Big| \ge \frac{c}{\sqrt{2}d}.$$
Now we show that the second term of the lower bound above is of order $o(1/d)$. 
For that purpose, we choose a special configuration ${w}_{\perp} = e_1/\sqrt{2} - e_2/\sqrt{2}$, where $e_1$ and $e_2$ are unit vectors such that the first and second coordinates are one, respectively, and all other coordinates are zero. In this case, we have
$$|{w}_{\perp}^T\hat{{w}}| = \frac{1}{\sqrt{2}}\Big|\hat{{w}}_1  - \hat{{w}}_2\Big| \le \frac{1}{\sqrt{2}}\pare{|\hat{{w}}_1 - {w}_1| + |\hat{{w}}_2 - {w}_2|} \le \sqrt{2}\norm{\hat{{w}} - {w}}_{\infty}.$$
Notice that ${w}  = \mbone/\sqrt{n}$ so it is perfectly incoherent. Therefore, by using Corollary 3.6 of \cite{lei2019unified}, or the similarly result of \cite{Abbe2017entrywise}, we have
$$|{w}_{\perp}^T\hat{{w}}| \le \sqrt{2}\norm{\hat{{w}} - {w}}_{\infty} \le \sqrt{\frac{\log{n}}{n^2\rho_n}} = \sqrt{\frac{\log{n}}{dn}}$$
with high probability. Consequently, when $d \ll n/\log{n}$, we have $|{w}_{\perp}^T\hat{{w}}| = o(1/d)$ and the lower bound on the bias of $\hat{\beta}$ follows. 

Finally, assume that the variance of the noise in the regression model \eqref{eq:RNC} is $\sigma^2 = 1$. From \eqref{eq:alpha beta theta hat} and \eqref{eq:PChat special case} we have 
\begin{eqnarray*}
\Var(\hat{\beta}) &=& (x^Tx)^{-1}x^T\hat{\mathcal{P}}_\mathcal{C}\hat{\mathcal{P}}_\mathcal{C}^Tx(x^Tx)^{-1}\\
&=& \frac{1}{n}z^T\left(\frac{1}{1-\hat{\sigma}_1^2} zz^T - \frac{\hat{\sigma}_1}{1-\hat{\sigma}_1^2} z\hat{{w}}^T\right)\left(\frac{1}{1-\hat{\sigma}_1^2} zz^T - \frac{\hat{\sigma}_1}{1-\hat{\sigma}_1^2} \hat{{w}}z^T\right)z\\
&=& \frac{1}{n}\left(\frac{1}{1-\hat{\sigma}_1^2} z^T - \frac{\hat{\sigma}_1}{1-\hat{\sigma}_1^2} \hat{{w}}^T\right)\left(\frac{1}{1-\hat{\sigma}_1^2} z - \frac{\hat{\sigma}_1}{1-\hat{\sigma}_1^2} \hat{{w}}\right)\\
&=& \frac{1}{n(1-\hat{\sigma}^2)^2}\norm{z - \hat{\sigma}_1\hat{{w}}}^2 \\
&\le& \frac{4}{n(1-\hat{\sigma}^2)^2} \\
&\le& \frac{16}{n}
\end{eqnarray*}
whenever $\hat{\sigma} < 1/\sqrt{2}$. This completes the proof.
\end{proof}

\newpage

\section{Bounding projection perturbation for the Laplacian}

\begin{proof}[Proof of Theorem~\ref{thm:projection-concentration-Laplacian}] The proof of Theorem~\ref{thm:projection-concentration-Laplacian} is very similar to the proof of Theorem~\ref{thm:projection-concentration} in Section~\ref{sec:projection bound adjacency matrix}. We highlight here the main differences. Denote 
$$E'=L-\e L = (D-\e D) - (A-\e A) = (D-\e D) - E.$$
Then by Bernstein's inequality, there exists a constant $C>0$ such that 
$$\|D-\e D\|\le C\sqrt{d\log n}$$ 
with high probability. Together with Lemma~\ref{lemma:concentration}, this implies $\|E'\|\le C\sqrt{d\log n}$. Using Assumption~\ref{ass:eigen-gap laplacian} about eigenvalue gap of $P=\e L$, eigen-projection representation \eqref{eq:eigen-projection} (where $E$ is replaced with $E'$) and arguing as in  the proof of Theorem~\ref{thm:projection-concentration}, we see that similar to \eqref{eq:Sk bound k greater one}, 
\begin{equation*}
\|S_k\| \le \frac{c}{\pi} \left(\frac{4c^2C\log n}{d}\right)^{k/2}, \quad k\ge 2.
\end{equation*}
It remains to bound $S_1$. According to \eqref{eq:Mv bound},  
$$
\|S_1{v}\|^2 = \frac{2}{\pi^2}\sum_{j\le K} \left| {w}_j^TE' Q_j {v} \right|^2
\le \frac{4}{\pi^2}\sum_{j\le K} \left| {w}_j^TE Q_j {v} \right|^2 + \frac{4}{\pi^2}\sum_{j\le K} \left| {w}_j^T(D-\e D) Q_j {v} \right|^2, 
$$
where $Q_j$ are fixed matrices with norms of order $1/d$. The first term on the right-hand side is of order $K\log n/d^2$, as in the proof of Theorem~\ref{thm:projection-concentration}. Regarding the second term, for ${x}={w}_j$ and ${{y}}=Q_j{x}$ with $\|{x}\|=1$ and $\|{y}\|=O(1/d)$, we have 
$${x}^T(D-\e D) {y}  = \sum_{i,i'=1}^n {x}_i{y}_i(A_{ij}-\e A_{ii'})$$ 
Again, by Hoeffding's inequality \cite[Theorem 2.6.3]{vershynin2018high}, we have for some absolute constant $c'>0$:
\begin{align*}
\p\left(|{{x}}^T(D-\e D) {{y}} |>t\right) &\le 2\exp\left(-\frac{c't^2}{\sum_{i<i'}{x}_i^2{y}_i^2}\right).
\end{align*}
Using ${x} = {w}_j$ with $|{x}_i|\le \|W\|_{\infty}$ and $\|{y}\|=O(1/d)$, we have 
$$\sum_{i<i'}{x}_i^2{y}_i^2\le n\sum_{i=1}^n{x}_i^2{y}_i^2\le n\|W\|_\infty^2\cdot \|{y}\|^2 = O(n\|W\|_\infty^2/d^2).$$
Therefore with high probability, $|{{x}}^T(D-\e D) {{y}}|^2 = O(\|W\|_\infty^2n\log n/d^2)$. Putting these inequalities together, we obtain 
$$
\|S_1{v}\|^2 = O\left(\frac{\left(K+\|W\|_\infty^2n\right)\log n}{d}\right).
$$
The proof is complete by using the triangle inequality the bound for $\|S_k\|$ with $k\ge 2$. 
\end{proof}

\begin{proof}[Proof of Corollary~\ref{coro:small-perturbation-L}]
The proof of Corollary~\ref{coro:small-perturbation-L} is very similar to the proof of Corollary~\ref{coro:small-perturbation-A} in Section~\ref{sec:projection bound adjacency matrix}, but with the help of Theorem~\ref{thm:projection-concentration-Laplacian} instead of Theorem~\ref{thm:projection-concentration}. We omit the detail.
\end{proof}

\newpage 

\section{Proofs for small projection perturbation for parametric estimation under block models}\label{app:blockmodel}

First, we introduce the parametric estimators for $P$ under block models. Given the community label, the maximum likelihood estimator (MLE) of $P$ is used. Under the SBM,  the MLE of $B$ can be obtained by 
\begin{equation}\label{eq:SBM-est}
\hat{B}_{k\ell} = \frac{\sum_{g_i = k, g_j = \ell}A_{ij}}{n_{k}n_{\ell}}.
\end{equation}
with the follow-up estimator $\hat{P}_{ij} = \hat{B}_{g_ig_j}$. Under the DCBM, let $d_i$ be the degree of node $i$. The $\nu_i$'s can be estimated by
\begin{equation}\label{eq:DCBM-est1}
\hat{\nu}_i = \frac{n_kd_i}{\sum_{g_j = k}d_j}
\end{equation}
while the $B$ matrix can be estimated by
\begin{equation}\label{eq:DCBM-est2}
\hat{B}_{k\ell} = \frac{1}{n_{k}n_{\ell}}\sum_{g_i = k, g_j = \ell}\frac{A_{ij}}{\hat{\nu}_i\hat{\nu}_j}.
\end{equation}
The resulting $\hat{P}_{ij}$ is given by $\hat{\nu}_i\hat{\nu}_j\hat{B}_{g_ig_j}$. These will be the estimators we use for the parametric version of our method.

\begin{proof}[Proof of Theorem~\ref{thm:blockmodel-small-perturbation}]
{\bf Part (i).} Note that the estimator $\hat{B}_{k\ell}$ is the sample average of independent Bernoulli random variables. Using Bernstein's inequality and $K=o(\sqrt{nd/\log n})$ from \ref{ass:SBM}, with high probability we have
$$|\hat{B}_{k\ell} - B_{k\ell}| = O\Big(\frac{K\sqrt{\kappa_n\log{n}}}{n}\Big)$$
holds for all $1 \le k, \ell \le K$. 
Therefore, we have
$$\norm{\hat{P} - P}_{\infty}  = O\Big(\frac{K\sqrt{\kappa_n\log{n}}}{n}\Big).$$
Notice that under the current assumption, the eigen-gap of $P$ for the $K$ leading eigenvectors is in the order of $d/K$ \citep{lei2014consistency}. Using David-Kahan theorem~\citep{yu2015useful} together with \ref{ass:SBM} and \ref{ass:strong-consistency}, we have
$$\norm{\hat{W}\hat{W}^T - WW^T} = O\Big(\frac{\norm{\hat{P}-P}}{d/K}\Big) = O\Big(\frac{\sqrt{n}\norm{\hat{P}-P}_{\infty}}{d/K}\Big) = O\Big(\frac{K^2\sqrt{\log{n}}}{n\sqrt{d}}\Big).$$
The condition in part (i) then implies that Assumption~\ref{ass:small-perturbation} holds. 

{\bf Part (ii).} Since $d \gg \log{n}$, for sufficiently large $n$, by Bernstein inequality we have
$$\max_i |d_i - \e d_i| \le C\sqrt{d\log{n}}$$
with high probability.
Denote $\bar{d}_k = \sum_{g_j = k}d_j/n_{k}$, then $\hat{\nu}_i = d_i/\bar{d}_{g_i}$ and $\nu_i = \e d_i/\e \bar{d}_{g_i}$. Therefore by the above inequality, Assumption~\ref{ass:DCBM} and the triangle inequality, 
\begin{eqnarray}\label{eq:nu bound}
|\hat{\nu}_i - \nu_i| = \Big|\frac{d_i}{\bar{d}_{g_i}} - \frac{\e d_i}{\e \bar{d}_{g_i}} \Big|
\le \frac{|d_i-\e d_i|}{\e \bar{d}_{g_i}} + \frac{d_i|\bar{d}_{g_i} - \e \bar{d}_{g_i}|}{\bar{d}_{g_i}\e \bar{d}_{g_i}} = O\left(\sqrt{\frac{\log n}{d}}\right).
\end{eqnarray}
Now we bound the error $\hat{B}_{k\ell}-B_{k\ell}$. 
Recall that
$$\hat{B}_{k\ell} = \frac{1}{n_kn_\ell}\sum_{g_i = k, g_j = \ell}\frac{A_{ij}}{\hat{\nu}_i\hat{\nu}_j}.$$
Define 
$$\hat{B}^*_{k\ell} = \frac{1}{n_kn_\ell}\sum_{g_i = k, g_j = \ell}\frac{A_{ij}}{\nu_i\nu_j}.$$
We have $|\hat{B}_{k\ell} - B_{k\ell}| \le |\hat{B}_{k\ell} - \hat{B}_{k\ell}^*|+|\hat{B}_{k\ell}^* - B_{k\ell}|$.
The second term $|\hat{B}_{k\ell}^* - B_{k\ell}|$ is simply a sum of centered Bernoulli random variables and can be bounded by $O_P(\frac{K\sqrt{\kappa_n\log{n}}}{n})$ as before (we use Assumption~\ref{ass:DCBM} to bound $\nu_i$). For the first term, notice that
\begin{eqnarray*}
|\hat{B}_{k\ell} - \hat{B}_{k\ell}^*| &\le& \frac{1}{n_1n_2}\sum_{g_i = k, g_j = \ell}A_{ij}\Big| \frac{1}{\hat{\nu}_i\hat{\nu}_j}-\frac{1}{\nu_i\nu_j}\Big| \\
&\le& \max_{ij}\Big|\frac{1}{\hat{\nu}_i\hat{\nu}_j}-\frac{1}{\nu_i\nu_j}\Big| \cdot\frac{1}{n_kn_\ell}\sum_{g_i = k, g_j = \ell}A_{ij}.
\end{eqnarray*}
By Assumption~\ref{ass:DCBM} and \eqref{eq:nu bound}, we have
\begin{align*}
\Big|\frac{1}{\hat{\nu}_i\hat{\nu}_j}-\frac{1}{\nu_i\nu_j}\Big| &\le \frac{|\nu_j - \hat{\nu}_j|}{\hat{\nu}_i\hat{\nu}_j\nu_j} +  \frac{|\nu_i - \hat{\nu}_i|}{\hat{\nu}_i\hat{\nu}_i\nu_j} = O\big( \max_{i}|\hat{\nu}_i - \nu_i|\big) = O\left(\sqrt{\frac{\log n}{d}}\right).
\end{align*}
Furthermore, using the assumption $K=o(\sqrt{nd/\log n})$ and Bernstein's inequality as before, we see that with high probability we have  
$$\max_{k, \ell}\frac{1}{n_{k}n_{\ell}}\sum_{g_i = k_1, g_j = k_2}A_{ij} = O(\kappa_n).$$
Therefore, $|\hat{B}_{k\ell} - \hat{B}_{k\ell}^*| = O_P(\sqrt{\kappa_n\log{n}/n})$. With the bound on $|\hat{B}_{k\ell}^* - B_{k\ell}|$, we have
$$|\hat{B}_{k\ell} - B_{k\ell}|  = O\left(\sqrt{\frac{\kappa_n\log{n}}{n}}+\frac{K\sqrt{\kappa_n\log{n}}}{n}\right) = O\left(\sqrt{\frac{\kappa_n\log{n}}{n}}\right).$$
%
Finally, for any $i$ and $j$, we have
\begin{align*}
|\hat{B}_{g_ig_j}\hat{\nu}_i\hat{\nu}_j - B_{g_ig_j}\nu_i\nu_j| & \le \hat{\nu}_i\hat{\nu}_j|\hat{B}_{g_ig_j}-B_{g_ig_j}| + B_{g_ig_j}|\hat{\nu}_i\hat{\nu}_j - \nu_i\nu_j| = O\left(\sqrt{\frac{\kappa_n\log{n}}{n}}\right). 
\end{align*}
Similar to the SBM case, we use the Davis-Kahan theorem and the bound on eigen-gap in \cite{lei2014consistency} to achieve the bound
$$\norm{\hat{W}\hat{W}^T - WW^T} = O\Big(\frac{\norm{\hat{P}-P}}{d/K}\Big) = O\Big(\frac{\sqrt{n}\norm{\hat{P}-P}_{\infty}}{d/K}\Big)  = O\Big(\frac{K\sqrt{\log{n}}}{\sqrt{nd}}\Big).$$
The condition in part (ii) then implies that Assumption~\ref{ass:small-perturbation} holds. 
\end{proof}

\newpage

\section{A degenerate case: directing using eigenvectors as covariates with no confounding}\label{sec:degenerate}

In this section, we give detailed discussion on one special case of our model, when $\mathcal{R} = \col(X) \cap S_K(P)$ is degenerate:
$$\mathcal{R} = \{0\}.$$
This is a special case when there is no confounding between the covariates and the network, such that $r=0$. This situation holds in many cases. For example, the school conflict example in the main paper is determined to have $r=0$. However, our consideration of the more general framework for non-degenerate $\mathcal{R}$ is preferable compared to directly assuming the non-confounding special model in many aspects. 
\begin{figure}[H]
\centering
\includegraphics[width=0.8\textwidth]{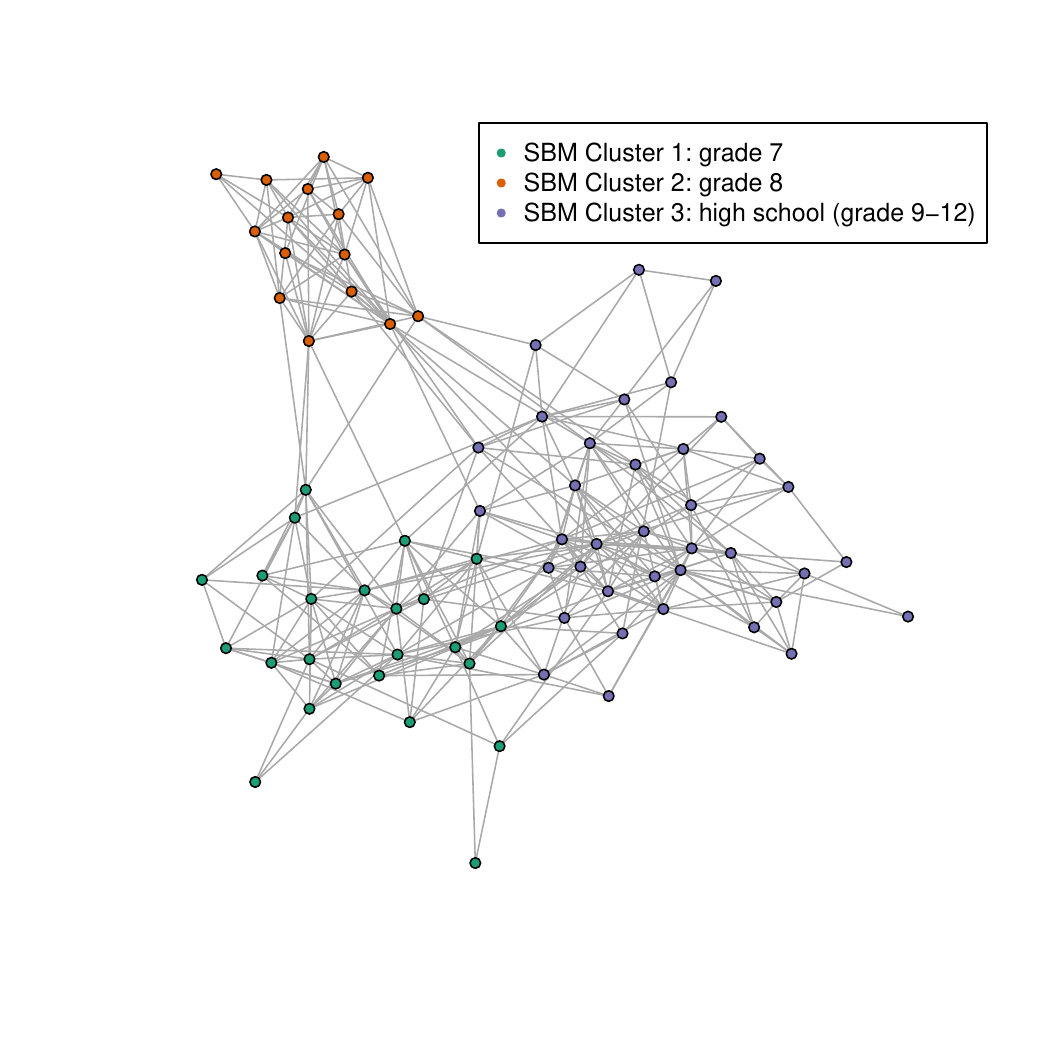}
\caption{The friendship network of one school district in the AddHealth study.}
\label{fig:school70}
\end{figure}
\paragraph{The scope of space confounding in practice.} Empirically, the  nontrivial $\mathcal{R}$ indeed exists in many applications.  As a concrete example, in Figure~\ref{fig:school70}, we show a friendship network of a school district from the AddHealth study \citep{harris2009national}. It happens that this network follows the stochastic block model (SBM) very well. We use three methods to identify the network model \citep{le2015estimating, wang2015likelihood, li2016network} and all of them agree on using the SBM with three communities for the network. The detected community labels perfectly recover the  middle school grade 7, middle school grade 8 and high school groups. Suppose we treat the network as an instantiation of the SBM, the intercept column (of all $1$'s) and the covariate ``school" both lie in the eigenspace, even though they do not perfectly assign with any of the eigenvectors. Generally speaking, we do not want to exclude the intercept or a particular covariate from the model for interpretation purposes. So in this case, directly using eigenvectors as other covariates and assuming a degenerate $\rcal$ would not be feasible. To better demonstrate the scope of nontrivial confounding $\rcal$, we show the inferred dimensionality $r$ for all of the 83 school districts of the Add Health study in Table~\ref{tab:estimated-r}.  19 out of the 83 data sets exhibit certain confounding issues between covariates and the network. This demonstrates that the non-trivial $\rcal$ may exist in a wide range of situations.
\begin{table}
\caption{\label{tab:estimated-r}Inferred dimensionality of the $\mathcal{R}$ for 83 school districts of the Add Health study. And $r>0$ indicates a non-empty intersection space between $X$ and the network eigenspace.}
\centering
\fbox{
\begin{tabular}{r|rrrr}
  \hline
 $r$& 0 & 1 & 2 & 3\\ 
  \hline
frequency &  64 &  13 &   5 & 1 \\ 
   \hline
\end{tabular}}
\end{table}
\paragraph{Conceptual interpretation.} The $r=0$ assumption is essentially assuming the $\col(X)$ and $S_K(P)$ are two orthogonal subspaces. Conceptually, this is an unpleasant constraint. We treat the regression analysis as a procedure to model a response $Y$ by any given set of covariates of interest $X$, and the network connection them. Enforcing the orthogonality means that inherently we could not freely pick what covariates we want to work with. Such a constraint would introduce difficulties for the applicability and interpretation of the model.

\section{Bootstrapping for $r$ selection and out-of-sample prediction}\label{sec:r-tuning}

\paragraph{A bootstrapping to select $r$.} The theory-driven tuning of $r$ introduced in \eqref{eq:r-selection} is based on large-sample properties and may be conservative when $n$ is not sufficiently large. Here, we propose an alternative way of tuning $r$ that provides better performance for small networks, according to our empirical evaluation. The idea is to approximate the actual deviation of principal angles using a bootstrapping procedure.  Specifically, we construct a singular value thresholding estimator $P^*$ as in \cite{chatterjee2015matrix}, but with use of the given rank $K$ instead.  Then  $P^*$ is rescaled so that the average degree of $P^*$ matches the average degree $\hat{d}$ of the original network. Note that the informative eigenvectors of $P^*$ are still $\hat{W}$. After that, we randomly generate $50$ networks $\hat{A}^{(b)}$ with $b=1,...,50$ from $P^*$ and denote their  $K$ leading eigenvectors by $\hat{W}^{(b)}$. Furthermore, the singular values of 
$Z^T\hat{W}^{(b)}$ are computed, denoted by $(\hat{\sigma}^{(b)}_1,...,\hat{\sigma}^{(b)}_K)$ with $b = 1,..., 50$. Finally, the maximum deviance 
$\delta = \max_{k,b} |\hat{\sigma}^{(b)}_k - \hat{\sigma}_k|$ is taken
to construct
$$\hat{r} = \max_{k}\{k: \hat{\sigma}_k > 1-\delta\}.$$
Intuitively, the sampling step of $\hat{A}^{(b)}$ mimics the generating procedure of $A$. Thus, $\delta$ should give an accurate estimate of the possible perturbation on the principal angles.  

\paragraph{Out-of-sample prediction.} For out-of-sample prediction, we would have the full data set $X = (X_{\text{tr}}^T,X_{\text{te}}^T)^T$ and network $A$ (for both training and test data), and the response $Y_{\text{tr}}$ on the training data. A natural way to formulate the model is by assuming the model by Defintion~\ref{defi:model}
for the full data set. Then the subspace basis calculation and intersection calculation can essentially be done in the same way as in our original algorithm. The only difference would be on when calculating the projection coefficients, it is based on $Y_{\text{tr}}$ along the training coordinates. 

\newpage

\section{Additional results of the simulation examples}\label{sec:additional-sim}

Table~\ref{tab:Controlled-SBM-Chisq} summarizes the type-I error probability of the $\chi^2$ test from the setting $\gamma = (0, 0, 0, 0)^T$ in the configuration of Section~\ref{sec:sim}. It shows the desired type-I error control to test $H_0: \alpha = 0$

\begin{table}
\caption{Type-I error probability of the $\chi^2$ test for $\gamma$ when the network is generated from the SBM and $\gamma = 0$. }
\label{tab:Controlled-SBM-Chisq}
\centering
\fbox{\begin{tabular}{l|rrrr}
  \hline
\multirow{2}{*}{Method} & \multirow{2}{*}{$n$} & \multicolumn{3}{c}{average expected degree} \\ 
 &  & $2\log{n}$ & $\sqrt{n}$ & $n^{2/3}$  \\ 
   \hline
 \multirow{3}{*}{SP} &300 & 0.0522 & 0.0524 & 0.0493 \\ 
 & 500 & 0.0493 & 0.0520 & 0.0526 \\ 
 & 1000& 0.0490 & 0.0497 & 0.0510 \\
 & 2000  & 0.0504 & 0.0497 & 0.0506 \\ 
 & 4000  & 0.0496 & 0.0490 & 0.0501 \\ 
\hline
 \multirow{3}{*}{SP-SBM} & 300 & 0.0525 & 0.0509 & 0.0503 \\ 
 & 500 & 0.0508 & 0.0504 & 0.0521 \\ 
  & 1000& 0.0496 & 0.0500 & 0.0511 \\ 
& 2000 & 0.0504 & 0.0495 & 0.0503 \\ 
 & 4000 & 0.0499 & 0.0493 & 0.0500 \\  
   \hline
\end{tabular}}
\end{table}

As a counterpart of the study in Section~\ref{sec:sim}, we generate networks from the DCBM. For the DCBM, the degree parameters are generated from a uniform distribution between 0.2 and 1.  All of the other configurations remain the same, as in the previous setting. Table~\ref{tab:Controlled-DCBM-Ratio} shows the average bias-SD ratio for $\hat{\beta}_2, \hat{\beta}_3$ and $\hat{\beta}_3$ under the DCBM, while Table~\ref{tab:Controlled-DCBM-Coverage} shows the empirical coverage of the 95\% confidence intervals in this setting. The high-level message remains similar to the SBM setting and matches the theory. It can be seen that both versions of SP are valid for large enough $n$ when $\varphi_n = n^{2/3}$, while the parametric version also works for $\varphi_n = 2\sqrt{n}$. However, the advantage of the parametric version is smaller compared with the SBM setting because the DCBM estimation is more complicated.  Lastly, Table~\ref{tab:Controlled-DCBM-Chisq} shows the type I error probability of the $\chi^2$ test, in which the validity is evident.

\begin{table}
\caption{Average theoretical bias-SD ratio for $\beta_2, \beta_3$ and $\beta_4$ when the network is generated from the DCBM.}
\label{tab:Controlled-DCBM-Ratio}
\centering
\fbox{\begin{tabular}{l|rrrr}
  \hline
\multirow{2}{*}{Method} & \multirow{2}{*}{$n$} & \multicolumn{3}{c}{average expected degree} \\ 
 &  & $2\log{n}$ & $\sqrt{n}$ & $n^{2/3}$  \\ 
  \hline
\multirow{5}{*}{SP}
& 300 & 1.456 & 0.877 & 0.259 \\ 
& 500 & 1.452 & 0.740 & 0.300 \\
& 1000 & 1.631 & 0.610 & 0.265 \\ 
 & 2000 & 2.019 & 0.689 & 0.238 \\ 
 & 4000 & 2.763 & 0.647 & 0.176 \\ 
\hline
\multirow{5}{*}{SP-DCBM} 
 & 300 & 1.302 & 0.515 & 0.179 \\ 
 & 500 & 0.811 & 0.401 & 0.140 \\ 
& 1000 & 0.978 & 0.243 & 0.118 \\ 
 & 2000 & 1.130 & 0.240 & 0.088 \\
 & 4000 & 1.252 & 0.188 & 0.100 \\ 
   \hline
\end{tabular}}
\end{table}

\begin{table}
\caption{Average coverage coverage probability of 95\% confidence intervals of $\beta_2, \beta_3, \beta_4$ when the network is generated from the DCBM.}
\label{tab:Controlled-DCBM-Coverage}
\centering
\fbox{\begin{tabular}{l|rrrr}
  \hline
\multirow{2}{*}{Method} & \multirow{2}{*}{$n$} & \multicolumn{3}{c}{average expected degree} \\ 
 &  & $2\log{n}$ & $\sqrt{n}$ & $n^{2/3}$  \\ 
  \hline
\multirow{5}{*}{SP}
& 300 & 0.792 & 0.906 & 0.946 \\ 
& 500 & 0.813 & 0.904 & 0.939 \\ 
& 1000& 0.775 & 0.912 & 0.944 \\
 & 2000 & 0.605 & 0.901 & 0.946 \\ 
 & 4000  & 0.329 & 0.892 & 0.947 \\ 
\hline
\multirow{5}{*}{SP-DCBM} 
& 300 & 0.817 & 0.937 & 0.947 \\ 
& 500 & 0.893 & 0.935 & 0.947 \\
& 1000  & 0.881 & 0.943 & 0.950 \\ 
 & 2000 & 0.848 & 0.945 & 0.948 \\  
 & 4000 & 0.754 & 0.946 & 0.951 \\ 
   \hline
\end{tabular}}
\end{table}

\begin{table}
\caption{Type-I error probability of the $\chi^2$ test for $\gamma$ when the network is generated from the DCBM and $\gamma = 0$.}
\label{tab:Controlled-DCBM-Chisq}
\centering
\fbox{\begin{tabular}{l|r|rrr}
  \hline
 \multirow{2}{*}{Method} & \multirow{2}{*}{$n$} & \multicolumn{3}{c}{average expected degree} \\ 
 & & $2\log{n}$ & $\sqrt{n}$ & $n^{2/3}$  \\ 
   \hline
\multirow{5}{*}{SP} & 300 & 0.0510 & 0.0519 & 0.0531 \\ 
& 500 & 0.0505 & 0.0509 & 0.0503 \\
& 1000& 0.0518 & 0.0514 & 0.0506 \\ 
& 2000  & 0.0515 & 0.0517 & 0.0497 \\ 
 & 4000  & 0.0508 & 0.0504 & 0.0501 \\
 \hline
 \multirow{5}{*}{SP-DCBM} 
 & 300 & 0.0515 & 0.0509 & 0.0535 \\ 
& 500 & 0.0499 & 0.0499 & 0.0518 \\ 
 & 1000 & 0.0512 & 0.0515 & 0.0502 \\ 
 & 2000 & 0.0518 & 0.0510 & 0.0492 \\  
 & 4000 & 0.0498 & 0.0489 & 0.0498 \\  
   \hline
\end{tabular}}
\end{table}

Table~\ref{tab:DCSBM-compare} shows the comparison of our methods with the benchmark methods in the DCBM setting.

\begin{table}
\caption{Mean squared error (MSE) $\times 10^2$ of $\e Y$ when the network is generated from the DCBM with $n=300, p=4, K=4$ with three types of individual effects.  } 
\centering
\fbox{\small
\label{tab:DCSBM-compare_n300}
\begin{tabular}{l|l|rrrrrrr}
  \hline
 indiv. effects &   avg. degree & SP & SP-DCBM & OLS & SIM &RNC & SP-L & SP-DCBM-L\\ 
  \hline
\multirow{4}{*}{eigenspace} & $2\log{n}$ & 11.89 & 7.01 & 39.79 & 37.22 & 12.27 & 39.77 & 39.78 \\ 
& $\sqrt{n}$ & 8.09 & 3.77 & 39.79 & 35.17 & 11.73 & 39.81 & 39.80 \\ 
 & $n^{2/3}$ & 2.46 & 1.00 & 39.79 & 31.86 & 10.11 & 39.71 & 39.65 \\ 
\hline
\multirow{4}{*}{$\gamma = 0$} & $2\log{n}$ & 0.39 & 0.38 & 0.32 & 0.64 & 1.30 & 0.32 & 0.34 \\ 
 & $\sqrt{n}$ & 0.39 & 0.38 & 0.32 & 0.63 & 1.26 & 0.34 & 0.34 \\ 
 & $n^{2/3}$ & 0.41 & 0.37 & 0.32 & 0.64 & 2.49 & 0.37 & 0.38 \\ 
\hline
\multirow{4}{*}{smooth} & $2\log{n}$ & 24.17 & 24.17 & 24.17 & 22.42 & 13.70 & 12.34 & 16.34 \\ 
  & $2\sqrt{n}$   & 24.82 & 24.82 & 24.82 & 23.30 & 15.51 & 12.70 & 16.82 \\ 
  & $n^{2/3}$ & 24.75 & 24.75 & 24.75 & 24.34 & 16.27 & 12.69 & 12.93 \\  
   \hline
\end{tabular}
}
\end{table}

\begin{table}
\caption{Mean squared error (MSE) $\times 10^2$ of $\e Y$ when the network is generated from the DCBM with $n=1000, p=4, K=4$ with three types of individual effects.  } 
\centering
\fbox{\small
\label{tab:DCSBM-compare}
\begin{tabular}{l|l|rrrrrrr}
  \hline
 indiv. effects &   avg. degree & SP & SP-DCBM & OLS & SIM &RNC & SP-L & SP-DCBM-L\\ 
  \hline
\multirow{4}{*}{eigenspace} & $2\log{n}$ & 10.51 & 4.80 & 41.41 & 40.86 & 12.24 & 41.41 & 41.41 \\ 
& $\sqrt{n}$ & 4.51 & 1.31 & 41.41 & 40.43 & 10.86 & 41.41 & 41.41 \\ 
 & $n^{2/3}$ & 1.23 & 0.39 & 41.41 & 37.78 & 9.23 & 41.41 & 41.41 \\ 
\hline
\multirow{4}{*}{$\gamma = 0$} & $2\log{n}$ & 0.14 & 0.14 & 0.12 & 0.21 & 0.35 & 0.12 & 0.14 \\
 & $\sqrt{n}$ & 0.14 & 0.13 & 0.12 & 0.21 & 0.82 & 0.12 & 0.13 \\ 
 & $n^{2/3}$ & 0.13 & 0.14 & 0.12 & 0.23 & 1.19 & 0.13 & 0.13 \\ 
\hline
\multirow{4}{*}{smooth} & $2\log{n}$& 16.80 & 16.80 & 16.75 & 16.56 & 3.05 & 8.47 & 12.10 \\  
  & $2\sqrt{n}$  & 24.82 & 24.82 & 24.82 & 24.37 & 17.35 & 12.51 & 17.01 \\ 
  & $n^{2/3}$ & 24.90 & 24.90 & 24.89 & 24.71 & 11.73 & 12.56 & 12.63 \\ 
   \hline
\end{tabular}
}
\end{table}

\newpage

\section{Additional results about the school conflict example}\label{sec:more_school_conflicts}

In this section, we include some additional results from our analysis of the school conflict study. Figure~\ref{fig:reg-hist-binary} shows distributions of binary variables in the data set. 
\begin{figure}
\centering
\includegraphics[width=0.6\textwidth]{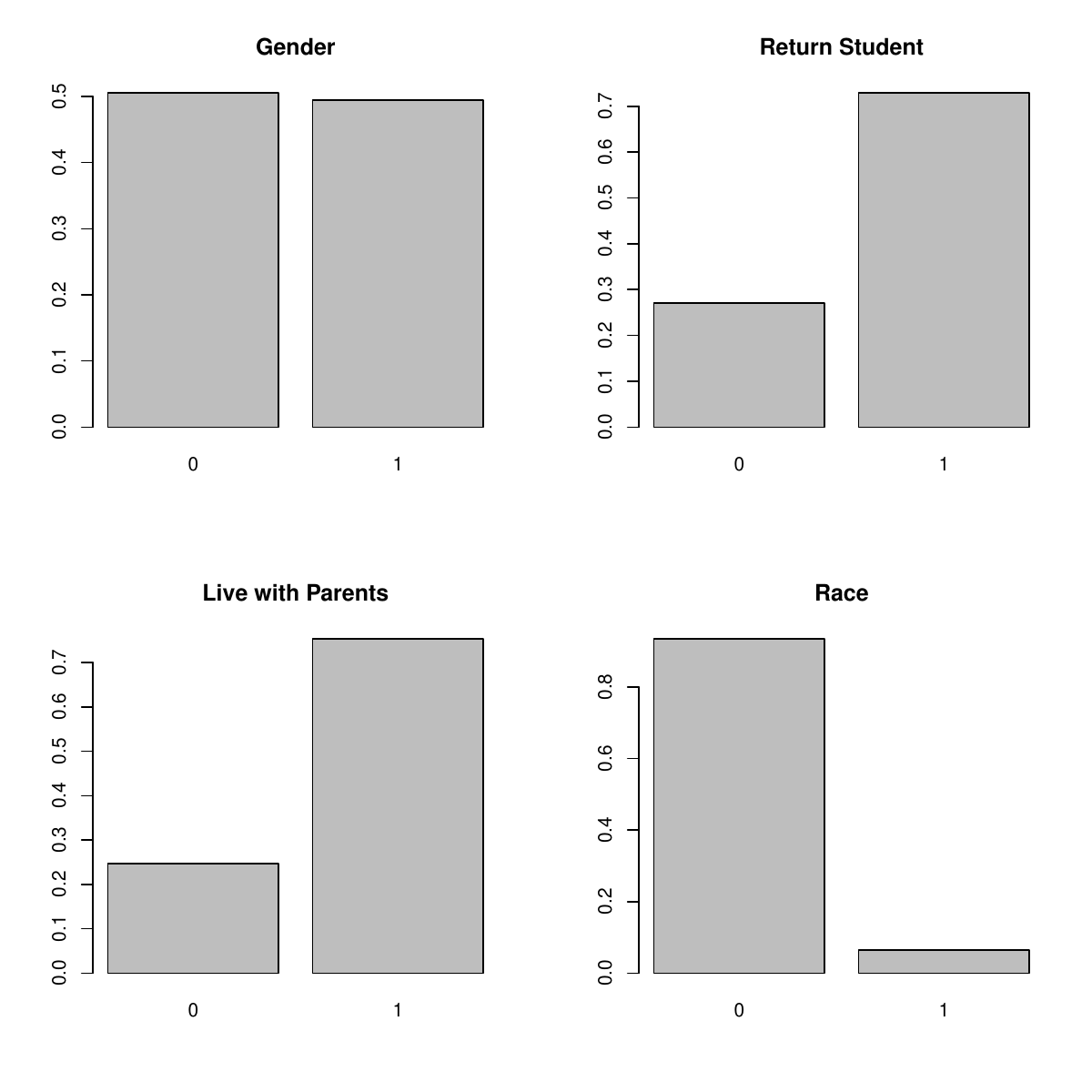}
\caption{Marginal distributions of the four binary variables in the data set. }
\label{fig:reg-hist-binary}
\end{figure}
\begin{figure}
\subfigure[  Full SP model fitting ]{\includegraphics[width = 0.45\textwidth]{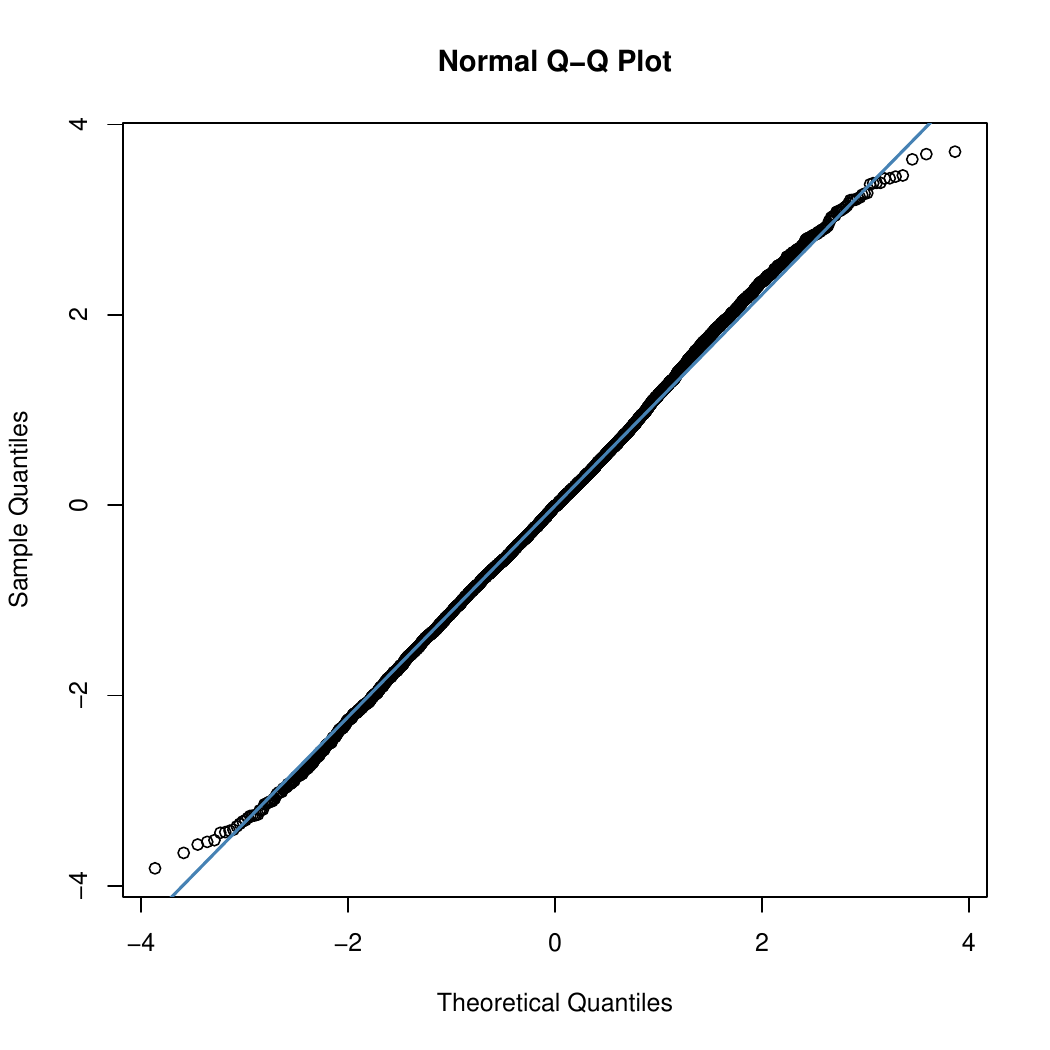}}
\subfigure[  Reduced SP modeling fitting  ]{\includegraphics[width = 0.45\textwidth]{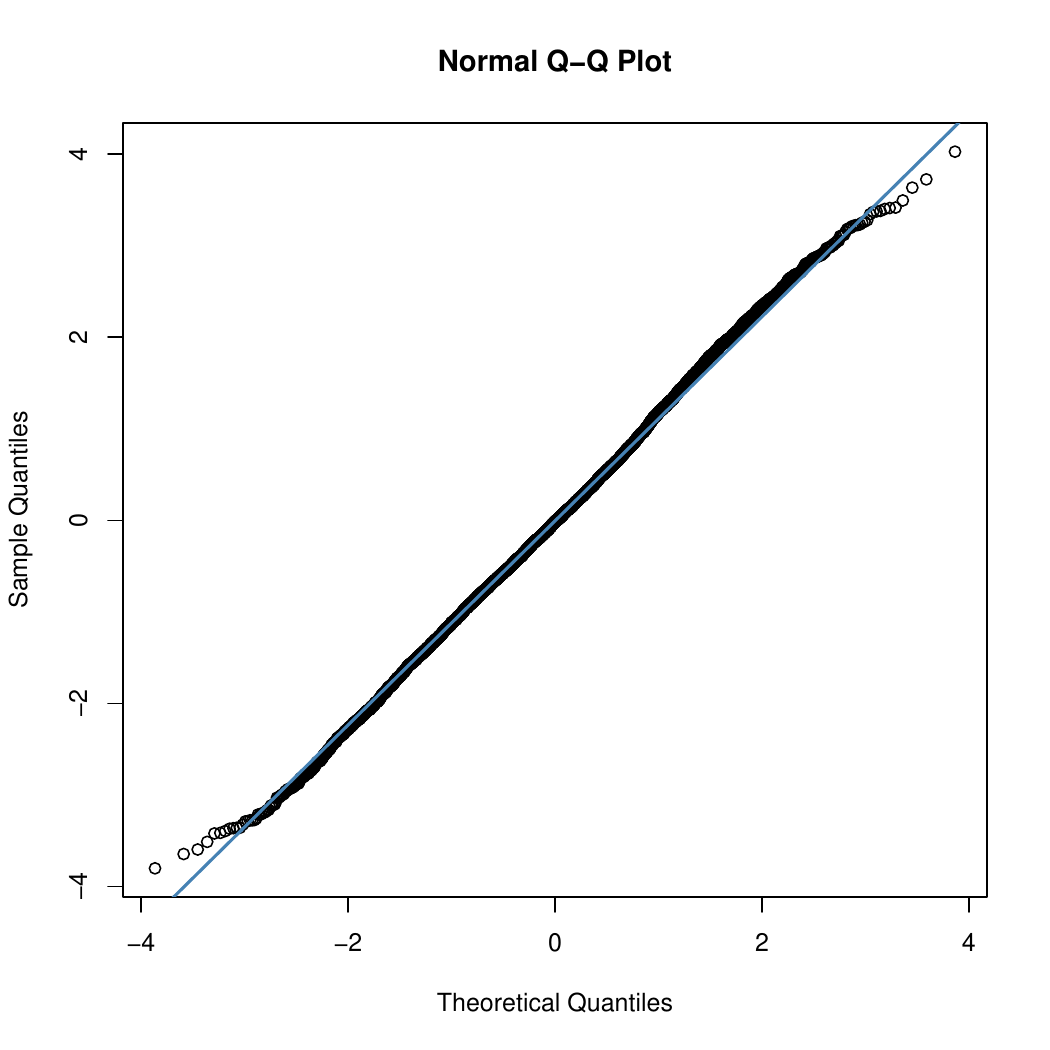}}
\caption{ QQ-plots of residuals in the full SP model fitting and the reduced SP modeling fitting. }
\label{fig:qqplots}
\end{figure}
Figure~\ref{fig:qqplots} includes the qq-plots of modeling fitting residuals of the full model and the reduced model in Section~\ref{sec:conflict}. The residuals follow the straight line closely with only slightly lighter tails than a normal distribution. This is reasonable as our response variable is a result of averaged discrete rating of the survey data. But the deviation from Gaussian is small so our model might serve as a good approximation for inference.
\begin{table}
\caption{Model fitting coefficients of the RNC method on the three networks}
\centering
\fbox{
\label{tab:RNC-results-3Networks}
\begin{tabular}{rrrr}
  \hline
 & Wave I+II & Wave II & Wave I \\ 
  \hline
Race: white & 0.2131 & 0.2078 & 0.2192 \\ 
Race: black & 0.0194 & 0.0091 & 0.0385 \\ 
 Race: hispanic  & 0.1102 & 0.1147 & 0.1131 \\  
 Race: asian & 0.1636 & 0.1625 & 0.1618 \\ 
 Grade & 0.3672 & 0.3921 & 0.3675 \\ 
  Gender & 0.1207 & 0.1108 & 0.1440 \\ 
  Return student & -0.1379 & -0.1360 & -0.1372 \\ 
  Live w/ both parents& 0.1710 & 0.1746 & 0.1713 \\ 
  Treatment & -0.0768 & -0.0887 & -0.0625 \\  
   \hline
\end{tabular}}
\end{table}

Table~\ref{tab:RNC-results-3Networks} includes the fitted coefficients of the covariates by the RNC method on the three different networks. Recall that the RNC method could not include intercepts or school effects, and it does not have an inference framework. So the model fitting results are not directly comparable with the other three methods in the paper and we could not precisely evaluate how statistically significant the network perturbations' impact is on the model fitting. Numerically, the model fitting results are reasonably stable for most of the variables, except for gender and the treatment. However, as discussed in the main paper, we do expect that the treatment effect estimate remains stable across different settings, so the large variation on treatment in Table~\ref{tab:RNC-results-3Networks} may be concerning for the RNC fit.

Figure~\ref{fig:school-other} reveals the principle angle cosine values measuring the subspace alignment between the Wave I network, Wave II network, and Wave I+II combined network for 25 schools in the study, excluding School 40 that has been demonstrated in Figure~\ref{fig:school40}.
\begin{figure}[h]
\centering
\includegraphics[width=1\textwidth]{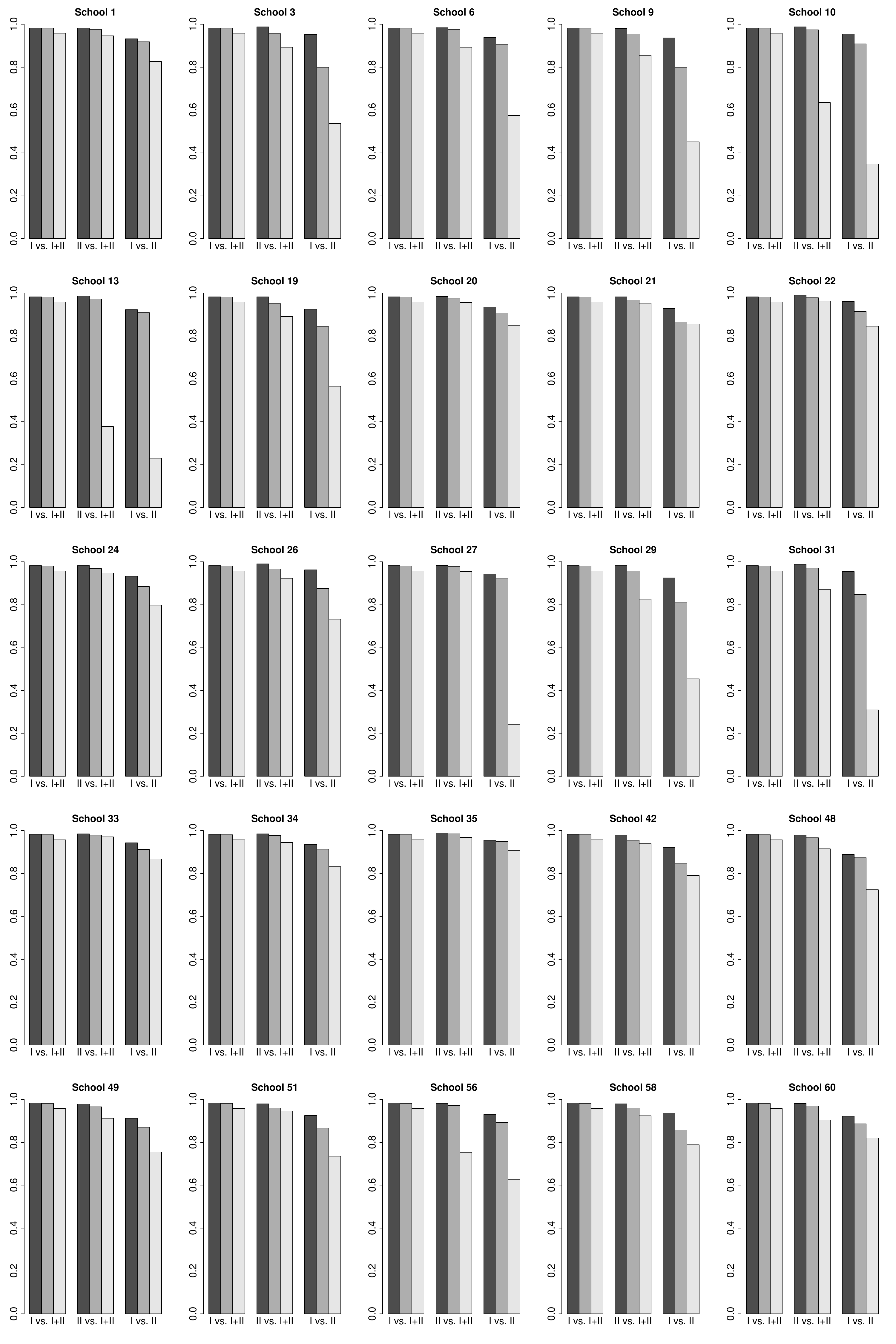}
\caption{Principal angle cosine values between different constructed networks of other 25 schools in the study.}
\label{fig:school-other}
\end{figure}

\end{appendix}

\end{document}